\pdfoutput=1
\PassOptionsToPackage{hyperfootnotes=false}{hyperref}
\documentclass[11pt,en,cite=authoryear,bibstyle=abbrvnat]{note}

\usepackage[utf8]{inputenc}
\usepackage{mathrsfs,bm,bbm,dsfont,mathtools}
\usepackage{algorithm,algpseudocode}
\usepackage{overpic,pifont,array,makecell}
\usepackage{setspace}
\floatevery{algorithm}{\small\setstretch{1.05}}

\counterwithout{theorem}{section}
\counterwithout{definition}{section}
\counterwithout{example}{section}

\theoremstyle{plain}
\newtheorem{corollary}[theorem]{Corollary}

\theoremstyle{remark}
\newtheorem{remark}{Remark}
\theoremstyle{definition}
\newtheorem{assumption}{Assumption}

\providecommand{\coloneq}{\coloneqq}
\makeatletter
\providecommand{\widecheck}{}
\renewcommand{\widecheck}[1]{%
  \mathchoice
    {\widecheckaccent{#1}{1.7}{0.15}}%
    {\widecheckaccent{#1}{1.7}{0.15}}%
    {\widecheckaccent{#1}{1.7}{0.15}}%
    {\widecheckaccent{#1}{1.35}{0.15}}%
}
\newcommand{\widecheckaccent}[3]{%
  \overset{\vbox to #3ex{\hbox{\scalebox{#2}[0.75]{$\m@th\check{}$}}\vss}}{#1}%
}
\makeatother
\providecommand{\widebar}[1]{\overline{#1}}

\newcommand\syr[1]{#1}

\newcommand\syb[1]{#1}

\newcommand{\bel}{\begin{eqnarray}\label}
\newcommand{\eel}{\end{eqnarray}}
\newcommand{\bes}{\begin{eqnarray*}}
\newcommand{\ees}{\end{eqnarray*}}
\newcommand{\bei}{\begin{itemize}}
\newcommand{\beiftnt}{\begin{itemize}\footnotesize}
\newcommand{\eei}{\end{itemize}}

\def\benu{\begin{enumerate}}
\def\eenu{\end{enumerate}}

\def\argmin{\mathop{\rm arg\, min}}

\def\complex{\mathop{{\rm I}\kern-.58em\hbox{\rm C}}\nolimits}

\def\rank{\hbox{\rm rank}}

\def\supp{\hbox{\rm supp}}

\def\mathbold{\boldsymbol} %\def\mathbold{\mathbf}
\def\ba{\mathbold{a}}

\def\bA{\mathbold{A}}\def\calA{\mathcal{A}}

\def\bb{\mathbold{b}}

\def\bF{\mathbold{F}}

\def\bI{\mathbold{I}}\def\hbI{{\widehat{\bI}}}

\def\calL{\mathcal{L}}

\def\bm{\mathbold{m}}

\def\bM{\mathbold{M}}
\def\calM{\mathcal{M}}

\def\calN{\mathcal{N}}

\def\bo{\mathbold{o}}
\def\hbo{{\widehat{\bo}}}
\def\bO{\mathbold{O}}\def\tbO{{\widetilde{\bO}}}
\def\calO{\mathcal{O}}

\def\br{\mathbold{r}}

\def\bR{\mathbold{R}}

\def\calT{\mathcal{T}}

\def\bU{\mathbold{U}}

\def\bv{\mathbold{v}}

\def\bw{\mathbold{w}}

\def\bW{\mathbold{W}}

\def\bx{\mathbold{x}}
\def\hbx{{\widehat{\bx}}}\def\tbx{{\widetilde{\bx}}}
\def\bX{\mathbold{X}}\def\hbX{{\widehat{\bX}}}\def\tbX{{\widetilde{\bX}}}
\def\calX{\mathcal{X}}

\def\bY{\mathbold{Y}}\def\hbY{{\widehat{\bY}}}

\def\bz{\mathbold{z}}
\def\hbz{{\widehat{\bz}}}
\def\bZ{\mathbold{Z}}\def\hbZ{{\widehat{\bZ}}}

\def\btheta{\mathbold{\theta}}

\def\hbtheta{{\widehat{\btheta}}}\def\tbtheta{{\widetilde{\btheta}}}

\def\bnu{\mathbold{\nu}}

\def\bSigma{\mathbold{\Sigma}}

\def\bphi{\mathbold{\phi}}

\def\hbphi{{\widehat{\bphi}}}\def\tbphi{{\widetilde{\bphi}}}

\def\bPhi{\mathbold{\Phi}}
\def\tbPhi{{\widetilde{\bPhi}}}

\def\bbR{{\mathbb{R}}}

\def\bbE{{\mathbb{E}}}
\def\bbP{{\mathbb{P}}}

\def\argmin{\operatorname{argmin} \displaylimits}
\def\bm{\boldsymbol{\mu}}

\def\bF{\mathbold{F}}

\def\bbR{{\mathbb{R}}}

\def\bbE{{\mathbb{E}}}
\def\bbP{{\mathbb{P}}}

\newcommand\dd{\mathrm{d}}

 \newcommand{\deq}{\ \overset{\text{d}}{=} \ }
 
 \newcommand{\bzero}{\mathbf{0}}
 \newcommand{\bone}{\mathds{1}}
 \newcommand{\bfone}{\mathbf{1}}
 \newcommand{\wbX}{\widebar{\bX}}
 \newcommand{\wbx}{\widebar{\bx}}
 
 \newcommand{\hwbx}{\widehat{\wbx}}
 
 \newcommand{\deta}{\mathbold{\eta}}
 \newcommand{\tdeta}{\widetilde{\mathbold{\eta}}}
 
 \newcommand{\deps}{\mathbold{\varepsilon}}
 
 \newcommand{\cdeps}{\widecheck{\mathbold{\varepsilon}}}
\newcommand{\tdxi}{\widetilde{\mathbold{\xi}}}

 \newcommand\independent{\perp\!\!\!\perp}
 \newcommand{\blp}{\bigl (}
 \newcommand{\brp}{\bigr )}
 \definecolor{pinegreen}{rgb}{0.0, 0.47, 0.44}
\definecolor{MyYellow}{HTML}{fea444}
\definecolor{MyRed}{HTML}{eb5f60}
\definecolor{MyGreen}{HTML}{3ed1ee}
 \newcommand{\out}{f}
 \newcommand{\bout}{\widebar{\out}}
 
 \newcommand{\trans}{g}
 \newcommand{\btrans}{\widebar{\trans}}
 
 \newcommand{\enc}{h}
 \newcommand{\dec}{u}
 \newcommand{\backt}{q}
 \newcommand{\law}{\mathsf{Law}}
\newcommand{\na}{\mathsf{na}}

\renewcommand{\wbX}{\overline{\bX}}
\renewcommand{\wbx}{\overline{\bx}}

\renewcommand{\hwbx}{\widehat{\wbx}}
\newcommand{\Col}{\operatorname{Col}}

\title{Latent-driven Imputation and Forecasting for partially observed Time Series (LIFTS)}
\author{Sanyou Wu\textsuperscript{1} \quad Ying Wei\textsuperscript{1} \quad Kaizheng Wang\textsuperscript{2}\thanks{Address for correspondence: Kaizheng Wang, Columbia University, \email{kw2934@columbia.edu}.}}
\institute{\small
\textsuperscript{1}Department of Biostatistics, Mailman School of Public Health,\\
Columbia University, New York, NY 10032, USA\\[0.5ex]
\textsuperscript{2}Department of IEOR and Data Science Institute,\\
Columbia University, New York, NY 10027, USA}
\date{}
\hypersetup{
  pdftitle={Latent-driven Imputation and Forecasting for partially observed Time Series (LIFTS)},
  pdfauthor={Sanyou Wu, Ying Wei, Kaizheng Wang}
}

\begin{document}
\maketitle
\begin{abstract}
Partially observed multivariate time series pose two coupled challenges: recovering missing measurements and forecasting future trajectories when observation patterns are informative. Existing methods often assume ignorable missingness or use the mask only as an auxiliary input. We propose LIFTS (Latent-driven Imputation and Forecasting for partly observed Time Series), a probabilistic framework in which a latent process jointly governs the complete measurement trajectory and observation mask, allowing observation patterns to inform latent-state inference for both forecasting and imputation. LIFTS combines flexible neural-network parameterizations with a forward-backward architecture: the forward pass performs filtering and autoregressive forecasting, while the backward pass performs smoothing and conditional multiple imputation. The model and inference components are trained jointly using a masked energy score and self-masking. We establish forward-backward distributional representations for the general framework and, for an explicit structural subclass, provide full-law identification and consistency guarantees. Simulations and an application to PhysioNet 2019 show that LIFTS provides accurate point predictions, improved distributional accuracy, and well-calibrated predictive intervals; it outperforms competing methods particularly when training data are limited and has substantially faster inference than the diffusion-based baseline.
\keywords{Partially observed time series, State-space models,
Informative missingness, Imputation, Forecasting, Energy scores}
\end{abstract}

\section{Introduction}
\noindent
Partially observed multivariate time series arise in a wide range of
scientific and operational settings, including physiological monitoring,
wearable sensing, environmental surveillance, financial systems, and brain
imaging analysis
\citep{reyna2020early,fortuin2020gp,tashiro2021csdi,li2024time}. Two inferential tasks are particularly important in such applications.
One is \textit{retrospective imputation}, which aims to recover measurements that were not recorded, and the other is \textit{forecasting} that characterizes the distribution of future trajectories given the information currently available. Although these tasks are often studied
separately, they are closely connected: both require learning the temporal and
cross-variable dependence of the underlying process, together with the
uncertainty induced by partial observation. This common structure motivates a unified probabilistic treatment of imputation and forecasting.

Partial observation complicates these tasks, especially when the observation process is informative. Measurements are often not collected according to a fixed design, but their timing and availability may depend on external constraints, operational decisions, resource availability, or the evolving state of the system. The observation process becomes informative when these factors are associated with the underlying trajectory. This phenomenon is particularly pronounced in clinical settings, where the timing and frequency of laboratory tests, vital signs, and other measurements may depend on the patient's underlying physiological state \citep{wells2013strategies,weiskopf2013sick,sisk2021informative}. Consequently, the missing mask is more than an indicator of unavailable entries: it is an observed stochastic process that may contain information about the underlying trajectory. Ignoring this information, or failing to account for its dependence on the underlying process, can mischaracterize the conditional distributions targeted by imputation and forecasting, resulting in biased predictions, miscalibrated uncertainty, and misleading downstream analyses \citep{little2019statistical,mitra2023learning}.

Existing methods for partially observed time series have largely addressed this challenge by incorporating the observation mask into flexible predictive models. Sequence-based methods, such as
\citep{che2018recurrent,cao2018brits,du2023saits}, use recurrent or attention-based architectures to learn temporal and cross-variable structure from observed values and
their mask 
, while generative methods
produce distributional imputations by sampling 
multiple completions conditional on observed data and masks
\citep{fortuin2020gp,tashiro2021csdi}. 
Although including the mask can improve prediction, it does not specify how observation patterns are generated or how they relate to unobserved measurements. Conditioning on the mask alone may therefore be insufficient to account for informative observation when estimating the distributions needed for imputation and forecasting. A separate literature on nonignorable missing data jointly models measurements and their observation indicators, but these methods have been developed primarily for non-temporal multivariate data \citep{ipsen2020not,ma2021identifiable,xie2026identifiable,liu2026missingness,yang2026emputation}. A methodological gap therefore remains for a flexible dynamic framework that jointly models complete measurements and observation patterns and provides a coherent basis for probabilistic imputation and forecasting under informative observation.

To address this gap, we adopt a latent-driven approach, in which an evolving latent state jointly governs the complete measurement trajectory and the missingness mask, as illustrated in Figure~\ref{fig:motivating_example}. The mask is therefore not merely a bookkeeping device, but an additional manifestation of the underlying state that can inform inference and downstream analysis. Marginally, the shared dependence on the latent state generally induces a missing-not-at-random mechanism, because observation probabilities may depend on unobserved measurements through that state. By linking the observation process to a shared latent state rather than directly to a high-dimensional complete measurement history, the formulation provides a parsimonious and operational way to accommodate complex, time-varying observation patterns without separately specifying a mechanism for each pattern. It represents a structured class of informative observation mechanisms rather than all possible missing-not-at-random mechanisms \citep{rubin1976inference,mealli2015clarifying}.
\begin{figure}[t]
    \centering
    \includegraphics[width=0.8\linewidth]{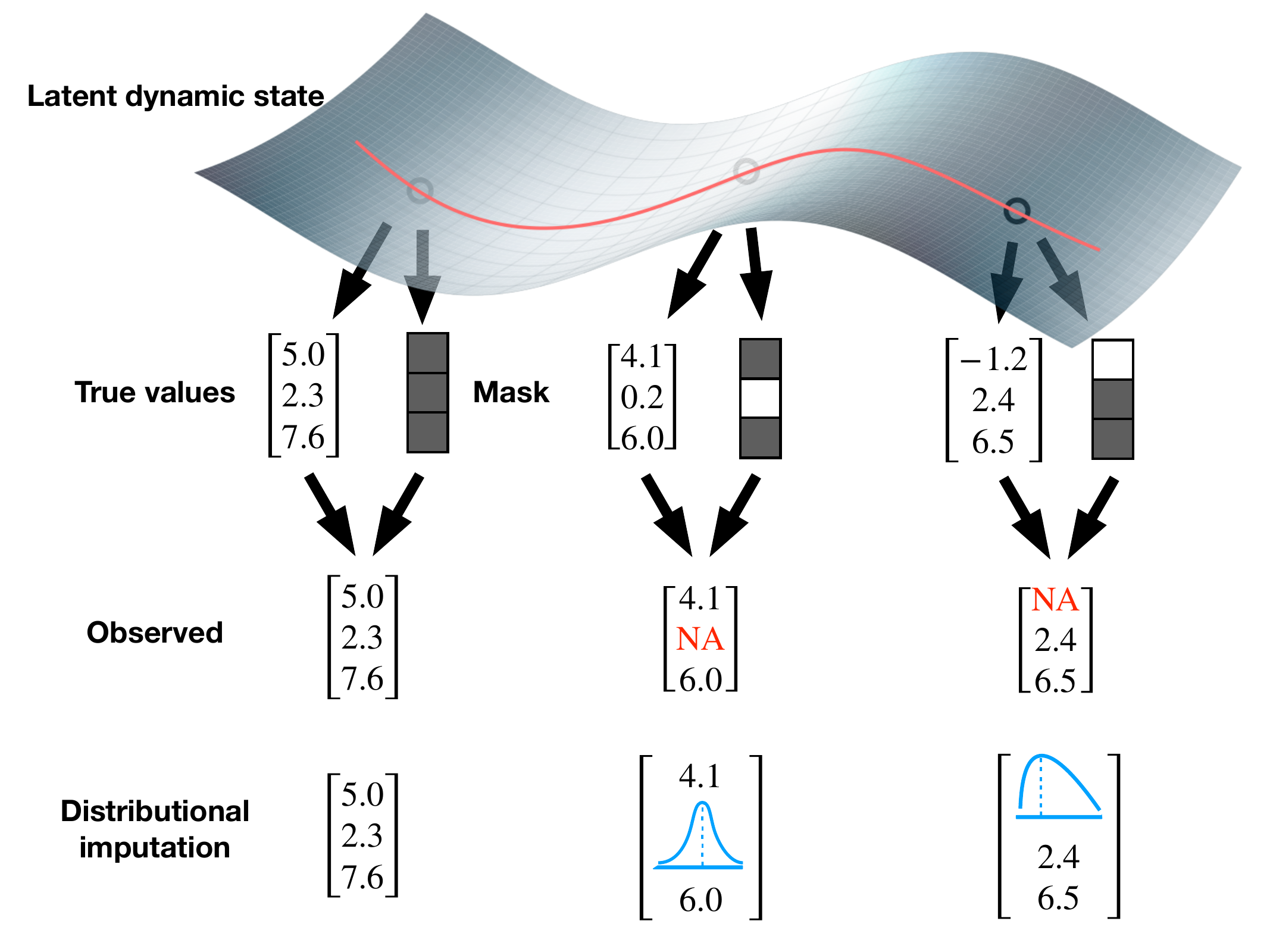}
\caption{Illustration of the latent-driven view of a partially observed time series.
An evolving latent state jointly governs the complete measurement vector and
its observation mask.}
    \label{fig:motivating_example}
\end{figure}

Building on this formulation, we propose LIFTS
(\emph{Latent-driven Imputation and Forecasting for partially observed Time Series}), a probabilistic framework based on a nonlinear state-space model, where a latent transition governs temporal evolution and two output mechanisms generate the complete measurements and their observation masks.

{LIFTS complements this generative model with learned forward-backward inference. In the forward pass, an encoder represents the latent filtering distribution given the history of
recorded measurements and masks, providing the starting point for forecasting through the generative model. In the backward pass,
a smoother incorporates future information to generate smoothed latent trajectories. An imputation decoder then samples missing
measurements conditional on these states and the observed measurements. The generative and inference components are parameterized by
neural networks and trained jointly using sample-based energy-score objectives \citep{gneiting2007strictly,szekely2013energy}, with
randomly withheld observed measurements providing training targets for imputation.

Our theory clarifies why this inference structure is appropriate and when learning from partial observations can recover the full trajectory
law. Under the general model, we show that suitable conditional maps reproduce the target forecasting and imputation laws. In particular,
the backward recursion recovers the joint smoothing distribution, preserving temporal dependence rather than only reproducing time-wise
marginals. We also establish strict propriety of the masked energy score, ensuring that its population objective targets the joint law
of recorded measurements and masks, including their dependence. However, learning this observed-data law does not by itself identify
the distribution of unobserved measurements. For an explicit structural subclass, we establish full-law identification and consistency of the
corresponding model-based estimator under the stated assumptions. These results provide sufficient conditions under which learning the
observation law can recover the joint distribution of complete measurement trajectories and masks. They complement identification
results for generative missing-data models
\citep{ma2021identifiable,xie2026identifiable} and identification-guided energy-score imputation \citep{yang2026emputation}, with the focus here
on joint trajectory laws under common latent dynamics and their connection to forward-backward inference.

Simulations and an analysis of PhysioNet 2019 show competitive point accuracy, improved distributional predictions, and predictive intervals with coverage close to nominal levels. The distributional gains are particularly pronounced when training data are limited, and LIFTS provides substantially faster inference than the diffusion-based benchmark CSDI.
}

The remainder of the paper is organized as follows. Section \ref{sec:set-up}
introduces the problem formulation, including imputation, forecasting, and our model. Sections \ref{sec:methodology} and \ref{sec:estimation}
develop the proposed inferential framework and estimation procedure.
Section~\ref{sec:identifiability_consistency} establishes full-law
identifiability and consistency from partially observed trajectories under a
structural model and regularity assumptions.
Section~\ref{sec:experiments} and \ref{sec:real-data} present the implementation details, simulation
studies, and PhysioNet data analysis. Section~\ref{sec:conclusion} concludes with discussion and future directions.

Throughout, let $[T]=\{1,2,\ldots,T\}$. 
Bold uppercase letters denote random variables (i.e., $\bX_t$), and the corresponding bold lowercase letters denote their realizations (i.e., $\bx_t$). We write $\bX_{1:j}=(\bX_1,\ldots,\bX_j)$. Let \(\law_{\btheta}(\bZ)\) denote the distribution of \(\bZ\) under parameter
\(\btheta\). For any vector $\bv\in\bbR^p$, let $\|\bv\|_\alpha=(\sum_{j=1}^p|\bv_j|^\alpha)^{1/\alpha}$ denote its $\ell_\alpha$-norm for $\alpha\geq 1$.

\section{Problem setup and a latent-driven model 
}\label{sec:set-up}

\noindent
We represent the complete longitudinal trajectory of a subject as a matrix $\bX_{1:T} = [\bX_1, \ldots,\bX_T] \in \bbR^{p \times T}$, where $\bX_t \in \bbR^p$ denotes the full feature vector at time $t$.  In many applications, including clinical settings,  $\bX_{1:T}$ is only partially observed. To model this, we introduce a binary mask $\bO_t \in \{0,1\}^p$, where $(\bO_t)_j = 1$ indicates that the $j$-th component of $\bX_t$ is observed, and $(\bO_t)_j = 0$ otherwise. Consequently, the observed vector $\bX_t^{\text{obs}}$ is defined component-wise as
\bel{def:obs}
\bigl(\bX_t^{\text{obs}}\bigr)_j = 
\begin{cases}
(\bX_t)_j, & \text{if } (\bO_t)_j = 1, \\
\mathsf{na}, & \text{if } (\bO_t)_j = 0,
\end{cases} \qquad \text{for each component } j = 1,\dots,p
\eel
where $\na$ denotes a missing entry (i.e., not available). In addition, we define the augmented observed history up to time $t$ as
$\wbX^{\text{obs}}_{1:t} \coloneq \bigl(\bX_{1:t}^{\text{obs}}, \bO_{1:t} \bigr)$, which jointly records the observed values and their corresponding missingness masks. 
Let $\calX\subseteq\bbR^p$ denote the complete-data space,
$\calX^{\text{obs}}\subseteq(\bbR\cup\{\na\})^p$ the partially observed
data space, and $\calO=\{0,1\}^p$ the mask space. The augmented observed-data
space is
\begin{align}\label{def:space:X-obs}
    \widebar{\calX}^{\text{obs}}
\coloneq
\left\{
(\bx^{\text{obs}},\bo)\in\calX^{\text{obs}}\times\calO:
\bx_j^{\text{obs}}=\na \text{ if and only if } \bo_j=0,
\quad j=1,\ldots,p
\right\}.
\end{align}
Accordingly, the space of observed histories up to time $t$ is $\widebar{\calX}_{1:t}^{\text{obs}}
\coloneq
\left(\widebar{\calX}^{\text{obs}}\right)^t.$

In this paper, we pursue two primary goals.
\begin{itemize}
    \item \textbf{Forecasting}  aims to generate samples from the predictive distribution $\law( \wbX_{t+J}\mid \wbX_{1:t}^{\text{obs}}) $ for any $t \in [T]$ and horizon $J \geq 1$.
    \item \textbf{Imputation} seeks to recover the missing entries in $\bX_{t}$ for $t\in [T]$  from the full posterior distribution $\law(\bX_t \mid \wbX_{1:T}^{\text{obs}})$.
\end{itemize}
To effectively model the temporal dynamics under such missingness, we consider a continuous-state Markov model, also known as a state-space model.
Formally, 
\bel{model:1}
\bZ_{t+1} \ &=& \ \trans(\bZ_t, \bv_t), \quad t \in [T-1] \quad \bigl(\text{transition model}\bigr); \cr
\bX_{t} \ &=& \ \out_x(\bZ_t,  \br_t^x) , \quad t \in [T] \quad \bigl(\text{measurement model}\bigr); \cr
\bO_t \ &=& \ \out_o(\bZ_t,  \br_t^o),  \quad t \in [T] \quad \bigl(\text{missingness model}\bigr);
\eel
with initial state $\bZ_1 \sim \calN(\bzero, \bI_d)$, where
$\bv_t \in \bbR^d$ and $\br_t^x,\br_t^o \in \bbR^p$. The function
$\trans:\bbR^d\times\bbR^d\rightarrow\bbR^d$ specifies the latent
transition, while $\out_x:\bbR^d\times\bbR^p\rightarrow\bbR^p$ and
$\out_o:\bbR^d\times\bbR^p\rightarrow\{0,1\}^p$ specify the measurement
and missingness functions, respectively.
The noise terms $\{\br_t^x\}_{t\ge 1}$, $\{\br_t^o\}_{t\ge 1}$, and
$\{\bv_t\}_{t\ge 1}$ are mutually independent sequences of i.i.d. random
variables, independent of $\bZ_1$.
\begin{figure}
    \centering
\includegraphics[width=1\linewidth]{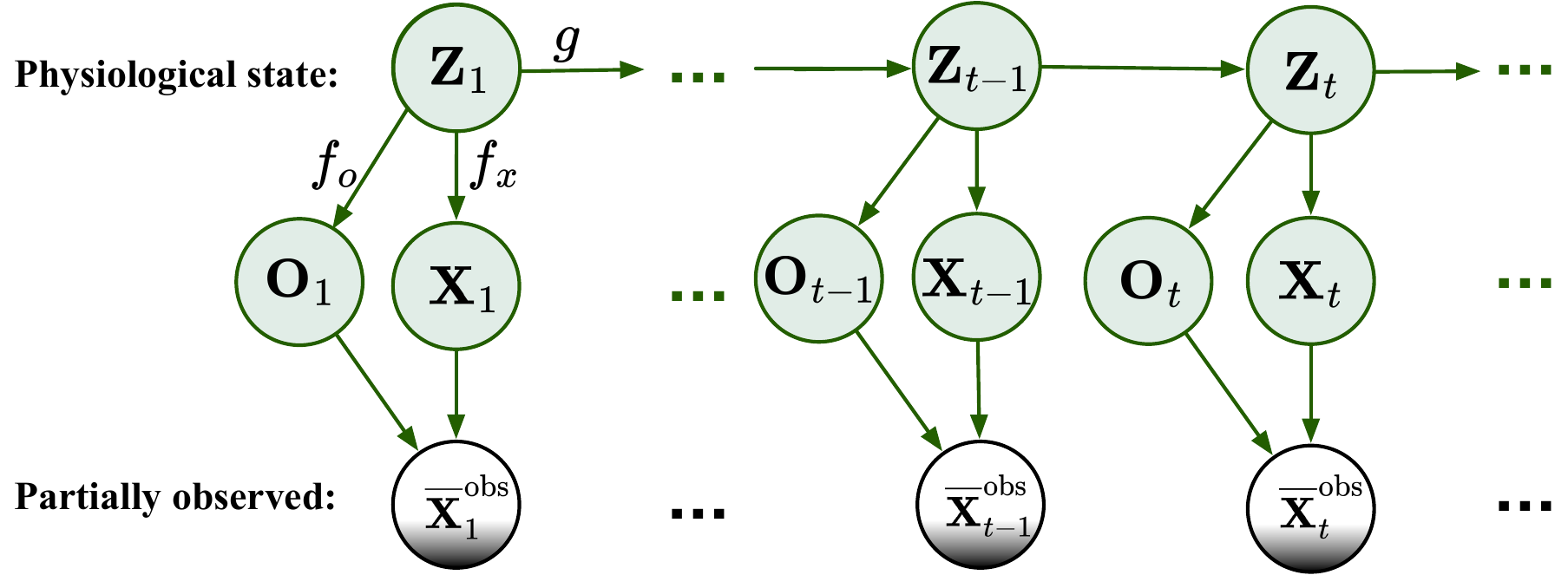}
    \caption{Directed acyclic graph (DAG)  illustrating the state-space model with missingness as defined in model \eqref{model:1}. At each time $t$, $\bZ_t$ represents the patient's latent physiological state, $\bX_t$ denotes the complete measurement vector, $\bO_t$ is the binary observation mask, and only the partially observed data $\bX_t^{\text{obs}}$ is available.}
    \label{fig:dag}
\end{figure}
In clinical applications, this framework offers a natural interpretation: $\bZ_t$ can be interpreted as a low-dimensional representation of the patient's physiological states at time $t$,  which governs both the observations $\bX_t^{\text{obs}}$ and mask $\bO_t$.

Crucially, model~\eqref{model:1} implies that the missingness is driven by the underlying state $\bZ_t$ rather than the measurement values themselves. We formalize this characteristic as follows:
\begin{definition}\label{def:ldm}
The missingness mechanism is said to satisfy
\emph{latent-driven missingness} (LDM) if $\bO_t$ and $\bX_t$
are conditionally independent given the latent state $\bZ_t$; that is,
\bes
    \bO_t \independent \bX_t \mid \bZ_t.
\ees
\end{definition}
Within the LDM framework, the missingness pattern is informative: observing which variables are measured, and when,
provides information about the underlying physiological trajectory. As
illustrated in Example~\ref{example:ldm}, this latent-state dependence generally
induces an MNAR mechanism. Indeed, because the latent state jointly drives the full measurements and the
mask, the conditional distribution of the mask given the data may depend on both
observed and unobserved components of the measurement vector through this latent
state.
\begin{example}[Latent-Driven MNAR]
\label{example:ldm}
Consider the case $p = 2$, $d = T = 1$. Let $z \sim \mathcal{N}(0,1)$ be a latent variable, and define the complete-data variables
\[
x_1 = z + \varepsilon_1, \qquad x_2 = z + \varepsilon_2,
\]
where $\varepsilon_1, \varepsilon_2 \sim \mathcal{N}(0,\sigma^2)$ are mutually independent and independent of $z$, with $0 < \sigma^2 < \infty$. The missingness mask $\bo = (o_1, o_2)$ is given by
\[
o_1 \equiv 1, \qquad o_2 = \bfone(z > 0),
\]
so that $x_1$ is always observed, while $x_2$ is missing whenever $z \leq 0$. By construction, the mechanism satisfies LDM, i.e., $o \independent x \mid z$. However, it is not MAR, since $\mathbb{P}(o_2 = 0 \mid x_1, x_2) \neq \mathbb{P}(o_2 = 0 \mid x_1),$
as the posterior distribution of $z$ given $(x_1, x_2)$ depends on both components. Hence, the mechanism is MNAR. 
\end{example}

\section{Forward-backward inference with LIFTS}\label{sec:methodology}

\noindent
{
Under the latent-driven model in Section~\ref{sec:set-up}, LIFTS organizes forecasting and imputation through forward filtering and backward smoothing. This structure parallels Kalman filtering
\citep{kalman1960new} and Rauch--Tung--Striebel smoothing
\citep{rauch1965maximum}, which provide exact conditional laws for a specified linear--Gaussian model. \syr{Particle filtering and smoothing extend these classical approaches to nonlinear and non-Gaussian state-space models \citep{liu1998sequential,chopin2020introduction}, with further developments for parameter learning \citep{chopin2013smc2} and  efficient backward smoothing \citep{dau2023backward}.}

As illustrated in Figure~\ref{fig:method-overview}, LIFTS approximates these conditional laws using learnable generators that do not require explicit density evaluation. An encoder provides filtered latent states for forecasting, a smoother samples latent trajectories backward, and a decoder generates conditional imputations.
}
\subsection{Forward pass: filtering and forecasting}\label{sec:forward_pass}
The forward pass serves two purposes: inferring the current latent state and forecasting future trajectories. We first focus on the filtering distribution $\law(\bZ_t \mid \wbX_{1:t}^{\text{obs}})$, which summarizes all information available up to time $t$. This conditional law is represented through a generative map, referred to as the \textit{encoder}, whose existence is formalized below.
\begin{proposition}[Existence of encoder representation]\label{prop:existence_enc}
Suppose model~\eqref{model:1} holds. Then there exists a measurable encoding
function $\enc:
\left(\bigcup_{t=1}^T\widebar{\calX}_{1:t}^{\text{obs}}\right)
\times\bbR^d
\longrightarrow\bbR^d$
such that, for each $t\in[T]$ and almost every observed history
$\wbx_{1:t}^{\mathrm{obs}}
\in\widebar{\calX}_{1:t}^{\mathrm{obs}}$,
\bel{obj:encoder}
\bZ_t \mid
\bigl\{\wbX_{1:t}^{\text{obs}}=\wbx_{1:t}^{\text{obs}}\bigr\}
\deq
\enc\left(\wbx_{1:t}^{\text{obs}},\cdeps_t\right),
\qquad \cdeps_t\sim\calN(\bzero,\bI_d).
\eel
Here, $\{\cdeps_t\}_{t=1}^T$ are independent Gaussian vectors, that are independent
of $\wbX_{1:T}^{\text{obs}}$.
\end{proposition}
\begin{remark}
 The encoder $\enc$ represents the filtering distribution induced by the state-space model. Hence it is determined by the transition, measurement, and missingness functions $\trans$, $\out_x$, and $\out_o$, together with the initial-state and noise distributions specified in model~\eqref{model:1}. However, unlike linear Gaussian models with complete data, where the Kalman filter provides an exact analytical solution, the closed-form expression for $\enc$ in our general setting is typically intractable.
\end{remark}
Consequently, we parameterize $\enc$ as a learnable function to approximate the filtering map. This approximation must adapt to the observed values and  time-varying masks in $\wbX_{1:t}^{\text{obs}}$, since the posterior uncertainty of $\bZ_t$ depends on both. This dependence distinguishes our encoder from VAE-based imputation methods, such as GP-VAE \citep{fortuin2020gp}, which rely on prespecified latent priors and are not designed to represent filtering distributions under dynamic masking.
The proof of Proposition \ref{prop:existence_enc} is provided in Appendix \ref{pf:prop:existence_enc}.

Given this filtering representation, one-step forecasting is obtained by propagating the latent state through the transition and output maps, as formalized below.
\begin{corollary}[Forward forecasting]\label{thm:forecasting}
Suppose model~\eqref{model:1} holds, and let $\enc$ be the encoder in
Proposition~\ref{prop:existence_enc}. Then there exist measurable maps $\btrans:\bbR^d\times\bbR^d\to\bbR^d,\quad
\bout:\bbR^d\times\bbR^{2p}\to\calX\times\calO,$
such that, for each $t\in[T]$ and almost every
$\wbx_{1:t}^{\text{obs}}\in\widebar{\calX}_{1:t}^{\text{obs}}$,
\bel{obj:forecast}
\left(\bX_{t+1},\bO_{t+1}\right)
\mid
\bigl\{\wbX_{1:t}^{\text{obs}}=\wbx_{1:t}^{\text{obs}}\bigr\}
\deq
\bout\left\{
\btrans\left(
\enc\left(\wbx_{1:t}^{\text{obs}},\cdeps_t\right),\deps_t
\right),\deta_t
\right\},
\eel
where $\cdeps_t,\deps_t\sim\calN(\bzero,\bI_d)$ and
$\deta_t\sim\calN(\bzero,\bI_{2p})$ are independent and independent
of $\wbX_{1:t}^{\text{obs}}$.
\end{corollary}
\begin{remark}\label{remark:gaussian_noise}
The reference noises $\deps_t$ and $\deta_t$ need not match the model-noise
distributions in \eqref{model:1}. By standard randomization, the model noises
can be represented as measurable transformations of standard Gaussian vectors.
The map $\btrans$ therefore combines $\trans$ with the transition-noise
transformation, while $\bout$ combines $\out_x$ and $\out_o$ with the
corresponding output-noise transformations.
\end{remark}
Corollary~\ref{thm:forecasting} provides a one-step-ahead predictive representation given almost every partially observed history up to time $t\in[T]$. At the terminal time $T$, multi-step forecasting is obtained by recursively propagating the sampled latent state through $\btrans$ and applying the output map at each future time point, yielding trajectories $\wbX_{T+1:T+s}$ for any horizon $s\geq 1$.
\begin{figure}[t]
    \centering
    \includegraphics[width=1\linewidth]{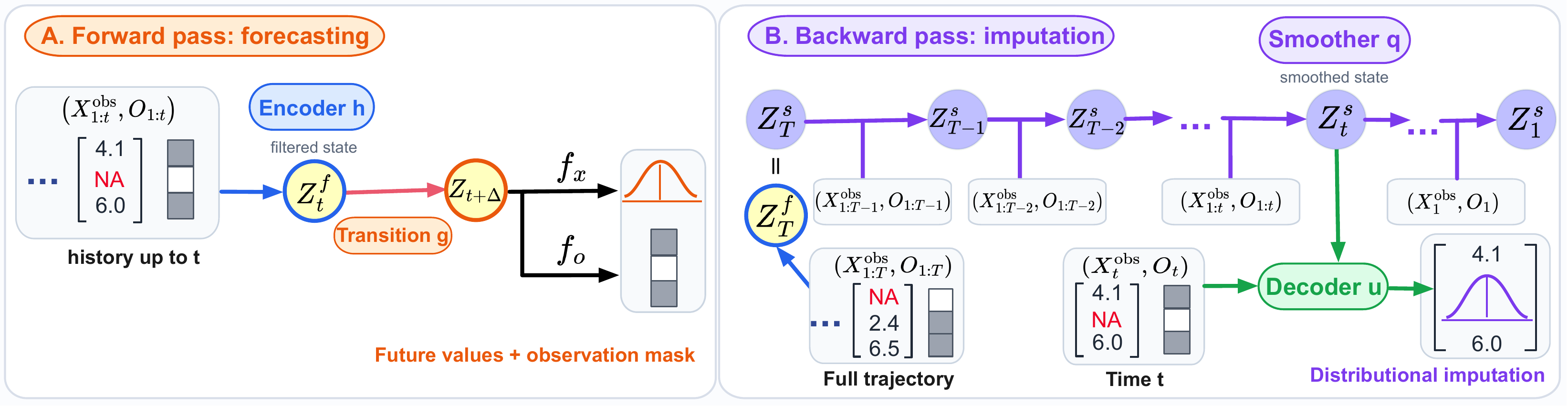}
    \caption{Schematic illustration of the inference framework for forecasting and imputation.
    }
    \label{fig:method-overview}
\end{figure}
\subsection{Backward pass: smoothing and imputation}
Having established the forward pass, we now address the generation of \(\bX_t\) conditional on the full observed trajectory, i.e., \(\law(\bX_t \mid \wbX_{1:T}^{\text{obs}})\). For ease of exposition, we use conditional density notation below; when such densities are not available, the same arguments can be interpreted in terms of conditional distributions. By marginalizing over the latent state $\bZ_t$, and utilizing the conditional independence properties $(iv)$ in Lemma \ref{lemma:fact-1}, this target distribution admits the decomposition:
\begin{align}\label{eq:sample-1}
p(\bX_t \mid \wbX_{1:T}^{\text{obs}})
= \int \underbrace{p(\bX_t \mid \bZ_t, \wbX_t^{\text{obs}})}_{\text{via decoder }\dec} \,\cdot \ \underbrace{p(\bZ_t \mid \wbX_{1:T}^{\text{obs}})}_{\text{smoothing distribution}} \, \dd\bZ_t.
\end{align}
\ We use the decoder $\dec$ to generate samples from the first conditional law in \eqref{eq:sample-1}. The main difficulty is therefore to sample from the smoothed posterior $\law(\bZ_t \mid \wbX_{1:T}^{\text{obs}})$. By the backward conditional independence in Lemma~\ref{lemma:fact-1} $(iii)$, this smoothed posterior admits the recursive decomposition
\begin{align}\label{eq:sample-2}
p(\bZ_t \mid \wbX_{1:T}^{\text{obs}}) =  \int   \underbrace{p(\bZ_t \mid \bZ_{t+1}, \wbX_{1:t}^{\text{obs}})}_{\text{via smoother } \backt} \ \cdot  \underbrace{p(\bZ_{t+1} \mid \wbX_{1:T}^{\text{obs}})}_{\text{next smoothing distribution}} \ \dd\bZ_{t+1}.
\end{align}
The recursion is initialized at the terminal step $T$ by sampling $\bZ_T \mid \wbX_{1:T}^{\text{obs}}$ using the encoder. For $t < T$, the key component is the {backward conditional distribution} $\law(\bZ_t \mid \bZ_{t+1}, \wbX_{1:t}^{\text{obs}})$. When the corresponding densities exist, Bayes' rule and the forward Markov property imply $p(\bZ_t \mid \bZ_{t+1}, \wbX_{1:t}^{\text{obs}}) \ \propto \  p(\bZ_{t+1} \mid \bZ_t) \cdot p(\bZ_t \mid \wbX_{1:t}^{\text{obs}}).$
Although it theoretically links this distribution to the transition kernel and filtering law, direct sampling is generally intractable for implicit generative models.  In particular, when the transition $\trans(\cdot)$ and encoder $\enc(\cdot)$ are parameterized by neural networks (without explicit density functions), standard sampling techniques that rely on density evaluations, such as rejection sampling or Markov chain Monte Carlo (MCMC), cannot be applied directly.

We approximate the backward conditional law with a parameterized smoother
\(\backt\), enabling direct sampling without density evaluation. The
procedure is formalized in Theorem~\ref{thm:imputation}.
\begin{theorem}[Backward imputation recursion]\label{thm:imputation}
Suppose model~\eqref{model:1} holds. Then there exist measurable maps
$\backt$ and $\dec$ such that, for almost every
$\wbx_{1:T}^{\text{obs}}\in\widebar{\calX}_{1:T}^{\text{obs}}$, starting from
$\bZ_T^s=\enc(\wbx_{1:T}^{\text{obs}},\cdeps_T)$, the recursion
\bel{eq:obj-smooth-1}
\bZ_t^s &\coloneq&
\backt\left(\bZ_{t+1}^s,\wbx_{1:t}^{\text{obs}},\tdxi_t\right),
\qquad \tdxi_t\sim\calN(\bzero,\bI_d),\quad t= T-1, \ldots,1,
\eel
generates the joint smoothing distribution
\bel{eq:obj-smooth-joint}
(\bZ_1,\ldots,\bZ_T)\mid
\bigl\{\wbX_{1:T}^{\text{obs}}=\wbx_{1:T}^{\text{obs}}\bigr\}
&\deq& (\bZ_1^s,\ldots,\bZ_T^s).
\eel
Moreover, setting
$\tbX_t\coloneq\ \dec(\bZ_t^s,\wbx_t^{\text{obs}},\tdeta_t)$, where
$\tdeta_t\sim\calN(\bzero,\bI_p)$, gives
\bel{eq:obj-smooth-2}
(\bX_1,\ldots,\bX_T)\mid
\bigl\{\wbX_{1:T}^{\text{obs}}=\wbx_{1:T}^{\text{obs}}\bigr\}
\deq (\tbX_1,\ldots,\tbX_T).
\eel
All reference noises are mutually independent and independent of
\(\wbX_{1:T}^{\text{obs}}\).
\end{theorem}
Theorem~\ref{thm:imputation} shows that the recursive construction recovers the joint smoothing distribution of the latent trajectory, rather than only its time-wise marginals. Thus, the generated latent paths preserve the temporal dependence encoded by the state-space model. Applying the decoder to the resulting latent trajectory then yields a joint imputation of the full sequence.

\syr{The backward recursion
is also fundamental to classical particle smoothing \citep{godsill2004monte,dau2023backward},
with recent neural extensions \citep{younis2024learning}.
While particle methods approximate latent-state posterior distributions using weighted particles, LIFTS employs learned conditional generators for direct sampling without particle weighting or resampling at inference time. Beyond this computational distinction, LIFTS jointly learns the measurement and missingness processes from partially observed trajectories, with latent-state inference serving as an intermediate step for forecasting and imputation.}

Prior to parameter estimation, we assume that $\bout$, $\btrans$, $\enc$, $\dec$, and $\backt$ are correctly specified. The overall forecasting and imputation procedure is summarized in Algorithm~\ref{alg:inference-LIFTS}.
\begin{algorithm}[H]
\caption{LIFTS inference for forecasting and imputation}
\label{alg:inference-LIFTS}
\setstretch{0.96}
\setlength{\abovedisplayskip}{2pt}
\setlength{\belowdisplayskip}{2pt}
\renewcommand{\algorithmicrequire}{\textbf{Input:}}
\begin{algorithmic}[1]
\Require Pretrained maps \(\bout,\btrans,\enc,\dec,\backt\), partially observed trajectory
\(\wbx_{1:T}^{\mathrm{obs}}\), and number of Monte Carlo samples \(J\).
\State Draw terminal smoothed states, for \(j=1,\ldots,J\),
\[
\bz_T^{s,(j)}
=
\enc\left(\wbx_{1:T}^{\mathrm{obs}},\cdeps_T^{(j)}\right),
\qquad
\cdeps_T^{(j)}\sim\calN(\bzero,\bI_d).
\]
\State \textbf{Forecasting: }generate one-step forecasts
\[
\left(\hbx_{T+1}^{(j)},\hbo_{T+1}^{(j)}\right)
=
\bout\left(
\btrans\left(\bz_T^{s,(j)},\deps_T^{(j)}\right),
\deta_T^{(j)}
\right),
\quad j=1,\ldots,J,
\]
where \(\deps_T^{(j)}\sim\calN(\bzero,\bI_d)\) and
\(\deta_T^{(j)}\sim\calN(\bzero,\bI_{2p})\).
\For{\(t=T,T-1,\ldots,1\)}
    \If{\(t<T\)}
        \State Draw backward smoothed states
        \[
        \bz_t^{s,(j)}
        =
        \backt\left(\bz_{t+1}^{s,(j)},\wbx_{1:t}^{\mathrm{obs}},
        \tdxi_t^{(j)}\right),
        \quad
        \tdxi_t^{(j)}\sim\calN(\bzero,\bI_d),\quad j=1,\ldots,J.
        \]
    \EndIf
    \State \textbf{Imputation:} impute complete observations
    \[
    \tbx_t^{(j)}
    =
    \dec\left(\bz_t^{s,(j)},\wbx_t^{\mathrm{obs}},\tdeta_t^{(j)}\right),
    \quad
    \tdeta_t^{(j)}\sim\calN(\bzero,\bI_p),\quad j=1,\ldots,J.
    \]
\EndFor
\State \textbf{Return:}
\(\{(\hbx_{T+1}^{(j)},\hbo_{T+1}^{(j)})\}_{j=1}^{J}\) and
\(\{\tbx_{1:T}^{(j)}\}_{j=1}^{J}\).
\end{algorithmic}
\end{algorithm}
\section{Joint learning of generative and inference maps in LIFTS}\label{sec:estimation}

\noindent
We now turn to estimation of the transition map \(\btrans\), the output map \(\bout\), and the inference maps \(\enc\), \(\dec\), and \(\backt\). These components are estimated jointly, as described in Section~\ref{subsec:joint-training}.

A central choice in distributional learning is the criterion used to assess
agreement between a predictive law and the target law. For a predictive
distribution \(Q\) and an observation \(\bx\), the energy score
\citep{gneiting2007strictly} is defined as
\bel{def:es}
S(\bx,Q)
\coloneq
\frac{1}{2}\bbE_{\bX,\bX'\sim Q}\|\bX-\bX'\|_2
-
\bbE_{\bX\sim Q}\|\bx-\bX\|_2 .
\eel
The first term captures the dispersion among predictive draws, whereas the
second measures their distance from the observed realization \(\bx\). For a
target distribution \(P\), we further define the population energy score as
\[
S(P,Q) \coloneq \bbE_{\bx\sim P} S(\bx,Q).
\]
It is this population score that measures distributional agreement between
\(P\) and \(Q\), with larger values indicating better agreement.

The energy score is strictly proper: for any target distribution \(P\),
\[
S(P,Q) \leq S(P,P),
\]
with equality if and only if \(Q=P\)
\citep{gneiting2007strictly,szekely2013energy}. Moreover, for conditional
predictive laws, averaging sample scores based on one observed outcome per
conditioning history and independent draws from the corresponding predictive
law yields an unbiased estimator of the expected population score
\citep{shen2024distributional,shen2025engression}. This provides a directly
estimable sample-based objective that is particularly suitable for implicit
generative models, as it avoids explicit density evaluation required by
KL-based objectives and optimal transport problems associated with
Wasserstein distances.
Our training losses will use the negative
energy score and, for partially observed outcomes, its masked version
$S^{\text{mask}}(\cdot;\cdot)$ defined in \eqref{mask-es} below.

\subsection{Masked energy score}\label{subsec:masked-es}
Given two observed realizations
\(\wbx_i^{\text{obs}}=(\bx_i^{\text{obs}},\bo_i)
\in\widebar{\calX}^{\text{obs}}\), \(i=1,2\), we define a projected distance
as follows:
\bel{dist-proj-main}
\operatorname{dist}(\wbx^{\text{obs}}_1, \wbx^{\text{obs}}_2)
=
\left(\bigl\|\bx^{\text{obs}}_1 \odot \bo_1 - \bx^{\text{obs}}_2\odot \bo_2\bigr\|_2^2
+  \bigl\|\bo_1 - \bo_2\bigr\|_2^2\right)^{1/2}.
\eel
where the operator $\odot$ represents an augmented Hadamard product, with the convention that $\na \cdot 0 =0$.  We subsequently define the \emph{masked energy score} as
\bel{mask-es}
S^{\text{mask}}(\wbx^{\text{obs}}; Q)
&\coloneq&
\frac{1}{2}\bbE_{\wbX_1^{\text{obs}}, \wbX_2^{\text{obs}} \sim Q}
\operatorname{dist}(\wbX_1^{\text{obs}}, \wbX_2^{\text{obs}})
- \bbE_{\wbX^{\text{obs}} \sim Q}
\operatorname{dist}(\wbx^{\text{obs}}, \wbX^{\text{obs}}).
\eel 
For complete data, where every entry of $\bO_t$ equals one, $S^{\text{mask}}$ becomes the classical energy score.
\begin{proposition}[Strict propriety of the masked energy score]
\label{prop:mask-proper}
Let $P_o$ be a distribution on $\widebar{\calX}^{\mathrm{obs}}$, and let
$\bY\sim P_o$. For any distribution $Q$ on
$\widebar{\calX}^{\mathrm{obs}}$ such that $P_o$ and $Q$ have finite first
moments under $\operatorname{dist}$,
\[
\bbE_{\bY\sim P_o}\bigl[S^{\text{mask}}(\bY;P_o)\bigr]
\geq
\bbE_{\bY\sim P_o}\bigl[S^{\text{mask}}(\bY;Q)\bigr],
\]
with equality if and only if $P_o=Q$.
\end{proposition}
Note that a larger score indicates a superior distributional prediction. In practice, we minimize the negative masked energy score when constructing the empirical loss function. Proposition~\ref{prop:mask-proper} shows that the masked energy score defined in \eqref{mask-es} constitutes a strictly proper scoring rule, with the formal proof provided in Appendix~\ref{pf:prop:mask-proper}.
\subsection{Joint training}\label{subsec:joint-training}
We estimate the forward and backward components jointly using three
distributional losses. The forecasting loss guides the estimation of \(\btrans\), \(\bout\), and \(\enc\) by matching the one-step observation law; the synthetic loss calibrates \(\enc\) to the filtering law implied by the current generative model; and the imputation loss guides the estimation of \(\backt\) and \(\dec\) via a self-masking strategy. 
    We write
\[
\calL_{\mathrm{fore}}(\btrans,\bout,\enc),\qquad
\calL_{\mathrm{syn},\enc}(\enc;\btrans,\bout),\qquad
\calL_{\mathrm{imp}}(\backt,\dec;\enc),
\]

for the forecasting, calibration, and imputation losses defined in
\eqref{joint-train-forecasting}, \eqref{def:L_syn_enc}, and
\eqref{joint-train-imp}, respectively. These losses update
\((\btrans,\bout,\enc)\), \(\enc\), and \((\backt,\dec)\), respectively.

The observation law admits the factorization
\(p(\wbx_{1:T}^{\mathrm{obs}})=p(\wbx_1^{\mathrm{obs}})
\prod_{t=1}^{T-1}p(\wbx_{t+1}^{\mathrm{obs}}\mid
\wbx_{1:t}^{\mathrm{obs}})\).  
We therefore train the forward model by matching each one-step conditional law.
Let
\(Q_{t+1\mid t}^{\mathrm{fore}}\)
denote the predictive law generated by \((\enc,\btrans,\bout)\) given \(\wbX_{1:t}^{\mathrm{obs}}\) as in
\eqref{obj:forecast}. The forecasting objective minimizes the negative masked
energy score,
\begin{align}\label{joint-train-forecasting}
\calL_{\mathrm{fore}}(\btrans,\bout,\enc)
&\coloneq
-
\bbE
\left[
\frac{1}{T}\sum_{t=0}^{T-1}
S^{\mathrm{mask}}
\left(
\wbX_{t+1}^{\mathrm{obs}};
Q_{t+1\mid t}^{\mathrm{fore}}
\right)
\right] \notag \\
&=
\frac{1}{T}\sum_{t=0}^{T-1}\bbE
\left[
\bbE_{\widehat\wbX^{\text{obs}}_{t+1} \sim Q_{t+1\mid t}^{\mathrm{fore}}}
\operatorname{dist}(\wbX_{t+1}^{\mathrm{obs}}, \widehat\wbX^{\text{obs}}_{t+1})
- \frac{1}{2}
\bbE_{\widehat\wbX_{t+1}^{(1)}, \widehat\wbX_{t+1}^{(2)} \sim Q_{t+1\mid t}^{\mathrm{fore}}}
\operatorname{dist}(\widehat\wbX_{t+1}^{(1)}, \widehat\wbX_{t+1}^{(2)})
\right].
\end{align}
Specifically,
\(Q_{1\mid 0}^{\mathrm{fore}}\) is generated
directly from \(\bout(\bZ_1,\deta_1)\).

The calibration loss uses the current transition and output maps to
generate synthetic pairs
\(\{(\wbX_{1:t}^{\mathrm{obs},\mathrm{syn},m},
\bZ_t^{\mathrm{syn},m})\}_{m\ge1}\) from the corresponding noise-outsourced
generative process. Since \(\bZ_t^{\mathrm{syn}}\) is available, the encoder can
be calibrated directly against the filtering law that generated it. Let
\(\law_t^{\enc}(\cdot\mid\wbX_{1:t}^{\mathrm{obs},\mathrm{syn}})\) denote the
encoder law, and define
\begin{align}\label{def:L_syn_enc}
\calL_{\mathrm{syn},\enc}(\enc;\btrans,\bout)
\coloneq
\frac{1}{T}\sum_{t=1}^{T}\bbE
\left[
\bbE_{\bZ^f_t\sim \law_t^{\enc}}
\|\bZ_t^{\mathrm{syn}}-\bZ^f_t\|_2
-
\frac{1}{2}
\bbE_{\bZ^f_t,\bZ^{f'}_t\sim \law_t^{\enc}}
\|\bZ^f_t-\bZ^{f'}_t\|_2
\right].
\end{align}
For imputation, the originally missing entries are unavailable as training
targets. We therefore use a self-masking strategy.  For each
trajectory, let \(\bM_{1:T}\) randomly retain \(80\%\) of the entries observed
under \(\bO_{1:T}\), and let
\(\law_t^{\mathrm{imp}}(\cdot\mid\bX_{1:T}^{\mathrm{obs}}\odot\bM_{1:T},\bM_{1:T})\)
denote the imputation law generated as in \eqref{eq:obj-smooth-2}. The held-out
entries \(\bO_{1:T}-\bM_{1:T}\) provide targets for matching the corresponding
conditional law. By construction, each generated trajectory preserves the retained entries,
$\widetilde{\bX}_t\odot\bM_t=\bX_t^{\mathrm{obs}}\odot\bM_t$; hence evaluating
the loss on $\bO_t$ is equivalent to evaluating it on $\bO_t-\bM_t$. The imputation loss $\calL_{\mathrm{imp}}(\backt,\dec;\enc)$ is the corresponding negative energy score:
\begin{align}\label{joint-train-imp}
    \frac{1}{T}\sum_{t=1}^{T}\bbE\!\Bigg[
\bbE_{\widetilde{\bX}_t\sim \law_t^{\mathrm{imp}}}
\left\|
\bX_t^{\mathrm{obs}} \odot \bO_t
-
\widetilde{\bX}_t\odot\bO_t
\right\|_2 -
\frac{1}{2}
\bbE_{\widetilde{\bX}_t,\widetilde{\bX}'_t\sim \law_t^{\mathrm{imp}}}
\left\|
\widetilde{\bX}_t\odot\bO_t
-
\widetilde{\bX}'_t\odot\bO_t
\right\|_2
\Bigg].
\end{align}
The weighted population objective is
\begin{align}\label{joint-train-total}
\calL_{\text{total}}(\calM)
\coloneq{}
\calL_{\text{fore}}(\btrans,\bout,\enc)
+\lambda\calL_{\text{imp}}(\backt,\dec;\enc)
+\lambda_{\enc}\calL_{\text{syn},\enc}(\enc;\btrans,\bout),
\end{align}
where \(\calM=(\btrans,\bout,\enc,\backt,\dec)\) collects the parameterized
maps, and \(\lambda\) and \(\lambda_{\enc}\) weight the imputation and synthetic
calibration losses, respectively.
The corresponding finite-sample forecasting, imputation, and synthetic
calibration losses will be
defined below in \eqref{def:L_hat_fore}, \eqref{def:L_hat_imp}, and
\eqref{def:L_hat_syn_enc}, respectively. Their weighted sum defines the empirical
objective
\begin{align}\label{joint-train-total-empirical}
\widehat{\calL}_{\mathrm{total},n,B}(\calM)
=
\widehat{\calL}_{\mathrm{fore},n}(\btrans,\bout,\enc)
+
\lambda\widehat{\calL}_{\mathrm{imp},n}(\backt,\dec;\enc)
+
\lambda_{\enc}
\widehat{\calL}_{\mathrm{syn},\enc,B}(\enc;\btrans,\bout),
\end{align}

which is an unbiased estimator of \(\calL_{\mathrm{total}}(\calM)\) in
\eqref{joint-train-total}. Specifically, the outer expectations in
\(\calL_{\mathrm{total}}(\calM)\) are replaced by averages over observed or
synthetically generated trajectories. For each trajectory and time point, two
independent draws from the relevant predictive, imputation, or encoder law are
used to estimate the corresponding energy score term. Under the neural-network parameterization of \(\calM\), the empirical objective
is differentiable almost everywhere and optimized using a stochastic gradient
method (Algorithm~\ref{alg:joint-training}); see
Section~\ref{subsec:implementation} for implementation details.

\begin{algorithm}[H]
\caption{Joint training for LIFTS}
\label{alg:joint-training}
\setstretch{1.1}
\setlength{\abovedisplayskip}{1pt}
\setlength{\belowdisplayskip}{1pt}
\setlength{\abovedisplayshortskip}{1pt}
\setlength{\belowdisplayshortskip}{1pt}
\renewcommand{\algorithmicrequire}{\textbf{Configuration:}}
\begin{algorithmic}[1]
\Require \qquad\qquad\qquad
\begin{itemize}
    \item Observed data \(\{\wbx_{1:T}^{\text{obs},i}\}_{i=1}^{n}\); number of synthetic trajectories \(B\); loss weights
\(\lambda,\lambda_{\enc}\); and initialized maps $\calM$.
    \item At each iteration, all reference noises are drawn independently from their specified standard normal distributions.
    \item A map declared ``held fixed'' receives no gradient from the corresponding loss.
\end{itemize}

\bigskip

\For{\emph{each training iteration}}
\smallskip
    \State \textbf{Forecasting.} For \(i\in[n]\) and \(j\in[2]\), draw
    \(\hwbx_{1,j}^{\text{obs},i}=\bout(\bz_{1,j}^{i},\deta_{1,j}^{i})\),
    where \(\bz_{1,j}^{i}\sim\calN(\bzero,\bI_d)\). \Statex \hspace{\algorithmicindent} For \(t = 1, \ldots, T-1\), set
  \(\displaystyle
    \hwbx_{t+1,j}^{\text{obs},i}\!=\!\bout\!\left\{
    \btrans\!\left(
    \enc\!\left(\wbx_{1:t}^{\text{obs},i},\cdeps_{t,j}^{i}\right),
    \deps_{t,j}^{i}\right),\deta_{t+1,j}^{i}\right\}.\)  Compute
 
   \begin{equation}
     \begin{aligned}\label{def:L_hat_fore}
    \widehat{\calL}_{\mathrm{fore},n}
    \left(\btrans,\bout,\enc\right)
    \!=\!\frac{1}{2nT}\sum_{i=1}^n\sum_{t=0}^{T-1}
    \!\Bigl[
    \operatorname{dist}\!\left(\wbx_{t+1}^{\text{obs},i},
    \hwbx_{t+1,1}^{\text{obs},i}\right)
    \!+\!\operatorname{dist}\!\left(\wbx_{t+1}^{\text{obs},i},
    \hwbx_{t+1,2}^{\text{obs},i}\right)
    \!-\!\operatorname{dist}\!\left(\hwbx_{t+1,1}^{\text{obs},i},
    \hwbx_{t+1,2}^{\text{obs},i}\right)
    \!\Bigr].
    \end{aligned}
   \end{equation}
    \State \textbf{Imputation via self-masking.} For each \(i\in[n]\), randomly retain
    \(80\%\) of the non-zero \Statex \hspace{\algorithmicindent} entries of  \(\bo_{1:T}^{i}\) to form
    \(\bM_{1:T}^{i}\), and write
    \(\wbx_{1:t}^{M,i}\coloneq
    (\bx_{1:t}^{\mathrm{obs},i}\odot\bM_{1:t}^{i},\bM_{1:t}^{i})\).
    For \(j\in[2]\), holding \Statex \hspace{\algorithmicindent}  $\enc$ fixed,  recursively set
    \bes
    \bz_{T,j}^{s,i}&=&\enc(\wbx_{1:T}^{M,i},\cdeps_{T,j}^{i}),\qquad
\bz_{t,j}^{s,i}=\backt(\bz_{t+1,j}^{s,i},\wbx_{1:t}^{M,i},\tdxi_{t,j}^{i}),
    \quad t=T-1,\ldots,1, \\
\tbx_{t,j}^{i}&=&\dec(\bz_{t,j}^{s,i},\wbx_t^{M,i},\tdeta_{t,j}^{i}),
    \quad t\in[T].
    \ees
    \Statex \hspace{\algorithmicindent}The decoder preserves the retained entries: $\tbx_{t,j}^{i} \odot\bM_{t}^{i} = \bx_t^{\mathrm{obs},i} \odot\bM_{t}^{i}$. Compute
 \begin{equation}\label{def:L_hat_imp}
\begin{aligned}
    \widehat{\calL}_{\mathrm{imp},n}
    \left(\backt,\dec;\enc\right)
    &=\frac{1}{2nT}\sum_{i=1}^n\sum_{t=1}^{T}
    \Bigl[
    \left\|\bx_t^{\mathrm{obs},i}\odot\bo_t^i
    -\tbx_{t,1}^{i}\odot\bo_t^i\right\|_2 \\
    &\quad+\left\|\bx_t^{\mathrm{obs},i}\odot\bo_t^i
    -\tbx_{t,2}^{i}\odot\bo_t^i\right\|_2
    -\left\|\tbx_{t,1}^{i}\odot\bo_t^i
    -\tbx_{t,2}^{i}\odot\bo_t^i\right\|_2
    \Bigr].
    \end{aligned}
 \end{equation}
\smallskip
    \State \textbf{Encoder calibration.} Using the fixed maps \((\btrans,\bout)\), generate \(B\) synthetic trajectories \Statex \hspace{\algorithmicindent}and form the corresponding pairs 
    \(\{(\wbx_{1:t}^{\mathrm{obs},\mathrm{syn},b},
    \bz_t^{\mathrm{syn},b})\}_{b\in[B],t\in[T]}\). For \(b\in[B], t\in[T]\), and \Statex \hspace{\algorithmicindent} \(j\in[2]\), draw 
    \(\hbz_{t,j}^{\mathrm{syn},b}
    =\enc(\wbx_{1:t}^{\mathrm{obs},\mathrm{syn},b},
    \cdeps_{t,j}^{\mathrm{syn},b})\). Compute
\begin{equation}\label{def:L_hat_syn_enc}
        \begin{aligned}
    \widehat{\calL}_{\mathrm{syn},\enc,B}
    \left(\enc;\btrans,\bout\right)
    =\frac{1}{2BT}\sum_{b=1}^{B}\sum_{t=1}^{T}
    \left[
    \left\|\bz_t^{\mathrm{syn},b}-\hbz_{t,1}^{\mathrm{syn},b}\right\|_2
    +\left\|\bz_t^{\mathrm{syn},b}-\hbz_{t,2}^{\mathrm{syn},b}\right\|_2
    -\left\|\hbz_{t,1}^{\mathrm{syn},b}-\hbz_{t,2}^{\mathrm{syn},b}\right\|_2
    \right].
    \end{aligned}
\end{equation}

    \State \textbf{Optimization.} Set
    \(
    \widehat{\calL}_{\mathrm{total},n,B}
    =     \widehat{\calL}_{\mathrm{fore},n}
    \left(\btrans,\bout,\enc\right)
    +\lambda\widehat{\calL}_{\mathrm{imp},n}(\backt,\dec;\enc)
    +\lambda_{\enc}\widehat{\calL}_{\mathrm{syn},\enc,B}(\enc;\btrans,\bout)\),\Statex \hspace{\algorithmicindent} and update  $\calM$ using one  Adam step. 

    \smallskip
    
\EndFor
\State \textbf{Return:}
\(\calM
=
\left(\btrans,\bout,\enc,\backt,\dec\right)\).
\end{algorithmic}
\end{algorithm}

\section{Identifiability and consistency}
\label{sec:identifiability_consistency}

\noindent
Under the general data-generating process in model~\eqref{model:1}, the observed-data law does not in general identify the full-data law, because distinct distributions of the unobserved entries and missingness mechanisms may induce the same observed-data law.  We therefore
show that, under a structural model and mild regularity conditions, the full
law is identifiable and consistently recoverable from the partially observed trajectories.
\subsection{Main result}
Throughout this section, let $\Theta\subset\bbR^q$ denote the parameter space
jointly indexing the transition, measurement, and missingness function classes.
For a generic parameter $\btheta\in\Theta$, our theoretical analysis focuses on
the following structural subclass of model~\eqref{model:1}:
\begin{align}
\bZ_{t+1} &= \trans^{\btheta}(\bZ_t,\bv_t),
\qquad t\in[T-1], \notag\\
\bX_t &= f_x^{\btheta}(\bZ_t)+\br_t,
\qquad t\in[T], \notag\\
\bbP(\bO_{t,j}=1\mid \bZ_t=\bz_t)
&=\sigma\{\pi_j^{\btheta}(\bz_t)\},
\qquad t\in[T],\quad j\in[p].
\label{eq:dgp_main}
\end{align}
Here \(f_x^{\btheta}\) is the noiseless measurement map, \(\pi_j^{\btheta}\) is
the \(j\)-th missingness logit, \(\sigma\) is the sigmoid function, and {the mask
coordinates are conditionally independent given \(\bz_t\).}

Compared to the general model~\eqref{model:1}, the structured
model~\eqref{eq:dgp_main} imposes two identification-oriented restrictions: an
additive observation model with known output noise, and a logistic missingness
mechanism governed by functions \(\pi_1,\ldots,\pi_p\). These
restrictions support identification: the known noise
law allows the distribution of \(f_x^{\btheta}(\bZ_t)\) to be recovered from the
fully observed law of \(\bX_t\) by deconvolution, while suitable mask patterns
identify the missingness logits. We assume
that the initial state $\bZ_1$, transition noises \(\{\bv_t\}_{t\ge1}\), and output
noises \(\{\br_t\}_{t\ge1}\) are mutually independent, with fixed and known
laws. 
Let
\begin{align}
    Q_T^{\btheta}\coloneq \law_{\btheta}(\bX_{1:T},\bO_{1:T}),
\qquad
Q_{o,T}^{\btheta}\coloneq \law_{\btheta}(\bX_{1:T}^{\text{obs}},\bO_{1:T})
\end{align}
denote the full trajectory law and the observed trajectory law under \(\btheta\), respectively. The true data-generating parameter is denoted by \(\btheta^*\), with corresponding laws \(Q_T^{\btheta^*}\) and \(Q_{o,T}^{\btheta^*}\). 

We next define the encoder-assisted one-step predictive law. Let \(\Phi\) denote the
encoder parameter space. For \((\btheta,\bphi)\in\Theta\times\Phi\) and \(t\ge2\), given
\(\wbx_{1:t-1}^{\text{obs}}\), draw a filtered latent state from
\(\enc^{\bphi}(\wbx_{1:t-1}^{\text{obs}},\cdeps_t)\), and propagate it one step
through \(\trans^{\btheta}\), \(f_x^{\btheta}\), and
\(\pi_1^{\btheta},\ldots,\pi_p^{\btheta}\). The resulting conditional full
predictive law is denoted by
\(Q_{t\mid 1:t-1}^{\btheta,\bphi}(\cdot\mid\wbx_{1:t-1}^{\text{obs}})\), and the
corresponding observed law by
\(Q_{o,t\mid 1:t-1}^{\btheta,\bphi}(\cdot\mid\wbx_{1:t-1}^{\text{obs}})\).
For \(t=1\), set \(Q_1^{\btheta,\bphi}=Q_1^{\btheta}\). The product of the one-step
observed conditional laws defines \(Q_{o,T}^{\btheta,\bphi}\).
Given an observed trajectory \(\wbx_{1:T}^{\text{obs}}\), define
\begin{align}
\ell_{\btheta,\bphi}(\wbx_{1:T}^{\text{obs}})
\coloneq
- \frac{1}{T}\sum_{t=1}^{T}
S^{\text{mask}}
\left(
\wbx_t^{\text{obs}};
Q_{o,t\mid 1:t-1}^{\btheta,\bphi}(\cdot\mid\wbx_{1:t-1}^{\text{obs}})
\right).
\label{eq:forecasting_loss_theory}
\end{align}
For i.i.d. samples
\(\{(\bX_{1:T}^{\text{obs},i},\bO_{1:T}^{i})\}_{i=1}^{n}\sim Q_{o,T}^{\btheta^*}\), let
\begin{align}
\calL_n(\btheta,\bphi)
&\coloneq
\frac{1}{n}\sum_{i=1}^{n}
\ell_{\btheta,\bphi}(\wbX_{1:T}^{\text{obs},i}),
\qquad
\calL(\btheta,\bphi)
\coloneq
\bbE_{\btheta^*}
\left[\ell_{\btheta,\bphi}(\wbX_{1:T}^{\text{obs}})\right].
\end{align}
{As discussed in Section~\ref{subsec:joint-training}, the encoder is calibrated
by the synthetic self-consistency loss \(\calL_{\mathrm{syn},\enc}\) in
\eqref{def:L_syn_enc}. For a candidate \(\btheta\), let
\(\calL_{\mathrm{syn},\enc}^{\btheta}(\bphi)\) denote this loss when the
synthetic data are generated under \(\btheta\), and define
$\tbPhi(\btheta)
\coloneq
\argmin_{\bphi\in\Phi}
\calL_{\mathrm{syn},\enc}^{\btheta}(\bphi).$
The aligned class and empirical estimator are $\calA \coloneq
\left\{(\btheta, \bphi)\in \Theta \times \Phi : \bphi \in \tbPhi(\btheta)\right\},$
and}
\begin{align}
(\hbtheta_n,\hbphi_n)
\in
\argmin_{(\btheta,\bphi)\in\calA}
\calL_n(\btheta,\bphi).
\label{def:em}
\end{align}
In practice, the
expectations in \(S^{\mathrm{mask}}\) are approximated using two independent draws from the 
one-step predictive law, yielding the unbiased estimator
\(\widehat{\calL}_{\mathrm{fore},n}\) of
\(\calL_n(\btheta,\bphi)\) in  \eqref{def:L_hat_fore}. The constraint
\(\bphi\in\tbPhi(\btheta)\) is empirically enforced by the calibration loss
\(\widehat{\calL}_{\mathrm{syn},\enc,B}\) in
Algorithm~\ref{alg:joint-training}.
Given the fitted parameter \(\hbtheta_n\), the full law
\(Q_T^{\hbtheta_n}\) is generated by the learned transition
\(\trans^{\hbtheta_n}\), output map \(f_x^{\hbtheta_n}\), and
logits \(\pi_1^{\hbtheta_n},\ldots,\pi_p^{\hbtheta_n}\), together with the known
initial-state and noise laws.
We formulate the consistency results in terms of Wasserstein distances. For the full-law
result, the underlying space is \((\calX\times\calO)^T\), with generic element
\(\wbx_{1:T}\), and metric $\rho(\wbx_{1:T},\wbx'_{1:T})
\coloneq
\sum_{t=1}^T
\left(\|\bx_t-\bx'_t\|_2^2 + \|\bo_t-\bo'_t\|_2^2\right)^{1/2}.$
For probability laws \(P,Q\) on this space with finite first moments, define
\[
W_1(P,Q)
=
\sup_{\varphi: \|\varphi\|_{\mathrm{L}}\le1}
\left|
\int \varphi\,dP-\int \varphi\,dQ
\right|,
\qquad 
\|\varphi\|_{\mathrm{L}}
\coloneq
\sup_{\wbx_{1:T}\neq\wbx'_{1:T}}
\frac{|\varphi(\wbx_{1:T})-\varphi(\wbx'_{1:T})|}
{\rho(\wbx_{1:T},\wbx'_{1:T})}.
\]

The following theorem shows the main conclusion of this section.
\begin{theorem}[Full-law consistency]
Consider the model given in~\eqref{eq:dgp_main}. Suppose Assumptions \ref{assum:encoder_consistency}--\ref{assum:static_base_main} hold, as stated below. Then
\[
W_1\bigl(Q_T^{\hbtheta_n},Q_T^{\btheta^*}\bigr)
\overset{p}{\longrightarrow}0
\quad\text{as } n \to \infty.
\]
\label{thm:full_law_consistency_main}
\end{theorem}
Theorem~\ref{thm:full_law_consistency_main} shows that the learned structural
model consistently recovers the true full law \(Q_T^{\btheta^*}\), even though it
is trained only on partially observed trajectories. Recovering the full law is important because it enables the fitted model to
generate complete trajectories, supporting both imputation and downstream tasks
such as treatment effect estimation and subgroup analysis.

The proof sketch of Theorem~\ref{thm:full_law_consistency_main} proceeds through two main steps:
observation-law consistency
(Proposition~\ref{prop:partial_law_main}), and finite-horizon full-law
identifiability
(Lemma~\ref{lem:full_law_finiteT_main}).  
The remainder of this
section states the sufficient conditions for the theorem and verifies them in a
linear factor model with logistic missingness. 
\subsection{Conditions for the main theorem}\label{subsec:assumption}
This subsection collects Assumptions~\ref{assum:encoder_consistency}--\ref{assum:static_base_main}, which are used in
Theorem~\ref{thm:full_law_consistency_main}.
\begin{assumption}[Well-specified encoder class]
\label{assum:encoder_consistency}
For each \(\btheta\in\Theta\), there exists $\bphi \in \Phi$ such that $\enc^{\bphi}$ represents the
filtering distributions as stated in Proposition~\ref{prop:existence_enc}. That is, for \(t=1,\ldots,T\),
\[
\enc^{\bphi}(\wbX_{1:t}^{\mathrm{obs}},\cdeps_t)
\sim
\law_{\btheta}(\bZ_{t}\mid\wbX_{1:t}^{\mathrm{obs}}).
\]
\end{assumption}
Assumption~\ref{assum:encoder_consistency} requires that, for every candidate
$\btheta\in\Theta$, the encoder class contain a representation of the filtering
distributions induced by that model. This is naturally supported by flexible sequence models: sufficiently
expressive Transformer classes can approximate general sequence-to-sequence
functions \citep{yun2019transformers}. 
\begin{assumption}[Model-encoder regularity]\label{assum:model_encoder_regularity}
Fix \(T<\infty\). Assume:
\begin{enumerate}
\item[$(i)$] \(\Theta\) and \(\Phi\) are compact parameter spaces for
fixed-architecture ReLU networks and standard Transformers, respectively.
\item[$(ii)$] The initial state and all noises have finite second moments.
\end{enumerate}
\end{assumption}
Compactness of \(\Theta\) and \(\Phi\) uniformly bounds the weights of the
fixed-architecture networks. ReLU networks and Transformers
with softmax attention are measurable, continuous, and have at most linear
growth in their inputs \citep{vaswani2017attention}. Together with the finite
second moments of the initial state and model noises, these properties imply,
for fixed \(T\), uniform square integrability of the model-generated and
encoder-assisted latent states and observations. This provides the tail control
needed for the masked energy score and the continuity and uniform convergence
arguments. The fixed architectures still retain substantial approximation
flexibility on compact domains \citep{yarotsky2017error,lu2021deep}.

Define the observed-law equivalence class of \(\btheta^*\) by
\begin{align}
    \Theta_{o,T}(\btheta^*) \coloneq \{\btheta\in\Theta: Q_{o,T}^{\btheta}=Q_{o,T}^{\btheta^*}\}.
\end{align}
For \(\bo\in\{0,1\}^p\), let $I(\bo):=\{j\in[p]:\bo_j=1\}$
be the observed-coordinate set, and let
\(\Pi_{\bo}:\bbR^p\to\bbR^{|I(\bo)|}\) be the coordinate projection onto
\(I(\bo)\). Define $f_x^{\btheta,\bo}:=\Pi_{\bo}\circ f_x^{\btheta},
$ and $
f_x^{\btheta^*,\bo}:=\Pi_{\bo}\circ f_x^{\btheta^*}.$
We denote the singleton-missing pattern collection $\widetilde{\calO}
=
\{\bo_{\setminus\{1\}},\dots,\bo_{\setminus\{p\}}\},$
where \(\bo_{\setminus\{j\}}\) is the pattern missing only the \(j\)-th coordinate.
\begin{assumption}[Regularity for full-law identification]
\label{assum:static_base_main}
Assume:
\begin{itemize}
\item[$(i)$] Given \(\bz\), mask coordinates are independent with  logits
\(\pi_j^{\btheta}:\bbR^d\to\bbR\). 

\item[$(ii)$] For every \(\btheta\in\Theta_{o,T}(\btheta^*)\), the full maps
\(f_x^{\btheta}\) and \(f_x^{\btheta^*}\) are continuous and one-to-one maps
from \(\bbR^d\) to \(\bbR^p\). Moreover, for every
\(\bo\in\widetilde{\calO}\), the true submap
\(f_x^{\btheta^*,\bo}\) is one-to-one.
\item[$(iii)$] The observation noises \(\br_t\sim\law(\br)\), \(t\in[T]\), are i.i.d.;
the law \(\law(\br)\) is known and has a characteristic function that is
non-vanishing on \(\bbR^p\).
\item[$(iv)$] The transition noises \(\{\bv_t\}_{t=1}^{T-1}\) are i.i.d. from a
known law \(\law(\bv)\).
\item[$(v)$] The law \(\law(\bZ_1)\) is known and has a density on \(\bbR^d\) that is
positive almost everywhere.
\item[$(vi)$] For every \(\btheta\in\Theta_{o,T}(\btheta^*)\) and
\(t=2,\ldots,T\), the law of \(\bZ_t\) is absolutely continuous with
respect to Lebesgue measure on \(\bbR^d\).
\end{itemize}
\end{assumption}
Assumption~\ref{assum:static_base_main} collects the regularity conditions on
the missingness mechanism, output map,
noise distributions, latent states used for finite-horizon full-law identification. 

Specifically, {the one-to-one requirement on \(f_x^{\btheta^*}\) means that the
complete measurement map preserves the latent information. When \(p\gg d\), sufficiently informative measurement maps can remain injective
even after one coordinate is removed.} {Furthermore, the common latent law fixes the latent coordinate system used to
compare two representations: if the latent variable is transformed by an
invertible map, the observed-data law can be represented after absorbing
this change into \(f_x\) and the missingness logits \(\pi_j\). Similarly, fixing common observation and transition noise  laws does not reduce the relevant model generality; as
discussed in Remark~\ref{remark:gaussian_noise}, general noise distributions can
be represented through transformations of reference Gaussian noises and absorbed
into the structural maps.}

Initial-state density positivity and subsequent absolute continuity allow the
static identification result to extend over a finite horizon. The latter
condition $(vi)$ holds for many nonlinear transitions, such as autoregressive models
with additive noise having a density.
\subsection{A linear factor model with logistic missingness}
\label{subsec:example}
We illustrate how the assumptions can be verified in a linear factor model with
logistic missingness, where they reduce to rank and Gaussian-noise
conditions. Fix \(T\ge2\) and consider
\begin{align}
\bZ_t &= \bA\bZ_{t-1}+\bv_t,
\qquad
\bv_t\sim\calN(\bzero,\bI_d),
\qquad t=2,\ldots,T, \notag\\
\bX_t &= \bF\bZ_t+\br_t,
\qquad
\br_t\sim\calN(\bzero,\bI_p),
\qquad t=1,\ldots,T, \notag\\
\bbP(\bO_{t,j}=1\mid\bZ_t=\bz)
&=
\sigma(a_j+\bw_j^\top\bz),
\qquad j=1,\ldots,p .
\label{model:linear_gaussian}
\end{align}
with $\bZ_1 \sim \calN(\bzero,\bI_d)$ and $\ba=(a_1,\ldots,a_p)$. The mask coordinates
are conditionally independent given \(\bZ_t\), and the noise variables
\(\{\br_t\}_{t=1}^T\), \(\{\bv_t\}_{t=2}^T\), and \(\bZ_1\) are mutually
independent. Let $\btheta=(\bA,\bF,\ba,\bw_1,\ldots,\bw_p),\ 
\btheta^*=(\bA^*,\bF^*,\ba^*,\bw_1^*,\ldots,\bw_p^*)$, 
and define
\[
\Theta
\coloneq
\left\{
\btheta:
\|\bA\|_F+\|\bF\|_F+\|\ba\|_2+\sum_{j=1}^p\|\bw_j\|_2\le C
\right\}.
\]
For a matrix \(\bF\in\bbR^{p\times d}\), let \(\bF_{-j}\) denote the submatrix
obtained by removing its \(j\)-th row. Let $O(d)\coloneq
\left\{\bR\in\bbR^{d\times d}:\bR^\top\bR=\bI_d\right\}$
denote the orthogonal group in dimension \(d\).
\begin{lemma}[Observed-law equivalence in the linear logistic model]
\label{lemma:linear_gaussian_identifiability}
Given model~\eqref{model:linear_gaussian}, suppose that $\rank(\bF^*_{-j})=d$ for $j\in [p]$.  If $Q_{o,T}^{\btheta}=Q_{o,T}^{\btheta^*}$, there exists $\bR\in O(d)$ such that
\[
    \bF=\bF^*\bR,
    \qquad \ba=\ba^*,
    \qquad \bw_j=\bR^\top \bw_j^*,\; j=1,\ldots,p,
    \qquad \bA=\bR^\top \bA^*\bR.
\]
In particular, $\bF$ and each $\bF_{-j}$ have full column rank.  Conversely, every parameter satisfying these identities has the same observed law as $\btheta^*$. Consequently, the equivalence class of \(\btheta^*\) is
\[
\Theta_{o,T}^{\mathrm{lin}}(\btheta^*)
=
\Theta\cap
\left\{
\btheta:
\exists\bR\in O(d),
\bF=\bF^*\bR,\ 
\ba=\ba^*,\
\bw_j=\bR^\top\bw_j^*\ \forall j,\ 
\bA=\bR^\top\bA^*\bR
\right\}.
\]
\end{lemma}

We now verify the assumptions used in the main results. 
Assumption~\ref{assum:encoder_consistency} is a well-specification condition on
the encoder class: it holds when the class is rich enough to represent the
filtering distributions induced by model~\eqref{model:linear_gaussian}.
Assumption~\ref{assum:model_encoder_regularity} follows from compactness of the
parameter spaces and the finite second moments of the Gaussian initial state
and all noise variables.

By construction, the mask coordinates are conditionally independent given
\(\bZ_t\), with logistic logits \(a_j+\bw_j^\top\bz\). By
Lemma~\ref{lemma:linear_gaussian_identifiability}, every
\(\btheta\in\Theta_{o,T}^{\mathrm{lin}}(\btheta^*)\) has \(\bF\) full column
rank, and the assumed rank condition on \(\bF^*_{-j}\) gives the required
one-to-one single-missing submaps. The known Gaussian initial and noise laws
satisfy conditions $(iii)$--$(v)$, including the positive initial-state density and
the nowhere-vanishing characteristic function of the observation noise.
Finally, the additive Gaussian transition noise makes every subsequent
\(\bZ_t\) nondegenerate Gaussian, verifying condition $(vi)$. The proof of
Lemma~\ref{lemma:linear_gaussian_identifiability} is provided in
Appendix~\ref{pf:lemma:linear_gaussian_identifiability}.

\section{Empirical studies}\label{sec:experiments}

\subsection{Implementation details}\label{subsec:implementation}

\noindent
We parameterize the transition, output, and smoother using
multilayer perceptrons (MLPs). The encoder uses a causal attention layer  to
summarize the observed history, while the decoder uses variable-wise attention and a residual MLP to generate
completed observations.

All components are optimized jointly using the Adam optimizer implemented in
PyTorch, following Algorithm~\ref{alg:joint-training}. We conducted two preliminary sensitivity
analyses in Case~3 of Section~\ref{sec:simulation}. For brevity, the
corresponding figures and numerical results are reported in
Supplementary Section~\ref{supp:subsec:implementation}.

First, we examined the loss weights
\(\lambda\in\{1,2,3,4,5\}\) and
\(\lambda_h\in\{0,1,2,3,4,5\}\). Performance was relatively insensitive to
\(\lambda\) over the range considered. In addition, the synthetic
encoder-calibration loss produced no systematic improvement in the evaluated
performance metrics, while requiring synthetic trajectories to be generated
during training. We therefore set \(\lambda=5\) and \(\lambda_h=0\) in all
subsequent simulation and real-data experiments. When \(\lambda_h=0\), the
synthetic encoder-calibration step in
Algorithm~\ref{alg:joint-training} is omitted.

Second, we treat the fitted latent-state dimension, denoted by
\(d_{\mathrm{fit}}\), as a fixed architectural hyperparameter. We examined sensitivity to
\(d_{\mathrm{fit}}\) in Case~3 under the LDM mechanism, with a
data-generating latent dimension of \(d=15\), \(50\%\) missingness, and
\(p\in\{30,50,100\}\). Across the settings considered, there was no
systematic deterioration when \(d_{\mathrm{fit}} \geq d\);
larger fitted dimensions gave comparable or better performance overall. This
empirical robustness is consistent with recent results for
latent variable models, suggesting that a working latent dimension modestly
above the underlying dimension can be preferable to smaller choices by
providing greater representational flexibility under finite network capacity
\citep{xie2026identifiable,wang2025deep}. 
We therefore use \(d_{\mathrm{fit}}=p\) as a simple default throughout the
empirical studies.

Complete architectural and optimization details, including layers, hidden dimensions, 
batch size, and learning rate, are provided in the code repository.  We developed a Python package named ``lifts'', available at
\url{https://pypi.org/project/lifts/}. More examples and source code are available
at \url{https://github.com/sanyouwu/lifts-open-source/tree/main}.
\subsection{Simulation studies}\label{sec:simulation}
In this section, we conduct comprehensive simulation studies to evaluate the forecasting and imputation performance of our proposed LIFTS framework. Our experiments encompass a variety of data-generating processes, ranging from baseline linear models to highly nonlinear dynamics, evaluated under a partially observed scenario.

\noindent {\textbf{Case 1}. }(Linear Gaussian.)
This serves as a baseline linear setting. The data generation process is defined as:
\begin{align} \label{eq:case-1}
    \bX_t &= \bF \bZ_t + 0.3 \cdot \br_t, & \br_t &\sim \calN(\bzero,  \bSigma_p) \cr
    \bZ_t &= \bphi \bZ_{t-1} + 0.1 \cdot \bv_t, & \bv_t &\sim \calN(\bzero,  \bSigma_{d})
\end{align}
\noindent {\textbf{Case 2}.} (Generalized stochastic volatility with nonlinear dynamics).  We extend the framework of \citet{chib2006analysis} by introducing a nonlinear transition function to capture saturated state evolution. This setting challenges the model with both heavy-tailed multiplicative observations and nonlinear latent dynamics:
\begin{align} \label{eq:case-3}
    \bX_t &= \bF \left( \exp\bigl(\bZ_t / 2\bigr) \odot \ \bigl(0.3 \cdot \deta_t\bigr) \right) + \bb , & \deta_t &\sim t_{10}(\bzero, \bSigma_d) \cr
    \bZ_{t,j} &= \tanh(0.5 \cdot \bZ_{t-1,j} + 0.5)  + 0.1 \cdot \bv_{t,j}, & \bv_{t,j}  &\sim \calN(0,1), \quad   j = 1,\dots, d
\end{align}
\noindent {\textbf{Case 3}.} (Third-order latent Markov dynamics).  We extend Case 2 as follows:
\begin{align} \label{eq:case-4}
    \bX_t &= \bF_1  \bZ_t+ \bF_2 \left( \exp\bigl(\bZ_t / 2\bigr) \odot \! \bigl(0.3 \cdot \deta_{t,1}\bigr) \right) + 0.3 \cdot \deta_{t,2}  ,  \notag \\
    \bZ_{t} &= \bphi_1 \bZ_{t-1} +  \bphi_2 \bZ_{t-2} + \bphi_3 \bZ_{t-3}  +  0.1 \cdot \bv_{t}, \qquad   \bv_{t}  \sim \calN(\bzero, \bSigma_d)
\end{align}
where $\deta_{t,1} \sim t_{10}(\bzero, \bSigma_d), \ \deta_{t,2} \sim t_{10}(\bzero, \bSigma_p)$.
Across all simulation cases, the entries of the output matrices   $\bF$, $\bF_j, j= 1,2$, $\bphi$, and transition matrices $\bphi_{j}, j =1,2,3$ are generated independently from a standard normal distribution. To control the system's dynamics, we scale all elements of the transition matrices $\bphi$, and $\bphi_{j}, j =1,2,3$ uniformly by a factor of $1/\sqrt{d}$. The covariances $\bSigma_p, \bSigma_d$ are Toeplitz matrices with correlation decaying linearly from $0.5$ to $0$ based on index distance. Additionally, for Case 2, the biased term $\bb$  is sampled element-wise from a uniform $[-1,1]$ distribution.  

While retaining the underlying data-generating processes defined in Cases 1--3, we introduce missing values to assess the model's capacity for simultaneous forecasting and imputation. To simulate realistic missingness patterns, we generate the binary observation mask $\bO_{1:T}$ under two distinct missingness mechanisms: Missing Completely at Random (MCAR) and Latent-Driven Missingness (LDM).
In the MCAR setting, we assume that the availability of each entry is independent of the underlying process. The mask entries are sampled independently from a Bernoulli distribution with an observation rate $s$ (defined as the ratio of observed points to the total number of points), where we examine performance across varying observation rates $s \in \{0.2, 0.5, 0.8\}$. To rigorously test the model's ability to infer dynamics under structural missingness, we further consider the LDM scenario, where the observation probability is modulated by the latent state itself. Specifically, the mask at time $t$ is generated according to $    \bO_t \sim \operatorname{Bernoulli}\left(\alpha + (1- \alpha)\cdot \sigma\bigl(\bW \bZ_t + \bb\bigr)\right),$
where $\bW \in \bbR^{p \times d}$ and $\bb \in \bbR^{p}$ parameterize the dependence, and $\sigma(\cdot)$ denotes the element-wise sigmoid function. We calibrate the intercept $\bb$ such that the resulting empirical observation rate matches the target levels $\{0.2,0.5,0.8\}$.  The offset of $\alpha = 0.1$ ensures that the observation probability for any component is uniformly bounded away from zero.

To comprehensively assess the imputation performance, we employ standard metrics. For a given test set containing $n_{test}$ independent time series observed over $T$ time steps with $p$ features, let $\bO_{t,j}^{i} \in \{0, 1\}$ denote the binary evaluation mask, which equals $1$ if the entry is withheld for validation and $0$ otherwise. For probabilistic evaluation, we draw $B=200$ independent samples $\hbX_{t}^{i,b}$ from the learned model to approximate the predictive distribution.

First, to evaluate point estimation accuracy, we compute the normalized mean absolute error (nMAE) over all masked entries:
\bel{mae}
\text{nMAE} = \frac{\sum_{i,j,t} \bO_{t,j}^{i} \left| \widehat{\widebar{\bM}}_{t,j}^{i} - \bX_{t,j}^{i} \right|}{\sum_{i,j,t} \bO_{t,j}^{i}|\bX_{t,j}^{i}|} 
\eel
where $\widehat{\widebar{\bM}}_{t,j}^{i} \coloneq \frac{1}{B}\sum_{b=1}^{B}\hbX_{t,j}^{i,b}$ is the empirical mean of the imputed samples, and $\bX_{t,j}^{i}$ is the ground truth.
Second, to assess the calibration of the univariate marginals, we calculate the empirical coverage probability at a nominal confidence level $1-\alpha$:
\bel{cover}
\operatorname{Cover}(1-\alpha) = \frac{\sum_{i,j,t} \bO_{t,j}^{i} \bone\left\{ \bX_{t,j}^{i} \in \hbI_{t,j}^{i}(1-\alpha) \right\}}{\sum_{i,j,t} \bO_{t,j}^{i}} 
\eel
where $\hbI_{t,j}^{i}(1-\alpha)$ denotes the empirical prediction interval constructed from the $B$ generated samples.
Third, we utilize the Continuous Ranked Probability Score (CRPS, \citealp{matheson1976scoring}) to evaluate the full univariate predictive distribution. Using a quantile-based approximation over a uniform grid $\mathcal{Q}=\{0.05, 0.10, \dots, 0.95\}$, the normalized CRPS (nCRPS) is given by:
\bel{crps}
 \text{nCRPS} = \frac{1}{|\mathcal{Q}| \sum_{i,j,t} \bO_{t,j}^{i}|\bX_{t,j}^{i}|} \sum_{q \in \mathcal{Q}} \sum_{i,j,t} \bO_{t,j}^{i} \cdot 2 \left( \hat{Q}_{t,j}^{i,q} - \bX_{t,j}^{i} \right) \left( \bone\left\{\bX_{t,j}^{i} \le \hat{Q}_{t,j}^{i,q}\right\} - q \right) \eel
where $\hat{Q}_{t,j}^{i,q}$ is the empirical $q$-th quantile of the samples $\{\hbX_{t,j}^{i,1}, \dots, \hbX_{t,j}^{i,B}\}$.
Note that for the imputation task, these univariate metrics (coverage, nCRPS) are evaluated exclusively over the missing entries where $\bO_{t,j}^i = 1$.

\begin{figure}[t]
    \centering
\includegraphics[width=1\linewidth]{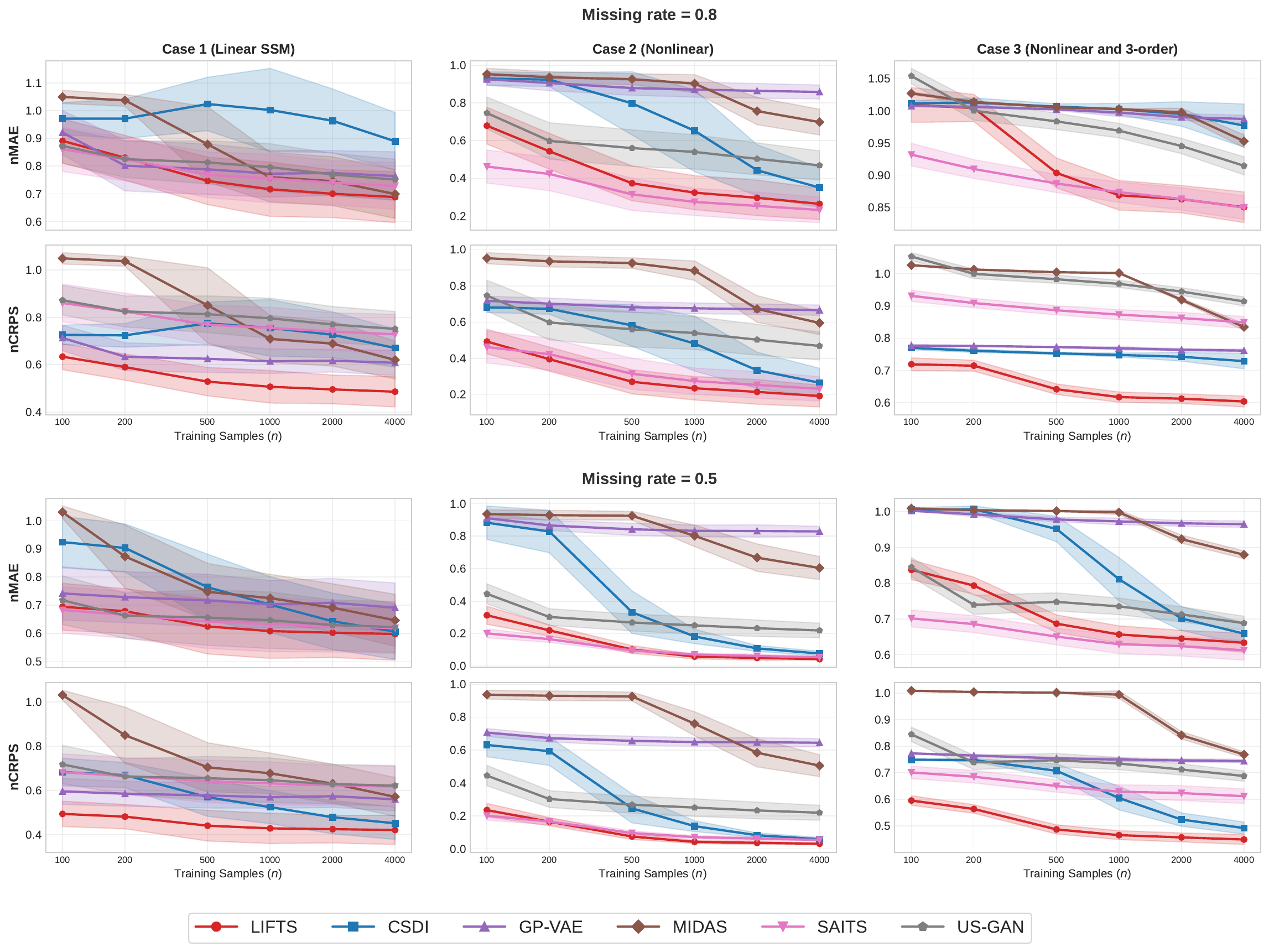}
    \caption{Performance across Cases 1--3 with increasing training sample size under LDM. The top two rows correspond to an $80\%$ missingness rate and the bottom two rows to a $50\%$ missingness rate; within each pair, the rows report nMAE and nCRPS. Lines and shaded bands denote the mean and one standard deviation over 10 random seeds.}
    \label{fig:sim-T10-p20-d5-LDM-0.2-0.5}
\end{figure}

For the quantitative comparison, we set $(d,p)=(5,20)$ and $T=10$, after
discarding a burn-in period of length $T_{\text{burn}}=100$. We consider
$n_{\text{train}}\in\{100,200,500,1000,2000,4000\}$, use a fixed test sample of
size $n_{\text{test}}=200$, and report averages over 10 independent
replications.

We compare LIFTS with two single-imputation methods, SAITS
\citep{du2023saits} and US-GAN \citep{miao2021generative}, and three methods
that produce multiple imputations, MIDAS \citep{gondara2018mida}, CSDI
\citep{tashiro2021csdi}, and GP-VAE \citep{fortuin2020gp}.

Figure~\ref{fig:sim-T10-p20-d5-LDM-0.2-0.5} compares performance
under the LDM mechanism at missingness rates of \(80\%\) and \(50\%\).
Errors are generally higher under the more severe missingness rate, and
the benefit of additional training data varies across methods and
data-generating mechanisms. The single-imputation methods remain
competitive in nMAE; in particular, SAITS attains the lowest point errors
in several nonlinear settings, but neither SAITS nor US-GAN provides a
predictive distribution. Among the methods that produce multiple
imputations, LIFTS achieves the lowest nCRPS in all three cases and at
both missingness rates, with its advantage already evident at the
smallest training sample sizes.

CSDI exhibits more case-dependent gains from increasing the training
sample size. Under \(80\%\) missingness, its errors decrease only modestly
in Cases~1 and~3, whereas larger improvements are observed in Case~2 and
under \(50\%\) missingness. This contrast can be understood from the
methods' distinct inductive structures. CSDI estimates a conditional score
over the full \(p\times T\) time--feature grid, whereas LIFTS uses a
state-space factorization to pool information sequentially through shared
transition and output maps. Under severe sparsity, this structured
factorization reduces the effective complexity of the learning problem and
thereby improves sample efficiency, particularly when the data-generating
process has a latent state-space structure, as in Cases~1 and~3. GP-VAE
likewise shows limited improvement with increasing training sample size in
most settings.

Furthermore, as shown in Figure \ref{fig:sim-coverage-case1}, LIFTS achieves empirical coverage close to the nominal levels while maintaining relatively narrow prediction intervals. Results under MCAR exhibit similar patterns; additional figures are provided in Supplementary Section \ref{supp:subsec:simulation}.
\begin{figure}
    \centering
    \includegraphics[width=1\linewidth]{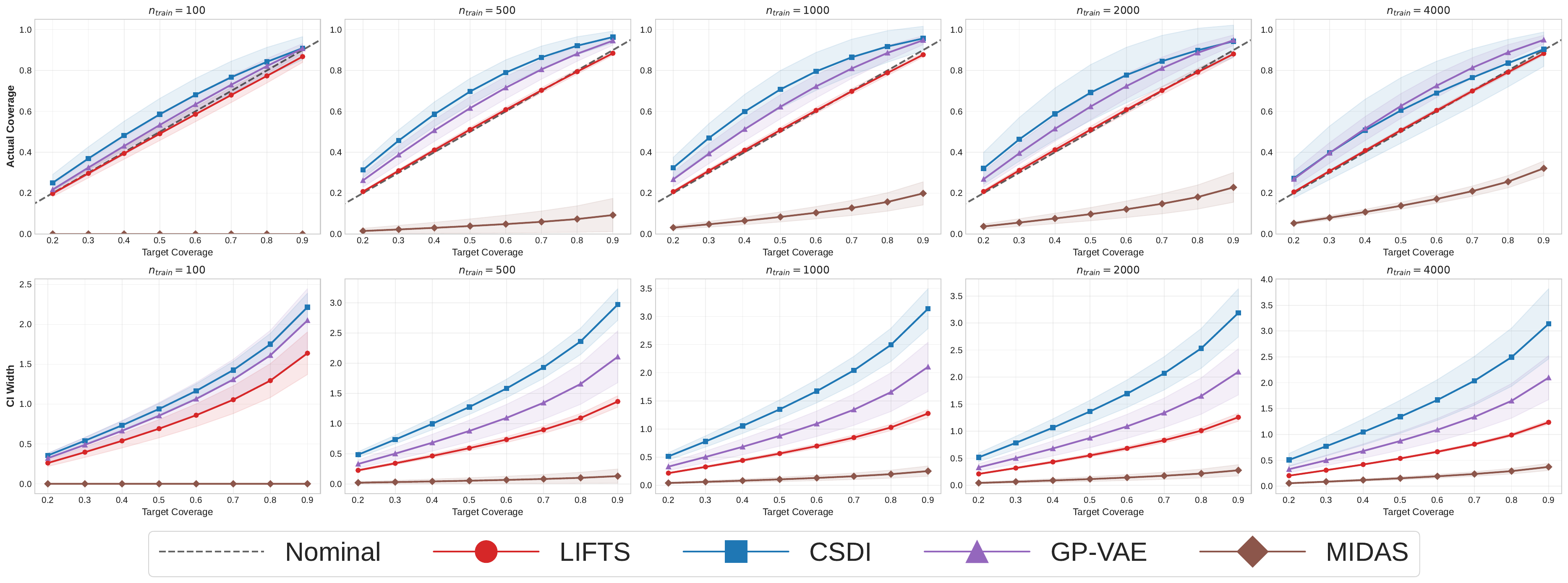}
    \caption{Empirical coverage and corresponding average prediction-interval width under Case 1 using the LDM mechanism with an $80\%$ missingness rate.}
    \label{fig:sim-coverage-case1}
\end{figure}

\section{Real data application: PhysioNet 2019}\label{sec:real-data}

\noindent
In this section, we evaluate the proposed framework using clinical data
from the PhysioNet Challenge 2019 \citep{reyna2020early}. The dataset consists
of partially observed multivariate time series from intensive care unit (ICU)
patients. Although originally designed for early sepsis prediction, it provides
a useful benchmark for missing-data imputation because
its measurements are extremely sparse.

The raw data originate from three electronic medical record systems and are
publicly available at
\href{https://physionet.org/content/challenge-2019/1.0.0/}{https://physionet.org/content/challenge-2019/1.0.0/},
where detailed variable descriptions are provided. Following the experimental
setup of \citet{du2024tsi}, we use subset A, which contains \(20,336\) patient
records. For each patient, \(40\) features were initially recorded, including
vital signs (e.g., heart rate, temperature), laboratory values (e.g., lactate,
blood urea nitrogen), clinical assessments (e.g., Glasgow Coma Scale), and
demographics (e.g., age, gender). For our analysis, we retain \(p=33\)
variables, excluding static demographic attributes and the platelet-count variable due
to near-complete missingness.

To ensure a uniform temporal structure, we select only records spanning at least $48$ hours and truncate longer sequences at $48$ hours. This filtering yields an analytic sample of $4,927$ patients with an average missing rate of $78.4\%$. Figure \ref{fig:physionet_missing_view} illustrates the missingness patterns for three randomly selected patients.
\begin{figure}[t]
    \centering
\includegraphics[width=1\linewidth]{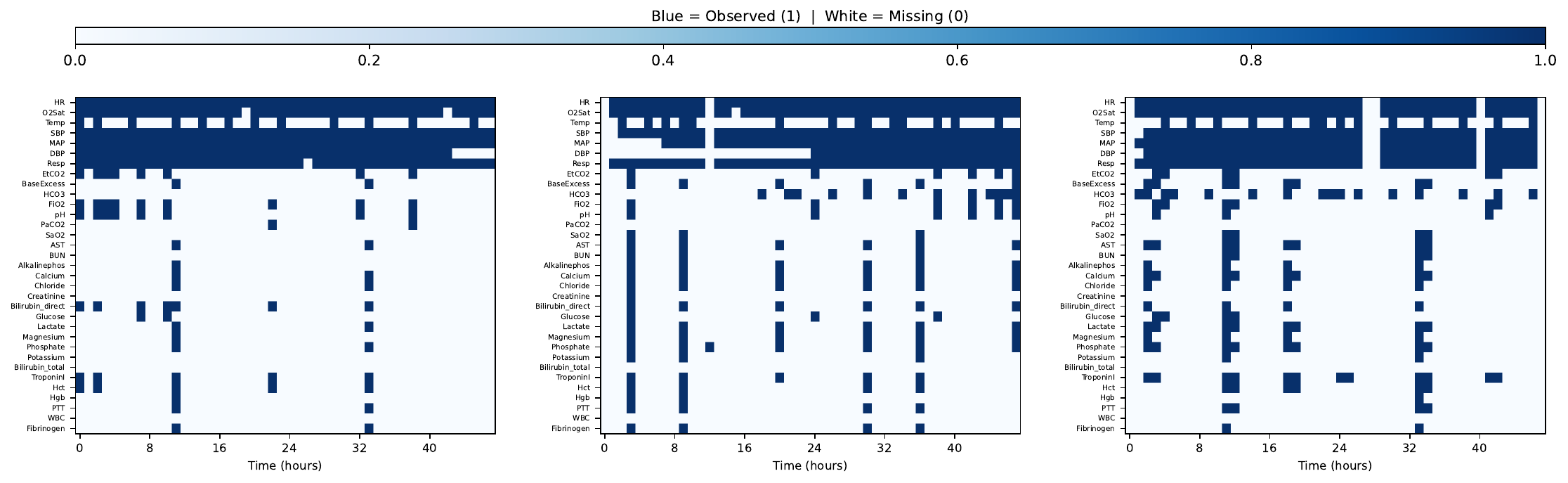}
\caption{Observation patterns for three randomly selected patients from the PhysioNet 2019 dataset.}
    \label{fig:physionet_missing_view}
\end{figure}
Since the complete trajectory \(\bX_t\in\bbR^{33}, t\in[48]\), is not observed
in real-world data, forecasting accuracy is assessed only on visible entries
(where \((\bO_t)_j=1\)). To evaluate imputation, we use a self-masking validation
scheme. Specifically, we construct a binary sub-mask \(\tbO_t\) by randomly
sampling \(80\%\) of the observed entries in the original mask \(\bO_t\). The
model infers values from the reduced observation
\((\bX_t^{\text{obs}}\odot\tbO_t,\tbO_t)\), while the entries indicated by
\(\bO_t-\tbO_t\) serve as the validation set.

To assess the robustness of the empirical comparison, we conduct $10$ independent experimental runs. In each repetition, the dataset is randomly partitioned into a training set of $4,000$ samples and a test set of $927$ samples. Moreover, to evaluate performance under varying degrees of  data scarcity, we train the methods on randomly sampled
subsets with sizes $n_{\text{train}} \in \{100, 200, 500, 1000, 2000, 4000\}$. 

We compare LIFTS with representative single-imputation methods (SAITS and
US-GAN) and multiple-imputation methods (MIDAS, GP-VAE, and CSDI). We use normalized mean
absolute error (nMAE) to evaluate point estimation accuracy and the normalized continuous
ranked probability score (nCRPS) to evaluate distributional prediction as in Section~\ref{sec:simulation}. For
uncertainty quantification, we report empirical coverage and average interval
width at nominal coverage levels
\(1-\alpha\in\{0.2,0.3,\ldots,0.9\}\).

\begin{figure}[t]
    \centering
\includegraphics[width=1\linewidth]{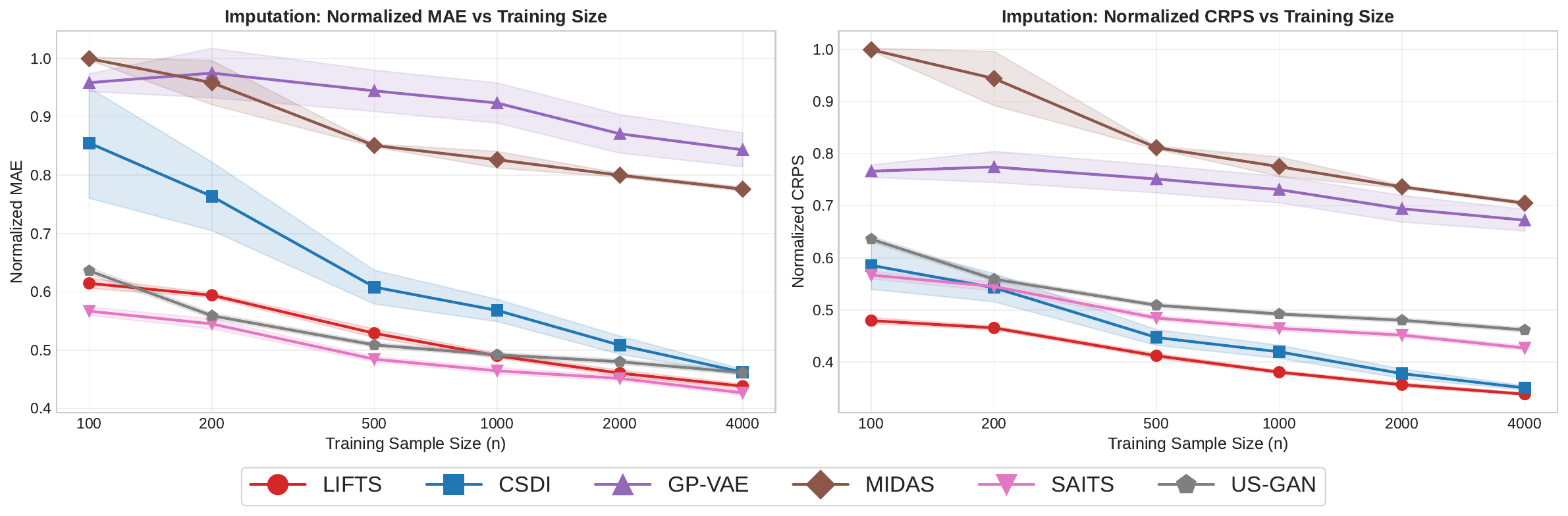}
\caption{Imputation performance across training sample sizes on the PhysioNet 2019 dataset.}
    \label{fig:real-sample_eff_all_methods}
\end{figure}
\begin{figure}[t]
    \centering
    \includegraphics[width=1\linewidth]{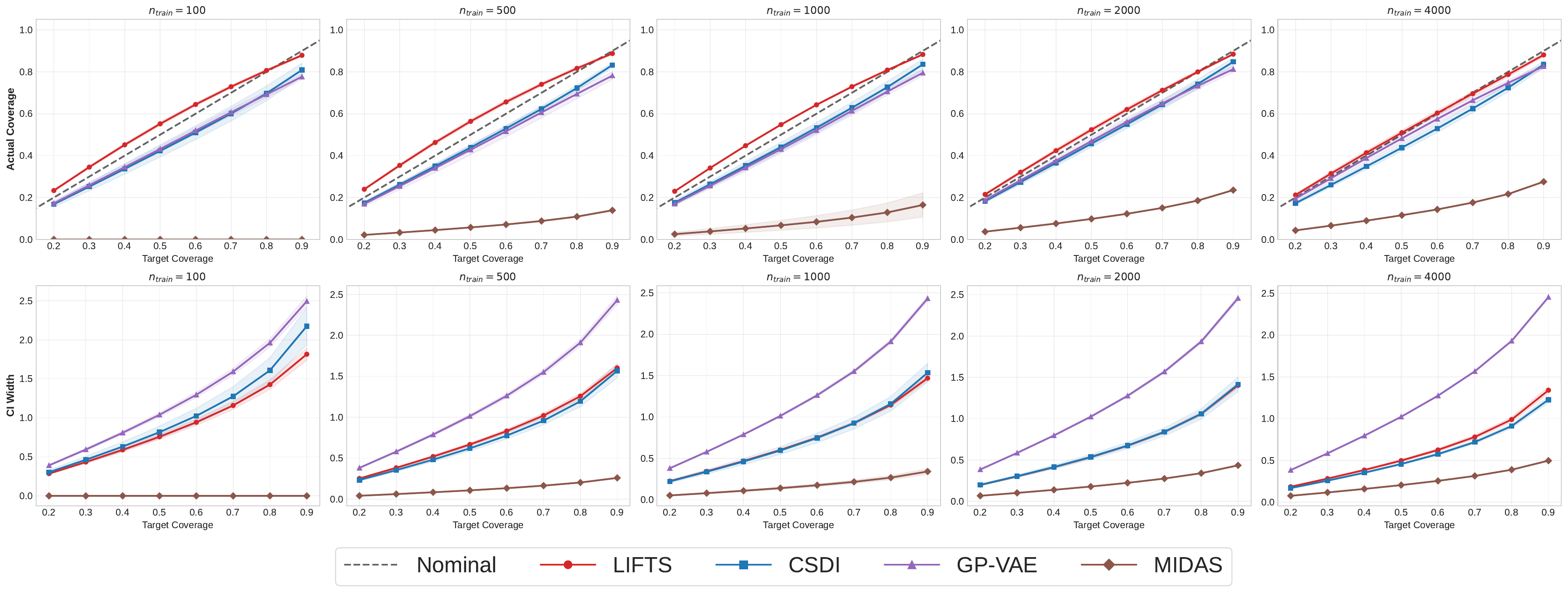}
    \caption{Empirical coverage and corresponding average interval width for four multiple-imputation methods.}
    \label{fig:real-coverage}
\end{figure}

For imputation, the mean and standard deviation over 10 independent runs are shown in Figures \ref{fig:real-sample_eff_all_methods} and \ref{fig:real-coverage}. Since SAITS \citep{du2023saits} and US-GAN \citep{miao2021generative} are single-imputation methods and do not provide uncertainty quantification, they are excluded from Figure \ref{fig:real-coverage}. For point estimation, SAITS attains the lowest nMAE, while LIFTS achieves comparable performance. For nCRPS, LIFTS shows superior sample efficiency and consistently outperforms the competing methods, particularly when the training sample size is small.

The performance of GP-VAE does not improve consistently as the sample size increases. This may reflect the difficulty of learning latent trajectories with a Gaussian-process prior under dynamic masking, which induces heterogeneous observation patterns across variables and subjects in complex physiological time series. CSDI, the second-best method, exhibits more stable improvement with increasing sample size, but remains less accurate than LIFTS across all sample sizes, especially when the training sample size is small. MIDAS performs poorly in this real-data setting, possibly because its denoising-autoencoder architecture is less effective at capturing complex temporal dependence and informative missingness patterns in clinical data.

As shown in Figure \ref{fig:real-coverage}, the empirical coverage of LIFTS is close to the nominal levels $1-\alpha$ across the evaluated coverage targets, while its average interval width remains comparable to that of CSDI. By contrast, GP-VAE produces substantially wider intervals, suggesting over-conservative uncertainty estimates. MIDAS exhibits under-dispersion, leading to low empirical coverage and narrow intervals. 

For forecasting, we conduct a rolling one-step-ahead evaluation over \(t=37,\ldots,48\). Given the observed history through time \(t-1\), LIFTS generates 200 paired draws \(\{(\widehat{\bX}_{t}^{(b)},\widehat{\bO}_{t}^{(b)})\}_{b=1}^{200}\) from the joint predictive distribution. For each observed test coordinate \(j\), we retain the measurement draws for which \((\widehat{\bO}_{t}^{(b)})_j=1\). We then compute nMAE using the mean of the retained draws and nCRPS using their empirical distribution.

The competing methods are designed primarily for imputation and do not
generate future observation masks. We adapt them to forecasting by setting all
entries from time \(t\) onward to missing, applying the fitted imputation model,
and using its output at time \(t\) as the one-step-ahead forecast. This
procedure is repeated for each forecast origin \(t-1=36,\ldots,47\). By contrast, LIFTS directly produces the one-step predictive
distribution at each forecast origin and jointly forecasts measurements
and masks, without repeatedly imputing artificially masked suffixes. As shown in Figure~\ref{fig:real-forecasting},
LIFTS achieves the lowest forecasting errors across all training sample sizes,
with the largest gains when training data are limited.
\begin{figure}[t]
    \centering
    \includegraphics[width=1\linewidth]{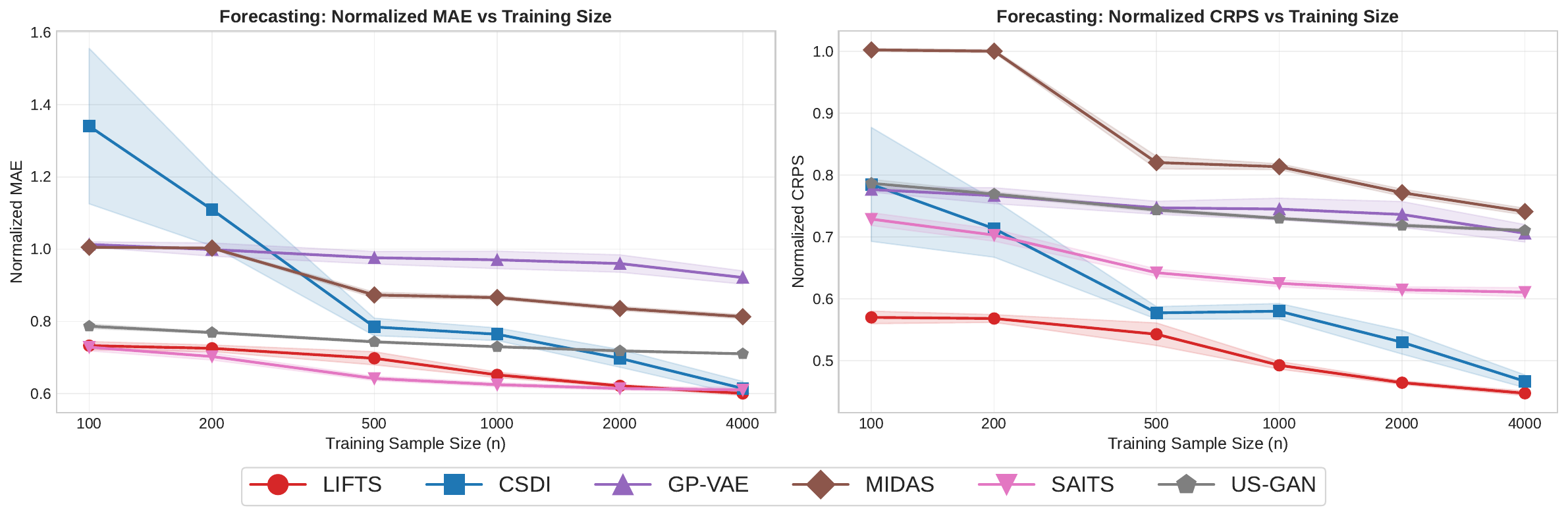}
    \caption{Mean and standard deviation of rolling one-step-ahead forecasting performance on the test set over $t=37,\ldots,48$, across different training sample sizes.} 
    \label{fig:real-forecasting}
\end{figure}

Computational results are reported in Supplementary Figure
\ref{fig:real-time}. Under a common hardware configuration, LIFTS requires
training time comparable to that of CSDI but is approximately 100 times faster
at inference. Together, these results provide evidence that LIFTS is accurate,
sample-efficient, and computationally efficient for partially observed ICU
time series.

\section{Discussion}\label{sec:conclusion}

\noindent
We have developed LIFTS, a generative framework for forecasting and
probabilistic imputation in partially observed multivariate time series
with informative missingness. LIFTS jointly models measurement trajectories and observation masks through a latent dynamic process, with forward-backward inference for filtering, forecasting, smoothing, and imputation. We establish full-law identification and consistency for an explicit structural subclass under the stated assumptions. Simulations and the real data analysis on PhysioNet 2019 demonstrate strong
distributional performance, particularly when training data are limited.

LIFTS can be readily extended to mixed-type longitudinal measurements. The masked energy score applies directly to counts and, after suitable encoding, to categorical variables; the measurement output function, parameterized by neural networks, would need to be adapted to generate discrete-valued draws. Beyond forecasting and imputation, joint draws of imputed trajectories could support downstream analysis, including subgroup comparisons and longitudinal causal inference.  Such analyses would need to account for imputation uncertainty, and causal conclusions would require additional assumptions beyond those established here.

{A central modelling assumption in LIFTS is that measurements and masks
are conditionally independent given the latent state. A natural extension is to allow both the measurement and missingness models to depend on observed static or time-varying covariates that may not be fully captured by the latent state. In clinical applications, these covariates may capture contextual factors such as insurance coverage, prior healthcare utilization, and care setting. The extended model would assume
conditional independence given both the latent state and these covariates. Future work could assess how violations of this assumption affect forecasting and imputation.
}

\section*{Acknowledgements}
\noindent
The research is supported by a National Science Foundation grant DMS-2544147, a National Institutes of Health grant 5R01AG087496, and the Columbia Data Science Institute Seed Fund (SF-120).
We thank Professor Long Feng, Sanyou Wu's PhD advisor at the University of Hong Kong, for his support of Sanyou's doctoral training and for helpful discussions and suggestions. Sanyou conducted part of this research during his visit to Columbia University, supported by Professor Kaizheng Wang's lab, while completing his PhD studies at the University of Hong Kong.

\begingroup
% \setstretch{1.15}
\bibliography{references}
\endgroup
\clearpage
% !TEX root = LIFTS_arxiv.tex
% Original supplementary body is preserved verbatim below.
% Only standalone document setup and bibliography are handled by the combined main file.
% Preserve the supplementary notation and sectionwise theorem numbering.
\let\hat\widehat
\let\tilde\widetilde
\def\bm{\mathbold{m}}
\def\bbZ{{\overline{\bZ}}}
\def\bbz{{\overline{\bz}}}
\numberwithin{theorem}{section}
\setcounter{theorem}{0}
% Distinct hyperlink destinations when counters restart in the supplement.
\renewcommand{\theHsection}{supp.\arabic{section}}
\renewcommand{\theHequation}{supp.\arabic{equation}}
\renewcommand{\theHfigure}{supp.\arabic{figure}}
\renewcommand{\theHtable}{supp.\arabic{table}}
\renewcommand{\theHtheorem}{supp.\arabic{section}.\arabic{theorem}}

\newcommand{\supplementtitle}{Supplementary Material to ``Latent-Driven Imputation and Forecasting for Partially Observed Time Series''}

\begin{center}
\begin{minipage}{0.98\textwidth}
\centering
{\Large\bfseries \supplementtitle\par}
\end{minipage}
\end{center}
\vspace{0.5em}

%%%%%%%%%%%%%%%%%%%%%% reference
%%%%%%%%%%%%%%%%%%%%%% reference
%%%%%%%%%%%%%%%%%%%%%% reference
%%%%%%%%%%%%%%%%%%%%%% reference
%%%%%%%%%%%%%%%%%%%%%% reference
%%%%%%%%%%%%%%%%%%%%%% reference

%%%%%%%%%%%%%%% appendix  %%%%%%%%%%%%
%%%%%%%%%%%%%%% appendix  %%%%%%%%%%%%
%%%%%%%%%%%%%%% appendix  %%%%%%%%%%%%
%%%%%%%%%%%%%%% appendix  %%%%%%%%%%%%
%%%%%%%%%%%%%%% appendix  %%%%%%%%%%%%
%%%%%%%%%%%%%%% appendix  %%%%%%%%%%%%
%%%%%%%%%%%%%%% appendix  %%%%%%%%%%%%
%%%%%%%%%%%%%%% appendix  %%%%%%%%%%%%
%%%%%%%%%%%%%%% appendix  %%%%%%%%%%%%
%%%%%%%%%%%%%%% appendix  %%%%%%%%%%%%
%%%%%%%%%%%%%%% appendix  %%%%%%%%%%%%
%%%%%%%%%%%%%%% appendix  %%%%%%%%%%%%
%%%%%%%%%%%%%%% appendix  %%%%%%%%%%%%
% \section{Appendix}
% \section{Supplementary material}
% \noindent 
The supplementary material is organized as follows.
Section~\ref{supp:sec:key_properties} collects the conditional independence
properties and the extension to \(k\)-order Markovian dynamics.  Section~\ref{appen:proof} contains the proofs
of the main-paper results.  The proofs for
conditional independence properties are collected in
Sections~\ref{supp:sec:proof_S1}. Section~\ref{supp:sec:proof_S2} presents the supporting
results for full-law consistency.
Finally,
Section~\ref{sec:appendix:lemmas} collects auxiliary lemmas, and
Section~\ref{supp:sec:experiment} gives additional experimental details.
\counterwithout{equation}{section}
\setcounter{equation}{0}
\def\theequation{S\arabic{equation}}
\setcounter{figure}{0}
\def\thefigure{S\arabic{figure}}
\setcounter{table}{0}
\def\thetable{S\arabic{table}}
\setcounter{section}{0}
\def\thesection{S\arabic{section}}
%%%%%%%%%%%%%%%%%%--------
\section{Key properties and model extensions}\label{supp:sec:key_properties}
\subsection{Conditional independence properties}
Before detailing the inference steps, we establish key conditional independence properties inherent to model \eqref{model:1}. These properties serve as the foundation for our forward--backward framework, enabling the decomposition of the relevant conditional laws into tractable recursive updates.
\begin{lemma}[Conditional Independence Properties]\label{lemma:fact-1}
Suppose model~\eqref{model:1} holds. Then the following properties hold:
\begin{list}{(\roman{enumi})}{%
\usecounter{enumi}%
\setlength{\leftmargin}{2.4em}%
\setlength{\labelwidth}{1.8em}%
\setlength{\labelsep}{0.6em}%
\setlength{\itemsep}{0.25em}%
\setlength{\parsep}{0pt}%
}
    \item \textit{(Forward Markov Property)} $
p\left(\bZ_{t+1} \mid \bZ_{1:t}, \wbX_{1:t}^{\text{obs}} \right) = p\left(\bZ_{t+1} \mid \bZ_{t} \right)
$
    \item \textit{(Markovian for Forecasting)} $
 p(\wbX_{t+1:T} \mid \bZ_{1:t+1}, \wbX_{1:t}^{\text{obs}}) = p(\wbX_{t+1:T} \mid \bZ_{t+1})$
    \item \textit{(Backward Markov Property)} $ p(\bZ_t \mid \bZ_{t+1:T}, \wbX_{1:T}^{\text{obs}})  = p(\bZ_t \mid \bZ_{t+1}, \wbX_{1:t}^{\text{obs}}) $
    \item \textit{(Memoryless for Imputation)} $
p(\bX_t \mid  \bZ_{1:T}, \wbX_{1:T}^{\text{obs}}) =  p(\bX_t \mid \bZ_t,\wbX_t^{\text{obs}})
    $
\end{list}
\end{lemma}
These properties are instrumental to our inference framework: $(i)$ and $(ii)$ show that, conditional on the current latent state, the future conditional law no longer depends on the earlier latent and observed history. $(iii)$ forms the cornerstone of the backward pass, ensuring that smoothing updates only depend on the immediate future state. Finally, $(iv)$ characterizes the observation mechanism as a {memoryless channel}: it establishes that $\bX_t$ depends solely on the current state $\bZ_t$  and partial observation $\wbX_t^{\text{obs}}$, thereby decoupling the global imputation into time-local conditional operations.  The proof of Lemma \ref{lemma:fact-1} is provided in Appendix \ref{pf:lemma:fact-1}.

\subsection{Extension to $k$-order markovian dynamics}\label{sec:k-order}
The first-order state-space model introduced in Section~\ref{sec:set-up} provides a convenient formulation for latent temporal dynamics under informative missingness. In some applications, however, a first-order transition may be too restrictive, since the latent state at time $t$ may depend on several previous states rather than only on $\bZ_{t-1}$. We therefore extend model~\eqref{model:1} to a $k$-order Markovian formulation. Formally, the model is given by
\bel{model:k-order}
\bZ_t &=& \trans(\bZ_{(t-k)_{+}:t-1}, \bnu_t) \quad \bigl(\text{transition model}\bigr); \cr
\bX_t &=& \out_x(\bZ_t,\br_t^x) \quad \bigl(\text{measurement model}\bigr); \cr
\bO_t &=& \out_o(\bZ_t,\br_t^o) \quad \bigl(\text{missingness model}\bigr).
\eel
where $(t-k)_{+} = \max(1, t-k)$. The forward--backward framework remains applicable, although the conditional independence structure and the associated recursive representations must be modified accordingly. Thus the extension modifies only the latent transition, while the measurement and missingness mechanisms remain conditionally memoryless given the current latent state.
\begin{lemma}[Conditional Independence Properties for $k$-order Markovian dynamics]\label{lemma:fact-2}
Suppose model~\eqref{model:k-order} holds. Then the following properties hold:
\begin{list}{(\roman{enumi})}{%
\usecounter{enumi}%
\setlength{\leftmargin}{2.4em}%
\setlength{\labelwidth}{1.8em}%
\setlength{\labelsep}{0.6em}%
\setlength{\itemsep}{0.25em}%
\setlength{\parsep}{0pt}%
}
    \item \textit{(Forward Markov Property)} $
p\left(\bZ_{t+1} \mid \bZ_{1:t}, \wbX_{1:t}^{\text{obs}} \right) = p\left(\bZ_{t+1} \mid \bZ_{(t-k+1)_+:t} \right)
$
    \item \textit{(Markovian for Forecasting)} $
 p(\wbX_{t+1:T} \mid \bZ_{1:t+1}, \wbX_{1:t}^{\text{obs}}) = p(\wbX_{t+1:T} \mid \bZ_{(t+2-k)_+:t+1})$
    \item \textit{(Backward Markov Property)} $ p(\bZ_t \mid \bZ_{t+1:T}, \wbX_{1:T}^{\text{obs}})  = p(\bZ_t \mid \bZ_{t+1:(t+k)\wedge T}, \wbX_{1:t}^{\text{obs}}) $
    \item \textit{(Memoryless for Imputation)} $
p(\bX_t \mid  \bZ_{1:T}, \wbX_{1:T}^{\text{obs}}) =  p(\bX_t \mid \bZ_t,\wbX_t^{\text{obs}})
    $
\end{list}
\end{lemma}
Recall that $(t+1-k)_{+} \coloneq \max(t+1-k,1)$ and $(t+k)\wedge T \coloneq \min(t+k,T)$.
By Lemma~\ref{lemma:fact-2}~$(iii)$, for any \(\ell<t\),
\[
p(\bZ_\ell\mid\bZ_{\ell+1:t},\wbX_{1:t}^{\mathrm{obs}})
=
p(\bZ_\ell\mid\bZ_{\ell+1:(\ell+k)\wedge t},
\wbX_{1:\ell}^{\mathrm{obs}}).
\]
Thus, starting from
\(\bZ_t^s=\enc(\wbX_{1:t}^{\mathrm{obs}},\cdeps_t)\), the backward recursion
uses at most the next \(k\) latent states. The following theorem formalizes the
resulting smoothing and forecasting representations.
%%%%% k-order backward theorem
\begin{theorem}[Smoothing and forecasting under $k$-order dynamics]\label{thm:k-backward}\label{thm:k-forecasting}
Suppose model~\eqref{model:k-order} holds. For \(1<t\leq T\), set
\(\bZ_t^s\coloneq\enc(\wbX_{1:t}^{\mathrm{obs}},\cdeps_t)\). For
\(\ell\in[t-1]\), let \(b_\ell\coloneq(\ell+k)\wedge t\). There exists a
measurable function \(\backt\) such that
    \bel{eq:obj-backward-k-order}
    \bZ_\ell^{s} \coloneq \backt \left( \bZ_{\ell+1:b_\ell}^{s}, \wbX_{1:\ell}^{\text{obs}}, \tdxi_\ell\right), \quad \tdxi_\ell \sim \calN(\bzero, \bI_{d}), \ \ell \in [t-1].
    \eel
Consequently, the backward samples recover the joint smoothing law,
\bel{eq:obj-backward-k-order-joint}
\left(\bZ_1^{s}, \ldots, \bZ_t^{s} \right)  \deq  \left(\bZ_1, \ldots, \bZ_t \right) \mid \wbX_{1:t}^{\text{obs}}
\eel
Moreover, there exist measurable transition and output maps \(\btrans\) and
\(\bout\) such that
\bel{obj:forecast-K}
\left(\bX_{t+1},\bO_{t+1}\right) \mid \wbX_{1:t}^{\text{obs}}
\deq
\bout \left\{ \btrans \left( \bZ_{(t+1-k)_{+}:t}^{s},\deps_t \right),\deta_t\right\}.
\eel
All reference noises are mutually independent Gaussian vectors of the required
dimensions and are independent of \(\wbX_{1:t}^{\text{obs}}\).
\end{theorem}
The proof is given in Appendix~\ref{pf:thm:k-backward}, which establishes the
smoothing and forecasting representations together.
The result also highlights an important distinction from the first-order case.
For \(k>1\), the predictive transition depends on the block
\(\bZ_{(t+1-k)_+:t}\), so forecasting cannot be based solely on the filtering
state at time \(t\). Instead, the backward recursion is first used to recover
the required smoothed latent block \(\bZ_{(t+1-k)_+:t}^{s}\), which is then
passed to the forward transition. Thus the LIFTS procedure remains sequential
and local, but the local context for forecasting is a smoothed \(k\)-state
history rather than a single filtered state.

% \section{Lemmas}\label{appen:lemma}
% \section{Supporting results for full-law consistency}\label{supp:sec:identifiability_consistency}
\section{Proofs for the main results}\label{appen:proof}
\subsection{Proof of Proposition~\ref{prop:existence_enc}}\label{pf:prop:existence_enc}
\begin{proof}
For each \(t\in[T]\), the history space
\(\widebar{\calX}_{1:t}^{\mathrm{obs}}\) and the latent space \(\bbR^d\)
are standard Borel spaces. Hence a regular conditional distribution of
\(\bZ_t\) given \(\wbX_{1:t}^{\mathrm{obs}}\) exists. By the
noise-outsourcing lemma (Lemma~\ref{lemma:noise_outsourcing}), there is a
measurable map
\[
\enc_t:\widebar{\calX}_{1:t}^{\mathrm{obs}}\times\bbR^d
\longrightarrow\bbR^d
\]
such that, for almost every
\(\wbx_{1:t}^{\mathrm{obs}}\in
\widebar{\calX}_{1:t}^{\mathrm{obs}}\),
\[
\bZ_t\mid
\bigl\{\wbX_{1:t}^{\mathrm{obs}}=\wbx_{1:t}^{\mathrm{obs}}\bigr\}
\deq
\enc_t(\wbx_{1:t}^{\mathrm{obs}},\cdeps_t),
\qquad \cdeps_t\sim\calN(\bzero,\bI_d).
\]
Define \(\enc\) on the finite union of the history spaces by
\[
\enc(\wbx_{1:t}^{\mathrm{obs}},\cdeps_t)
\coloneq
\enc_t(\wbx_{1:t}^{\mathrm{obs}},\cdeps_t),
\qquad t\in[T].
\]
The history length identifies the relevant component of the union, so this
piecewise definition is measurable. Taking
\(\cdeps_1,\ldots,\cdeps_T\) mutually independent and independent of
\(\wbX_{1:T}^{\mathrm{obs}}\) gives \eqref{obj:encoder} simultaneously for all
\(t\in[T]\).
\end{proof}
\subsection{Proof of Corollary~\ref{thm:forecasting}}\label{pf:thm:forecasting}
\begin{proof}
Because the transition noises are identically distributed, the
noise-outsourcing lemma gives a measurable map
\(a_v:\bbR^d\to\bbR^d\) such that
\[
a_v(\deps_t)\deq\bv_t,
\qquad \deps_t\sim\calN(\bzero,\bI_d).
\]
Applying the same lemma jointly to the measurement and missingness noises gives
measurable maps \(a_x:\bbR^{2p}\to\bbR^p\) and
\(a_o:\bbR^{2p}\to\bbR^p\) such that
\[
\bigl(a_x(\deta_t),a_o(\deta_t)\bigr)
\deq
(\br_{t+1}^x,\br_{t+1}^o),
\qquad \deta_t\sim\calN(\bzero,\bI_{2p}).
\]
Define
\[
\btrans(\bz,\deps)
\coloneq
\trans\bigl(\bz,a_v(\deps)\bigr)
\]
and
\[
\bout(\bz,\boldsymbol\eta)
\coloneq
\left(
\out_x\bigl(\bz,a_x(\boldsymbol\eta)\bigr),
\out_o\bigl(\bz,a_o(\boldsymbol\eta)\bigr)
\right).
\]
These maps are measurable and have the domains and codomains stated in the
corollary.

Fix \(t\in[T]\) and an observed history
\(\wbx_{1:t}^{\mathrm{obs}}\) for which the representation in
Proposition~\ref{prop:existence_enc} holds. The encoder gives
\[
\enc(\wbx_{1:t}^{\mathrm{obs}},\cdeps_t)
\deq
\bZ_t\mid
\bigl\{\wbX_{1:t}^{\mathrm{obs}}=\wbx_{1:t}^{\mathrm{obs}}\bigr\}.
\]
Since \(\deps_t\) and \(\deta_t\) are independent of this encoded state and
of the observed history, the definitions of \(\btrans\) and \(\bout\), together
with model~\eqref{model:1}, yield
\[
\bout\!\left\{
\btrans\!\left(
\enc(\wbx_{1:t}^{\mathrm{obs}},\cdeps_t),\deps_t
\right),\deta_t
\right\}
\deq
(\bX_{t+1},\bO_{t+1})
\mid
\bigl\{\wbX_{1:t}^{\mathrm{obs}}=\wbx_{1:t}^{\mathrm{obs}}\bigr\}.
\]
\end{proof}

%%%%%%%%%%%%%%%%%%%%%%%%%%%%%%%%
%%%%%%%%%%%%%%%%
\subsection{Proof of Theorem~\ref{thm:imputation}}\label{pf:thm:imputation}
\begin{proof}
Fix an observed trajectory
$\wbx_{1:T}^{\mathrm{obs}}\in\widebar{\calX}_{1:T}^{\mathrm{obs}}$ for which
the required conditional representations hold.
For each \(t\in[T-1]\), Lemmas~\ref{lemma:fact-1}~$(iii)$ and
\ref{lemma:noise_outsourcing} give a measurable map \(\backt_t\) such that
\[
\backt_t(\bz_{t+1},\wbx_{1:t}^{\mathrm{obs}},\tdxi_t)
\deq
\bZ_t\mid
\bigl\{\bZ_{t+1}=\bz_{t+1},
\wbX_{1:t}^{\mathrm{obs}}=\wbx_{1:t}^{\mathrm{obs}}\bigr\},
\qquad
\tdxi_t\sim\calN(\bzero,\bI_d).
\]
Combining these maps componentwise gives a measurable \(\backt\). Starting from
\[
\bZ_T^s\coloneq\enc(\wbx_{1:T}^{\mathrm{obs}},\cdeps_T),
\]
define recursively
\[
\bZ_t^s
\coloneq
\backt(\bZ_{t+1}^s,\wbx_{1:t}^{\mathrm{obs}},\tdxi_t),
\qquad t=T-1,\ldots,1.
\]
By Lemma~\ref{lemma:fact-1}~$(iii)$, each update is the conditional law of
\(\bZ_t\) given the sampled future latent states and the full observed
trajectory. The backward chain rule therefore yields
\begin{align}\label{supp:joint-Z}
    (\bZ_1^s,\ldots,\bZ_T^s)
\deq
(\bZ_1,\ldots,\bZ_T)\mid
\bigl\{\wbX_{1:T}^{\mathrm{obs}}=\wbx_{1:T}^{\mathrm{obs}}\bigr\}.
\end{align}
Finally, Lemma~\ref{lemma:fact-1}~$(iv)$ and the noise-outsourcing lemma give a
measurable decoder \(\dec\) satisfying
\[
\dec(\bz_t,\wbx_t^{\mathrm{obs}},\tdeta_t)
\deq
\bX_t\mid
\bigl\{\bZ_t=\bz_t,\wbX_t^{\mathrm{obs}}=\wbx_t^{\mathrm{obs}}\bigr\},
\qquad
\tdeta_t\sim\calN(\bzero,\bI_p).
\]
Substituting the smoothed sample \(\bZ_t^s\) gives
\[
\dec(\bZ_t^s,\wbx_t^{\mathrm{obs}},\tdeta_t)
\deq
\bX_t\mid
\bigl\{\wbX_{1:T}^{\mathrm{obs}}=\wbx_{1:T}^{\mathrm{obs}}\bigr\}.
\]
% as required, where all reference noises are mutually independent and independent
% of \(\wbX_{1:T}^{\mathrm{obs}}\).
Moreover, let
\(\tbX_t \coloneq\dec(\bZ_t^s,\wbx_t^{\text{obs}},\tdeta_t)\). By the conditional independence property in Lemma~\ref{lemma:fact-1}~$(iv)$ and \eqref{supp:joint-Z},
\bel{supp:joint-X}
(\bX_1,\ldots,\bX_T)\mid
\bigl\{\wbX_{1:T}^{\mathrm{obs}}=\wbx_{1:T}^{\mathrm{obs}}\bigr\}
&\deq& (\tbX_1,\ldots,\tbX_T).
\eel
\end{proof}
%%%%%%%%%%%%%%%%

%%%%%%%%%%%%%%%%
%%%%%%%%%%%%%%%%
%%%%%%%%%%%%%%%%
%%%%%%%%%%%%%%%%%%%%%%%%%%%%%%%%
\subsection{Proof of proposition \ref{prop:mask-proper}}\label{pf:prop:mask-proper}
%%%%%%%%%%%%%%%-- new on, July-13-2026
\begin{proof}
Define
\[
\Gamma(\wbx^{\mathrm{obs}})
\coloneq
\left(
(\bx^{\mathrm{obs}}\odot\bo)^\top,\bo^\top
\right)^\top
\in\bbR^{2p}.
\]
It is clear that \(\Gamma\) is injective on
\(\widebar{\calX}^{\mathrm{obs}}\). By \eqref{dist-proj-main},
\[
\operatorname{dist}(\wbx_1^{\mathrm{obs}},\wbx_2^{\mathrm{obs}})
=
\|\Gamma(\wbx_1^{\mathrm{obs}})
-\Gamma(\wbx_2^{\mathrm{obs}})\|_2.
\]
Thus \(S^{\mathrm{mask}}(\wbx^{\mathrm{obs}};Q)\) is the ordinary energy
score of \(\Gamma(\wbx^{\mathrm{obs}})\) under the pushforward law
\(\Gamma_{\#}Q\). Theorem~5 of \citet{gneiting2007strictly}, with
\(m=2p\) and \(\psi(t)=t^{1/2}\), gives strict propriety over laws with
finite first moments. Hence equality of the expected scores holds if and only
if
\(\Gamma_{\#}P_o=\Gamma_{\#}Q\), which, by injectivity of \(\Gamma\),
is equivalent to \(P_o=Q\).
\end{proof}

\subsection{Proof of theorem \ref{thm:full_law_consistency_main}}\label{pf:thm:full_law_consistency_main}
\begin{proof}
By Proposition~\ref{prop:partial_law_main},
\[
W_{1,o}\bigl(Q_{o,T}^{\hbtheta_n},Q_{o,T}^{\btheta^*}\bigr)
\overset{p}{\longrightarrow}0.
\]
Define the continuous functions
\[
G_{o,T}(\btheta)\coloneq W_{1,o}\bigl(Q_{o,T}^{\btheta},Q_{o,T}^{\btheta^*}\bigr),
\qquad
G_T(\btheta)\coloneq W_1\bigl(Q_T^{\btheta},Q_T^{\btheta^*}\bigr).
\]
Continuity follows from Lemma~\ref{lemma:theta_law_continuity}. Moreover, if
\(G_{o,T}(\btheta)=0\), then
\(Q_{o,T}^{\btheta}=Q_{o,T}^{\btheta^*}\). Lemma~\ref{lem:full_law_finiteT_main}
then gives \(Q_T^{\btheta}=Q_T^{\btheta^*}\). Since the full laws have finite
first moments under Assumption~\ref{assum:model_encoder_regularity}, this
implies
\[
G_{o,T}(\btheta)=0 \quad\Longrightarrow\quad G_T(\btheta)=0.
\]
Fix \(\varepsilon>0\), and consider the compact set
\[
K_\varepsilon\coloneq\{\btheta\in\Theta:G_T(\btheta)\ge\varepsilon\}.
\]
{If \(K_\varepsilon=\varnothing\), then
\(\bbP\{G_T(\hbtheta_n)\ge\varepsilon\}=0\) for every \(n\). Hence it remains
to consider the case \(K_\varepsilon\neq\varnothing\).} Since
\(G_{o,T}\) is continuous and \(K_\varepsilon\) is compact,
\[
\delta_\varepsilon\coloneq\inf_{\btheta\in K_\varepsilon}G_{o,T}(\btheta)>0.
\]
Therefore,
\[
\{G_T(\hbtheta_n)\ge\varepsilon\}
\subseteq
\{G_{o,T}(\hbtheta_n)\ge\delta_\varepsilon\}.
\]
Taking probabilities and using the observation-law consistency result yields
\[
\bbP\{G_T(\hbtheta_n)\ge\varepsilon\}
\le
\bbP\{G_{o,T}(\hbtheta_n)\ge\delta_\varepsilon\}
\longrightarrow0.
\]
Since \(\varepsilon>0\) was arbitrary,
\[
W_1\bigl(Q_T^{\hbtheta_n},Q_T^{\btheta^*}\bigr)
\overset{p}{\longrightarrow}0.
\]
\end{proof}
%%%%%%%%%%%%%%%%%%%%%%%%%%%%%-- 
%%%%%%%%%%%%%%%%%%%%%%%%
\subsection{Proof of Lemma~\ref{lemma:linear_gaussian_identifiability}}
\label{pf:lemma:linear_gaussian_identifiability}
\begin{proof}
The proof is divided into five steps.

\noindent \emph{Step 1: fully observed one-time law forces $\bF$ to have full column rank.} 
Since $Q_{o,T}^{\btheta}=Q_{o,T}^{\btheta^*}$, the first-time observed marginals are equal.  Conditional on $\bO_1={\mathbf 1}$, by \eqref{eq:full_pushforward_equal} we have
\begin{equation}
    \bF_\#\nu_{{\mathbf 1}}^{\btheta}
    =
    \bF^*_\#\nu_{{\mathbf 1}}^{\btheta^*}.
    \label{eq:full_pattern_deconv}
\end{equation}
Taking supports in \eqref{eq:full_pattern_deconv},
\[
    \Col(\bF)
    =
    \supp(\bF_\#\nu_{{\mathbf 1}}^{\btheta})
    =
    \supp(\bF^*_\#\nu_{{\mathbf 1}}^{\btheta^*})
    =
    \Col(\bF^*).
\]
Since $\rank(\bF^*)=d$, it follows that $\rank(\bF)=d$.  Thus the candidate full measurement map $z\mapsto \bF z$ is continuous and one-to-one.

In addition, according to \eqref{eq:f_relation_new}, there exists a unique invertible matrix
\[
    \bR=(\bF^{*\top}\bF^*)^{-1}\bF^{*\top}\bF
\]
such that
\begin{equation}
    \bF=\bF^*\bR .
    \label{eq:proof_F_equals_FstarR}
\end{equation}
\noindent\emph{Step 2: singleton-missing patterns identify the affine logits up to $\bR$.}
By \eqref{eq:nu_transport_new}, for each $j=1,\ldots,p$,
\[
    \bF_{-j}=\bF^*_{-j}\bR .
\]
Since $\bF^*_{-j}$ has full column rank and $\bR$ is invertible, $\bF_{-j}$ has full column rank.   
By \eqref{eq:h_mu_equal_new}, 
\begin{equation}
    a_j=a_j^*,
    \qquad
    \bw_j=\bR^\top\bw_j^*,
    \qquad j=1,\ldots,p .
    \label{eq:logit_identification}
\end{equation}

\noindent\emph{Step 3: $\bR$ is orthogonal.}
Let $\gamma=\calN(\bzero,\bI_d)$ and by \eqref{eq:mu_z_transformation}, we get
\[
    \bR_\#\gamma=\gamma.
\]
Equivalently, if $G\sim\calN(\bzero,\bI_d)$, then $\bR G\sim\calN(\bzero,\bI_d)$.  Hence $\bR\bR^\top=\bI_d$, so $\bR\in O(d)$.

\medskip
\noindent\emph{Step 4: the two-time observed law identifies $A$.}
Under the candidate model,
\[
    \bZ_2=\bA\bZ_1+\bv_1,
    \qquad \bv_1\sim\calN(\bzero,\bI_d),
\]
with $\bv_1$ independent of $\bZ_1\sim\calN(\bzero,\bI_d)$.  Hence
\[
    \bbE_{\btheta}[\bZ_2\bZ_1^\top]=\bA.
\]
Under the true model, $\bbE_{\btheta^*}[\bZ_2^*(\bZ_1^*)^\top]=\bA^*$.  By \eqref{eq:mu_z_transformation_T}, we have
\[
    \bR\bA\bR^\top=\bA^*.
\]
Since $\bR$ is orthogonal,
\begin{equation}
    \bA=\bR^\top\bA^*\bR .
    \label{eq:A_identified}
\end{equation}

\medskip
\noindent\emph{Step 5: converse inclusion.}
Suppose that for some $\bR\in O(d)$,
\[
    \bF=\bF^*\bR,
    \qquad
    \ba=\ba^*,
    \qquad
    \bw_j=\bR^\top\bw_j^*,
    \qquad
    \bA=\bR^\top\bA^*\bR.
\]
Let $\widetilde{\bZ}_t=\bR\bZ_t$ under the candidate model.  Since $\bR$ is orthogonal, $\widetilde{\bZ}_1\sim\calN(\bzero,\bI_d)$ and $\bR\bv_t\sim\calN(\bzero,\bI_d)$.  Moreover,
\[
    \widetilde{\bZ}_{t+1}
    =
    \bR\bA\bZ_t+\bR\bv_t
    =
    \bA^*\widetilde{\bZ}_t+\widetilde{\bv}_t,
    \qquad
    \widetilde{\bv}_t\sim\calN(\bzero,\bI_d).
\]
Thus $\widetilde{\bZ}_{1:T}$ has the same law as the true latent path.  The measurement signal and logits also match:
\[
    \bF\bZ_t=\bF^*\widetilde{\bZ}_t,
    \qquad
    a_j+\bw_j^\top\bZ_t
    =
    a_j^*+(\bw_j^*)^\top\widetilde{\bZ}_t .
\]
With the same output noise law, this gives $Q_{o,T}^{\btheta}=Q_{o,T}^{\btheta^*}$.
% \medskip
% \noindent\emph{Step 7: conclusion.}
% Steps 1--5 prove the necessary form of every parameter in $\Theta_{o,T}^{\mathrm{lin}}({\btheta^*})$.  Step 6 proves the reverse inclusion.  The full-rank conclusions follow from $\bF=\bF^*\bR$, $\bF_{-j}=\bF^*_{-j}\bR$, the full column rank of $\bF^*_{-j}$, and invertibility of $\bR$.
\end{proof}
%%%%%%%%%%%%%%%%%%%%%%
%%%justification for  self-mask
%%%%%%%%%%%%%%%%%%%%%%
%%%%%%%%%%%%%%%%%%%%%%%%
%%%%%%%%%%%%%%%%%%%%%%
%%%%%%%%%%%%%%%%%%%%%%%% Identifiability and consistency
%%%%%%%%%%%%%%%%%%%%%%
\section{Proofs for Section \ref{supp:sec:key_properties}}\label{supp:sec:proof_S1}
\subsection{Proof of lemma \ref{lemma:fact-1}}\label{pf:lemma:fact-1}
\begin{proof}
By Lemma~\ref{lemma:factorization}, the joint law factorizes according to the
first-order state-space DAG. Its d-separations give
\begin{align*}
\bZ_{t+1}
&\perp\!\!\!\perp
(\bZ_{1:t-1},\wbX_{1:t}^{\mathrm{obs}})
\mid \bZ_t,\\
\wbX_{t+1:T}
&\perp\!\!\!\perp
(\bZ_{1:t},\wbX_{1:t}^{\mathrm{obs}})
\mid \bZ_{t+1},\\
\bZ_t
&\perp\!\!\!\perp
(\bZ_{t+2:T},\wbX_{t+1:T}^{\mathrm{obs}})
\mid(\bZ_{t+1},\wbX_{1:t}^{\mathrm{obs}}),\\
\bX_t
&\perp\!\!\!\perp
(\bZ_{-t},\wbX_{-t}^{\mathrm{obs}})
\mid(\bZ_t,\wbX_t^{\mathrm{obs}}),
\end{align*}
where the subscript \(-t\) denotes all time indices other than \(t\).
These four conditional independences are exactly statements~$(i)$--$(iv)$.
\end{proof}

\subsection{Proof of lemma \ref{lemma:fact-2}}\label{pf:lemma:fact-2}
\begin{proof}
By Lemma~\ref{lemma:factorization-k}, each \(\bZ_s\) has latent parents
\(\bZ_{(s-k)_+:s-1}\), while the observation mechanism remains local in
\(\bZ_s\). Set
\[
a_t=(t-k+1)_+,\qquad c_t=(t+2-k)_+,\qquad b_t=(t+k)\wedge T.
\]
The corresponding graph separations are
\begin{align*}
\bZ_{t+1}
&\perp\!\!\!\perp
(\bZ_{1:a_t-1},\wbX_{1:t}^{\mathrm{obs}})
\mid\bZ_{a_t:t},\\
\wbX_{t+1:T}
&\perp\!\!\!\perp
(\bZ_{1:c_t-1},\wbX_{1:t}^{\mathrm{obs}})
\mid\bZ_{c_t:t+1},\\
\bZ_t
&\perp\!\!\!\perp
(\bZ_{b_t+1:T},\wbX_{t+1:T}^{\mathrm{obs}})
\mid(\bZ_{t+1:b_t},\wbX_{1:t}^{\mathrm{obs}}),\\
\bX_t
&\perp\!\!\!\perp
(\bZ_{-t},\wbX_{-t}^{\mathrm{obs}})
\mid(\bZ_t,\wbX_t^{\mathrm{obs}}).
\end{align*}
Empty index blocks are omitted. These are precisely statements~$(i)$--$(iv)$.
\end{proof}
\subsection{Proof of theorem \ref{thm:k-backward}}\label{pf:thm:k-backward}\label{pf:thm:k-forecasting}
\begin{proof}
For \(\ell<t\), let \(b_\ell=(\ell+k)\wedge t\). Lemma~\ref{lemma:fact-2}~$(iii)$
and the noise-outsourcing lemma give a measurable map \(\backt_\ell\) such that
\[
\backt_\ell(\bZ_{\ell+1:b_\ell},\wbX_{1:\ell}^{\text{obs}},\tdxi_\ell)
\deq
\bZ_\ell\mid(\bZ_{\ell+1:t},\wbX_{1:t}^{\text{obs}}).
\]
Combining these maps into \(\backt\), initializing with the encoder draw
\(\bZ_t^s\), and applying the backward chain rule give
\[
(\bZ_1^s,\ldots,\bZ_t^s)
\deq
(\bZ_1,\ldots,\bZ_t)\mid \wbX_{1:t}^{\text{obs}}.
\]
Let \(a_t=(t+1-k)_+\). By Lemma~\ref{lemma:fact-2}~$(i)$,
\[
p(\bZ_{t+1}\mid \bZ_{1:t},\wbX_{1:t}^{\text{obs}})
=
p(\bZ_{t+1}\mid \bZ_{a_t:t}).
\]
The noise-outsourcing lemma supplies the transition and joint output maps.
Since the smoothing result gives
\(\bZ_{a_t:t}^s\deq\bZ_{a_t:t}\mid\wbX_{1:t}^{\text{obs}}\), substitution yields
\[
\bout\left\{\btrans\left(\bZ_{a_t:t}^s,\deps_t\right),\deta_t\right\}
\deq
(\bX_{t+1},\bO_{t+1})\mid\wbX_{1:t}^{\text{obs}},
\]
which proves both claims.
\end{proof}
\section{Supporting results for full-law consistency}
\label{supp:sec:proof_S2}
The proof sketch of Theorem~\ref{thm:full_law_consistency_main} proceeds through two main steps:
observation-law consistency
(Proposition~\ref{prop:partial_law_main}), and finite-horizon full-law
identifiability
(Lemma~\ref{lem:full_law_finiteT_main}).
\subsection{Encoder-assisted partial-law identification and consistency}\label{sec:partial_identification}
The first ingredient behind Theorem~\ref{thm:full_law_consistency_main} is
that the encoder-assisted forecasting objective defined above consistently estimates the
observed law $Q_{o,T}^{\btheta^*}$. 

%%%%%%%%%%%%%%%%%%%%%%%%
%%%%%%%%%%%%%%%%%%%%%%%%%%%%%%%%%%
For observation laws, \(W_{1,o}\) is defined analogously on the space \((\calX^{\text{obs}}\times\calO)^T\). For
\(\wbx_{1:T}^{\mathrm{obs}}\) and \(\wbx_{1:T}^{\mathrm{obs},\prime}\), the metric is $\rho_o(\wbx_{1:T}^{\mathrm{obs}},\wbx_{1:T}^{\mathrm{obs},\prime}) = \sum_{t=1}^T \operatorname{dist(\wbx_t^{\mathrm{obs}},\wbx_t^{\mathrm{obs},\prime})}$, which is defined as Eq.~\ref{dist-proj-main}.
% $\rho^{\text{obs}}(\wbx_{1:T}^{\mathrm{obs}},\wbx_{1:T}^{\mathrm{obs},\prime})
% \coloneq
% \sum_{t=1}^T
% \left(
% \|\bx_t^{\mathrm{obs}}\odot\bo_t-\bx_t^{\mathrm{obs},\prime}\odot\bo_t'\|_2
% +\|\bo_t-\bo_t'\|_2
% \right).$
\begin{proposition}[Observation-law identifiability and consistency]
\label{prop:partial_law_main}
Consider the model given in~\eqref{eq:dgp_main}. Let Assumptions~\ref{assum:encoder_consistency} and  \ref{assum:model_encoder_regularity} hold. Then, for any fixed \(T<\infty\), every population minimizer 
\begin{align}
       (\tbtheta, \tbphi) \ \in \ \argmin_{(\btheta, \bphi)\in \calA} \calL(\btheta, \bphi)
\end{align} satisfies 
\begin{align}
 Q_{o,T}^{\tbtheta,\tbphi} \ = \ Q_{o,T}^{\btheta^*}, \quad  Q_{o,T}^{\tbtheta} \ = \ Q_{o,T}^{\btheta^*}. 
%  Q_{o,T}^{\tbtheta,\tbphi} \ = \ Q_{o,T}^{\btheta^*} ,\quad
\end{align}
Moreover, as $n \to \infty$, for the empirical estimator defined in \eqref{def:em}, 
\begin{align}
   W_{1,o}\left(Q_{o,T}^{\hbtheta_n, \hbphi_n},\ Q_{o,T}^{\btheta^*}\right) \overset{p}{\longrightarrow}0, \quad   W_{1,o}\left(Q_{o,T}^{\hbtheta_n},\ Q_{o,T}^{\btheta^*}\right) \overset{p}{\longrightarrow}0 . 
\end{align}
\end{proposition}
Proposition~\ref{prop:partial_law_main} has two implications. First, at the
population level, any minimizer of the proposed objective recovers the true
observation law \(Q_{o,T}^{\btheta^*}\). This holds both for the
encoder-assisted law \(Q_{o,T}^{\tbtheta,\tbphi}\) and for the generative
law \(Q_{o,T}^{\tbtheta}\), because the encoder is aligned with the filtering
distribution under Assumption~\ref{assum:encoder_consistency}. Second, the empirical estimator is consistent: as the sample size increases, the
learned observed law converges to the true observed law \(Q_{o,T}^{\btheta^*}\)
in \(W_{1,o}\).
%%%%%%%%%%%%%%%%%%%%%%%%%%%%%%%%%%%%%%%%%%%%%%%%%%%%%%%%%%%%%%%%%%%%%%%%%%%%%%%%%%%%%%%%
\subsection{Full-law identifiability}\label{sec:full_identification}
We study when the full law \(Q_T^{\btheta}\) can be recovered from the
observation law \(Q_{o,T}^{\btheta}\) for \(T<\infty\). We first consider
the static case \((T=1)\); the finite-horizon result then extends this one-time
argument. 
Throughout this subsection, we compare a candidate parameter \(\btheta\) with the
data-generating parameter \(\btheta^*\), assuming that they share the same
initial latent law and noise laws.

\noindent\emph{Static backbone} (\(T=1\)).
Consider
\begin{align}
\bX &= f_x^{\btheta}(\bZ)+\br, \notag\\
\bbP(\bO_j=1\mid \bZ=\bz)
&=
\sigma\bigl(\pi_j^{\btheta}(\bz)\bigr),
\qquad j=1,\dots,p,
\label{eq:static_full_law_model}
\end{align} 
where \(\bZ\sim \mu_{\bZ}\) and $\br \sim \law(\br)$ is independent of \(\bZ\). Let $Q_1^{\btheta}\coloneq \law(\bX_1, \bO_1)$ and $Q_{o,1}^{\btheta}\coloneq \law(\bX_1^{\text{obs}},\bO_1)$ denote the full and observed one-time laws induced by $\btheta$.

\begin{lemma}[Static full-law identifiability]
\label{lem:full_law_static_main}
Consider the model given in~\eqref{eq:static_full_law_model}. Let  Assumption~\ref{assum:static_base_main} holds. Then, if $Q_{o,1}^{\btheta}=Q_{o,1}^{\btheta^*},$ there exists a measurable latent reparameterization
\(\calT:\bbR^d\to\bbR^d\), one-to-one \(\mu_{\bZ}\)-almost everywhere, such that
\[
\calT_{\#}\mu_{\bZ} = \mu_{\bZ},\qquad f_x^{\btheta}=f_x^{\btheta^*}\circ\calT,\qquad
\pi_j^{\btheta}=\pi_j^{\btheta^*}\circ\calT,\quad j=1,\dots,p,
\]
where the functional identities hold \(\mu_{\bZ}\)-almost everywhere.
Consequently,
\begin{align}
Q_1^{\btheta}=Q_1^{\btheta^*}.
\end{align}
\end{lemma}
Lemma~\ref{lem:full_law_static_main} shows that equality of the observed
one-time laws identifies the full one-time law, with the latent distribution,
measurement map, and missingness logits identified up to a common invertible
reparameterization.

\noindent\emph{Finite-horizon model.}
We now return to model~\eqref{eq:dgp_main}. To propagate the static
identification argument over time, we impose the following nondegeneracy
conditions on the latent transition.

\begin{lemma}[Finite-horizon full-law identifiability]
\label{lem:full_law_finiteT_main}
Fix \(T<\infty\).  Consider the model given in~\eqref{eq:dgp_main}. Let Assumptions~\ref{assum:static_base_main} holds. If $Q_{o,T}^{\btheta}=Q_{o,T}^{\btheta^*},$ 
then
\[
Q_T^{\btheta}=Q_T^{\btheta^*}.
\]
\end{lemma}
Lemma~\ref{lem:full_law_finiteT_main} extends the static identification
result to finite trajectories: under the stated structural conditions, the
partial observation law determines the corresponding full law.
\subsection{Proof of proposition \ref{prop:partial_law_main}}
\label{pf:prop:partial_law_main}

\begin{proof}
We divide the proof into five steps.

\medskip
\noindent\emph{Step 1: continuity of the masked energy score
\(S^{\text{mask}}\).}
The triangle inequality and \eqref{dist-proj-main} give
\begin{align}
|S^{\mathrm{mask}}(\wbx^{\mathrm{obs}};Q)|
&\le
\|\bx^{\mathrm{obs}}\odot\bo\|_2
+2\bbE_Q\|\bX^{\mathrm{obs}}\odot\bO\|_2
+\frac32\sqrt p,
\label{supp:S-mask_bounded}
\end{align}
for every predictive law \(Q\) with a finite first moment. Since
\(\|\bX^{\mathrm{obs}}\odot\bO\|_2\le\|\bX\|_2\), this bound is uniform over
the model-induced predictive laws at each fixed history by
Lemma~\ref{lemma:network_moments}.

Fix \(t\in[T]\), a history \(\wbx_{1:t-1}^{\mathrm{obs}}\), and
\((\btheta_n,\bphi_n)\to(\btheta,\bphi)\). Write
\[
Q_n=Q_{o,t\mid1:t-1}^{\btheta_n,\bphi_n}
(\cdot\mid\wbx_{1:t-1}^{\mathrm{obs}}),
\qquad
Q=Q_{o,t\mid1:t-1}^{\btheta,\bphi}
(\cdot\mid\wbx_{1:t-1}^{\mathrm{obs}}).
\]
Apply the conditional analogue of the coupling in the proof of
Lemma~\ref{lemma:theta_law_continuity} independently to the two predictive
draws. The two distance terms defining the score then converge almost surely.
By Lemma~\ref{lemma:network_moments} and the triangle inequality, their second
moments are uniformly bounded, and hence they are uniformly integrable. The Vitali convergence theorem
\citep[Theorem~4.6.3]{durrett2019probability} then gives
\[
S^{\mathrm{mask}}(\wbx^{\mathrm{obs}};Q_n)
\longrightarrow
S^{\mathrm{mask}}(\wbx^{\mathrm{obs}};Q).
\]
\medskip
\noindent\emph{Step 2: encoder-assisted law equals the true conditional law. }
Fix \((\btheta,\bphi)\in\calA\). Since
\(\bphi\in\tbPhi(\btheta)\), strict propriety of the energy score in
\(\calL_{\mathrm{syn},\enc}\), together with
Assumption~\ref{assum:encoder_consistency}, implies that, for each
\(t=2,\ldots,T\) and
\(Q_{o,1:t-1}^{\btheta}\)-almost every history,
\[
\enc^{\bphi}(\wbx_{1:t-1}^{\text{obs}},\cdeps_t) \ \sim \ 
\law_{\btheta}\{\bZ_{t-1}\mid\wbX_{1:t-1}^{\text{obs}}=\wbx_{1:t-1}^{\text{obs}}\}.
\]
Starting from the filtered state generated by the encoder, we generate the next
step using the transition, output, and missingness mechanisms indexed by
\(\btheta\). Hence
\[
Q_{o,t\mid1:t-1}^{\btheta,\bphi}(\cdot\mid\wbx_{1:t-1}^{\text{obs}})
=
Q_{o,t\mid1:t-1}^{\btheta}(\cdot\mid\wbx_{1:t-1}^{\text{obs}})
\]
for each \(t=2,\ldots,T\) and \(Q_{o,1:t-1}^{\btheta}\)-almost every history. The equality is immediate
for \(t=1\). By the chain rule for regular conditional distributions,
\begin{align}
Q_{o,T}^{\btheta,\bphi}=Q_{o,T}^{\btheta},
\qquad (\btheta,\bphi)\in\calA.
\label{eq:aligned_law_equals_exact_law}
\end{align}

%%%%%%%%%%%%%%%%%>>> L_1 bracketing
\medskip
\noindent\emph{Step 3: uniform convergence of the encoder-assisted empirical risk.}
Recall that 
\[
\calA \coloneq
\left\{(\btheta, \bphi)\in \Theta \times \Phi : \bphi \in \tbPhi(\btheta)\right\}
\]
Since \(\Theta\) and \(\Phi\) are compact, and
Lemma~\ref{lemma:encoder_calibration_loss_continuity} shows that
\((\btheta,\bphi)\mapsto
\calL_{\mathrm{syn},\enc}^{\btheta}(\bphi)\) is jointly continuous, the
argmin graph \(\calA\) is closed in \(\Theta\times\Phi\), and hence compact.
Let
\[
P\coloneq Q_{o,T}^{\btheta^*},
\qquad
P_n f\coloneq\frac1n\sum_{i=1}^n f(\wbX_{i,1:T}^{\text{obs}}), \qquad P f \coloneq  \int fdP
\]
and define the class
\[
\mathcal L_T
\coloneq
\{\ell_{\btheta,\bphi}:(\btheta,\bphi)\in\calA\}.
\]
Then
\[
\calL_n(\btheta,\bphi)=P_n\ell_{\btheta,\bphi},
\qquad
\calL(\btheta,\bphi)=P\ell_{\btheta,\bphi}.
\]
By Lemma \ref{lemma:encoder_assisted_continuity}, for each fixed trajectory
\(\wbx_{1:T}^{\text{obs}}\), the map
\((\btheta,\bphi)\mapsto \ell_{\btheta,\bphi}(\wbx_{1:T}^{\text{obs}})\)
is continuous on \(\calA\). Equations~\eqref{supp:S-mask_bounded} and
\eqref{eq:conditional-network-moment} give the common envelope
\begin{align}
|\ell_{\btheta,\bphi}(\wbx_{1:T}^{\text{obs}})|
&\le F_T(\wbx_{1:T}^{\text{obs}})
\coloneq C_T\left(1+\sum_{t=1}^T
\|\bx_t^{\text{obs}}\odot\bo_t\|_2\right),
\qquad (\btheta,\bphi)\in\calA,
\label{eq:envelope-1}
\end{align}
which is \(P\)-integrable by Lemma~\ref{lemma:network_moments}.
Thus \(\mathcal L_T\) is indexed by a compact parameter set; for each fixed
trajectory, \(\ell_{\btheta,\bphi}\) is continuous in \((\btheta,\bphi)\); and
the class admits a \(P\)-integrable envelope. By the standard bracketing
criterion for such classes \citep[Theorem~19.4 and Example~19.8]{van2000asymptotic},
for every \(\varepsilon>0\),
\begin{align}
    N_{[]}\bigl(\varepsilon,\mathcal L_T,L_1(P)\bigr)<\infty .
\end{align}
By Theorem 19.4 of \citet{van2000asymptotic}, \(\mathcal L_T\) is \(P\)-Glivenko--Cantelli.
Therefore,
\[
\sup_{(\btheta,\bphi)\in\calA}
\left|
(P_n-P)\ell_{\btheta,\bphi}
\right|
\overset{a.s.}{\longrightarrow}0 .
\]
Equivalently,
% \[
% \sup_{(\btheta,\bphi)\in\calA}
% \left|
% \calL_n(\btheta,\bphi)-\calL(\btheta,\bphi)
% \right|
% \overset{a.s.}{\longrightarrow}0 .
% \]
% In particular,
\begin{align}
    \sup_{(\btheta,\bphi)\in\calA}
\left|
\calL_n(\btheta,\bphi)-\calL(\btheta,\bphi)
\right|
\overset{a.s.}{\longrightarrow}0 .
\label{eq:uniform_lln_A}
\end{align}

\medskip
\noindent\emph{Step 4: population minimizer identifies the observation law.}
By Assumption~\ref{assum:encoder_consistency} and Step 2, there is a feasible
aligned pair \((\btheta^*,\bphi^*)\in\calA\) whose one-step predictive laws
equal the true conditional laws. Hence these laws are attainable in \(\calA\).
Recall that, for \((\btheta,\bphi)\in\calA\),
\[
\calL(\btheta,\bphi)
=
-\frac1T\sum_{t=1}^T
\bbE_{\btheta^*}
\left[
S^{\text{mask}}\left(
\wbX_t^{\text{obs}};
Q_{o,t\mid1:t-1}^{\btheta,\bphi}(\cdot\mid\wbX_{1:t-1}^{\text{obs}})
\right)
\right].
\]
Let
\[
\mathcal M\coloneq\operatorname*{arg\,min}_{(\btheta,\bphi)\in\calA}\calL(\btheta,\bphi).
\]
The set \(\mathcal M\) is nonempty by compactness of \(\calA\) and continuity of
\(\calL\). In addition, 
by strict propriety of \(S^{\mathrm{mask}}\), each conditional term is minimized
only when
\[
Q_{o,t\mid1:t-1}^{\btheta,\bphi}(\cdot\mid\wbX_{1:t-1}^{\text{obs}})
=
Q_{o,t\mid1:t-1}^{\btheta^*}(\cdot\mid\wbX_{1:t-1}^{\text{obs}})
\]
for \(Q_{o,1:t-1}^{\btheta^*}\)-almost every history. Therefore, any
\((\tbtheta,\tbphi)\in\mathcal M\) must attain them for every
\(t\). The chain rule therefore gives
\[
Q_{o,T}^{\tbtheta,\tbphi}=Q_{o,T}^{\btheta^*}.
\]
Finally, since \((\tbtheta,\tbphi)\in\calA\), Step 2 yields
\begin{align}
Q_{o,T}^{\tbtheta}=Q_{o,T}^{\tbtheta,\tbphi}=Q_{o,T}^{\btheta^*}.
\label{eq:population_minimizer_induce_true_law}
\end{align}

\medskip
\noindent\emph{Step 5: consistency of the induced observation law. }
Let \((\tbtheta,\tbphi)\in\mathcal M\). By definition of \((\hbtheta_n,\hbphi_n)\),
\[
\calL_n(\hbtheta_n,\hbphi_n)
\le
\calL_n(\tbtheta,\tbphi).
\]
Therefore,
\begin{align*}
0
&\le
\calL(\hbtheta_n,\hbphi_n)-\calL(\tbtheta,\tbphi)\\
&\le
2\sup_{(\btheta,\bphi)\in\calA}
\left|
\calL_n(\btheta,\bphi)-\calL(\btheta,\bphi)
\right|,
\end{align*}
and uniform convergence \eqref{eq:uniform_lln_A} in Step 3 implies
\begin{align}
\calL(\hbtheta_n,\hbphi_n)-\calL(\tbtheta,\tbphi)
\overset{p}{\longrightarrow}0.
\label{eq:risk_excess_to_zero}
\end{align}
Define
\[
G_{o,T}(\btheta,\bphi)
\coloneq
W_{1,o}\{Q_{o,T}^{\btheta,\bphi},Q_{o,T}^{\btheta^*}\}.
\]
By Lemma \ref{lemma:encoder_assisted_continuity}, \(G_{o,T}\) {is continuous on
\(\calA\)}. Step 4, \eqref{eq:population_minimizer_induce_true_law} shows that
every population minimizer satisfies \(G_{o,T}=0\). Hence, for each
\(\varepsilon>0\), the compact set
\[
K_\varepsilon\coloneq\{(\btheta,\bphi)\in\calA:\ G_{o,T}(\btheta,\bphi)\ge\varepsilon\}
\]
contains no population minimizer. {If \(K_\varepsilon=\varnothing\), then
\(G_{o,T}(\btheta,\bphi)=0\) for all \((\btheta,\bphi)\in\calA\). Hence it remains
to consider the case \(K_\varepsilon\neq\varnothing\).}  {Since \(\calL\) is continuous and
\(K_\varepsilon\) is compact, it has a positive risk gap:}
\[
\Delta_\varepsilon
\coloneq
\inf_{(\btheta,\bphi)\in K_\varepsilon}\calL(\btheta,\bphi)
-
\calL(\tbtheta,\tbphi)
>0.
\]
Using \eqref{eq:risk_excess_to_zero},
\[
\bbP\{G_{o,T}(\hbtheta_n,\hbphi_n)\ge\varepsilon\}
\le
\bbP\left
\{
\calL(\hbtheta_n,\hbphi_n)-\calL(\tbtheta,\tbphi) \ge\Delta_\varepsilon
\right\}
\to0.
\]
Thus
\[
W_{1,o}\left(Q_{o,T}^{\hbtheta_n,\hbphi_n},Q_{o,T}^{\btheta^*}\right)
\overset{p}{\longrightarrow}0.
\]
Since \((\hbtheta_n,\hbphi_n)\in\calA\), Step 2 also gives
\(Q_{o,T}^{\hbtheta_n,\hbphi_n}=Q_{o,T}^{\hbtheta_n}\). Therefore
\[
W_{1,o}\left(Q_{o,T}^{\hbtheta_n},Q_{o,T}^{\btheta^*}\right)
\overset{p}{\longrightarrow}0.
\]
\end{proof}
%%%%%%
%%%%%%%%%%%%%%%%%%%%%%%%%%%%%%%%%%%%%%%%%%%%%%%%%%%%%%%%%%%%%%%%%%%%%%%%%%%%%%%%%%%%%%%

% \begin{assumption}[Identifiable mask patterns]
% \label{assum:general_selected_main}
% There exists a collection
% \[
% \widetilde{\calO}=\{\mathbf 1,\bo^{(1)},\dots,\bo^{(m)}\}\subseteq\{0,1\}^p
% \]
% such that:
% \begin{enumerate}
% \item for every \(\bo\in\widetilde{\calO}\), the maps \(f_{x}^{\bo}\) and \(f_{x}^{\bo,\prime}\)  are one-to-one;

% \item the incidence matrix \(M\in\{0,1\}^{m\times p}\), defined by
% \[
% M_{kj}\coloneq \mathds{1} \{j\in S(\bo^{(k)})\},
% \qquad
% k=1,\dots,m,\ \ j=1,\dots,p,
% \]
% has full column rank.
% \end{enumerate}
% \end{assumption}
% Assumption~\ref{assum:general_selected_main} requires sufficiently many
% informative mask patterns to identify the coordinate-wise missingness logits.
% A simple sufficient condition is the singleton-missing case, in which
% \(\widetilde{\calO}\) contains the fully observed pattern and the \(p\)
% patterns that each remove exactly one coordinate; then \(M=I_p\) has full column rank.
%%%%%%%%%%%%%%%%%%%%%%%%%%%%%%%%%%%%%%%%%%%%%%%%%%%%%%%%%%%%%%%%%%%%%%%%%%%%%%%%%%%%%%%
%%%%%%%%%%%%%%%%%%%%%%%%%%%%%%%%%%%%%%%%%%%%%%%%%%%%%%%%%%%%%%%%%%%%%%%%%%%%%%%%%%%%%%%
%%%%%%%%%%%%%%%%%%%%%%%%%%%%%%%%%%%%%%%%%%%%%%%%%%%%%%%%%%%%%%%%%%%%%%%%%%%%%%%%%%%%%%%

\subsection{Proof of lemma~\ref{lem:full_law_static_main}}
\label{pf:lem:full_law_static_main}
\begin{proof}
For brevity, write primed quantities for the model indexed by \(\btheta^*\); for instance,
\(f_x'=f_x^{\btheta^*}\), \(\bZ'=\bZ^*\), and
\(\nu_{\bo}^{\btheta'}=\nu_{\bo}^{\btheta^*}\).
For \(\bo \in \{0,1\}^p\), define
\[
\bX^{\bo} \coloneq (\bX_j : \bo_j=1),\qquad
\br^{\bo} \coloneq (\br_j : \bo_j=1),
\]
and 
\[
\nu_{\bo}^{\btheta}\coloneq\law_{\btheta}(\bZ\mid \bO=\bo),
\qquad
\nu_{\bo}^{\btheta^*}\coloneq\law_{\btheta^*}(\bZ^*\mid \bO=\bo).
\]
We divide the proof into five steps.

\medskip
\noindent\emph{Step 1: the fully observed pattern identifies the latent reparameterization.}
Recall that 
\begin{align}
    \Theta_o(\btheta^*) \coloneq \{\btheta\in\Theta: Q_{o,1}^{\btheta}=Q_{o,1}^{\btheta^*}\}.
\end{align}
Since $\btheta \in \Theta_o(\btheta^*)$, we have
\[
\law_{\btheta}(\bX\mid \bO=\mathbf 1)
=
\law_{\btheta^*}(\bX^*\mid \bO=\mathbf 1).
\]
On the other hand, by the model definition, $\bX = f_x^{\btheta}(\bZ)+\br $ and $
\bX^* = f_x^{\btheta^*}(\bZ^*)+\br,$ we have
\[
\law_{\btheta}(\bX\mid \bO=\mathbf 1)
=
(f_{x}^{\btheta})_{\#}\nu_{\mathbf 1}^{\btheta} * \law(\br), \qquad
\law_{\btheta^*}(\bX^*\mid \bO=\mathbf 1)
=
(f_x^{\btheta^*})_{\#}\nu_{\mathbf 1}^{\btheta^*} * \law(\br).
\]
where the symbol $*$ means convolution of probability laws. Taking Fourier transforms and using the characteristic function of \(\br\)
is non-vanishing on \(\mathbb R^p\), we obtain
\begin{equation}\label{eq:full_pushforward_equal}
(f_{x}^{\btheta})_{\#}\nu_{\mathbf 1}^{\btheta}
=
(f_{x}^{\btheta^*})_{\#}\nu_{\mathbf 1}^{\btheta^*}.
\end{equation}
Thus \(f_x^{\btheta}(\bz)\in f_x^{\btheta^*}(\mathbb R^d)\) for
\(\nu_{\mathbf 1}^{\btheta}\)-almost every \(\bz\), and conversely
\(f_x^{\btheta^*}(\bz^*)\in f_x^{\btheta}(\mathbb R^d)\) for
\(\nu_{\mathbf 1}^{\btheta^*}\)-almost every \(\bz^*\). Almost surely under the respective latent laws, define
\[
\calT(\bz)=(f_x^{\btheta^*})^{-1}\{f_x^{\btheta}(\bz)\},
\qquad
\calT^{-1}(\bz^*)=(f_x^{\btheta})^{-1}\{f_x^{\btheta^*}(\bz^*)\},
\]
and set \(\calT(\bz)=\bzero\) elsewhere. The two expressions are mutual
inverses on the relevant full-probability sets, so \(\calT\) is measurable and
one-to-one \(\nu_{\mathbf 1}^{\btheta}\)-almost everywhere. Thus,
\begin{equation}\label{eq:f_relation_new}
f_x^{\btheta}=f_x^{\btheta^*}\circ\calT,
\qquad \nu_{\mathbf 1}^{\btheta}\text{-a.s.}
\end{equation}
Moreover, by \eqref{eq:full_pushforward_equal} and injectivity of \(f_x^{\btheta^*}\),
\begin{align}\label{eq:nu_one}
    \nu_{\mathbf 1}^{\btheta^*}=\calT_{\#}\nu_{\mathbf 1}^{\btheta}.
\end{align}
\medskip
\noindent\emph{Step 2: selected patterns identify subset sums of the logits}. 
Recall that the singleton-missing pattern collection
\begin{align}\label{supp:def:Otilde}
    \widetilde{\calO}
=
\{\bo^{(1)},\dots,\bo^{(p)}\},
\end{align}
where \(\bo^{(j)}\) is the pattern missing only the \(j\)-th coordinate.
For \(\bo\in\widetilde{\calO}\), let
\(\Pi_{\bo}:\mathbb R^p\to\mathbb R^{|\bo|}\) denote the coordinate projection onto
the observed coordinates of pattern \(\bo\). Then
\[
f_{x}^{\btheta,\bo}=\Pi_{\bo}\circ f_{x}^{\btheta},
\qquad
f^{\btheta^*,\bo}_x=\Pi_{\bo}\circ f^{\btheta^*}_{x}.
\]
Since \(\nu_{\bo}^{\btheta}\ll\nu_{\mathbf 1}^{\btheta}\), \eqref{eq:f_relation_new} gives
\begin{equation}\label{eq:submap_relation}
f_{x}^{\btheta,\bo}
=
\Pi_{\bo}\circ f_x^{\btheta}
=
\Pi_{\bo}\circ f^{\btheta^*}_{x}\circ\calT
=
f^{\btheta^*,\bo}_x\circ\calT,
\qquad \nu_{\bo}^{\btheta}\text{-a.s.}
\end{equation}
Fix \(\bo\in\widetilde{\calO}\). Equality of the observed law gives 
\[
\law_{\btheta}(\bX^{\bo}\mid \bO=\bo)
=
\law_{\btheta^*}\bigl(\bX^{\bo,*}\mid \bO=\bo\bigr).
\]
Then, we arrive at
\[
(f_{x}^{\btheta,\bo})_{\#}\nu_{\bo}^{\btheta} * \law(\br^{\bo}) = (f_{x}^{\btheta^*,\bo})_{\#}\nu_{\bo}^{\btheta^*} * \law(\br^{\bo}).
\]
Since the characteristic function of $\br$ is non-vanishing on \(\mathbb R^p\), it follows that the characteristic function of
\({\br^{\bo}}\) is non-vanishing on \(\mathbb R^{|\bo|}\). Hence deconvolution yields
\begin{equation}\label{eq:pattern_pushforward_equal_new}
(f_{x}^{\btheta,\bo})_{\#}\nu_{\bo}^{\btheta}
=
(f_{x}^{\btheta^*,\bo})_{\#}\nu_{\bo}^{\btheta^*}.
\end{equation}
On the other hand, by \eqref{eq:submap_relation},
\[
(f_{x}^{\btheta,\bo})_{\#}\nu_{\bo}^{\btheta}=
(f^{\btheta^*,\bo}_x\circ\calT)_{\#}\nu_{\bo}^{\btheta}
=
(f^{\btheta^*,\bo}_x)_{\#}(\calT_{\#}\nu_{\bo}^{\btheta}).
\]
Combining this with \eqref{eq:pattern_pushforward_equal_new}, we obtain
\begin{align}
    (f^{\btheta^*,\bo}_x)_{\#}(\calT_{\#}\nu_{\bo}^{\btheta})
=
(f^{\btheta^*,\bo}_x)_{\#}\nu_{\bo}^{\btheta^{*}}.
\end{align}
%%%%% this is why we need the injectivity assumption on Q_{o,T} rather than Q_{T}.
{By Assumption \ref{assum:static_base_main}, \(f^{\btheta^*,\bo}_x\) is injective for every \(\bo\in\widetilde{\calO}\), its pushforward map is
injective on probability measures. Therefore}
\begin{equation}\label{eq:nu_transport_new}
\nu_{\bo}^{\btheta^*}=\calT_{\#}\nu_{\bo}^{\btheta},
\qquad  \bo\in\widetilde{\calO}.
\end{equation}
From \eqref{eq:nu_one} and \eqref{eq:nu_transport_new}, we have
\begin{align}
    \nu_{\mathbf 1}^{\btheta^*}=\calT_{\#}\nu_{\mathbf 1}^{\btheta}, \qquad 
\nu_{\bo}^{\btheta^*}=\calT_{\#}\nu_{\bo}^{\btheta} \qquad \forall \bo\in\widetilde{\calO}.
\end{align}
We now compare the corresponding likelihood ratios. Define
\[
r_{\bo}^{\btheta}(\bz)\coloneq\frac{d\nu_{\bo}^{\btheta}}{d\nu_{\mathbf 1}^{\btheta}}(\bz),
\qquad
r_{\bo}^{\btheta^*}(\bz^*)\coloneq\frac{d\nu_{\bo}^{\btheta^*}}{d\nu_{\mathbf 1}^{\btheta^*}}(\bz^*).
\]
By \eqref{eq:nu_one} and \eqref{eq:nu_transport_new}, \(\calT\) transports both
the numerator and denominator laws to their starred counterparts. By change of variables and the uniqueness of Radon--Nikodym derivatives
\citep[Theorems~1.6.9 and A.4.8]{durrett2019probability}, applying the same
one-to-one transformation to both laws preserves their likelihood ratio.
Therefore,
\begin{equation}\label{eq:RN_relation_new}
\frac{\dd\nu_{\bo}^{\btheta}}{\dd\nu_{\mathbf 1}^{\btheta}}(\bz)
=
\frac{\dd\nu_{\bo}^{\btheta^*}}{\dd\nu_{\mathbf 1}^{\btheta^*}}(\calT(\bz))
\qquad
\nu_{\mathbf 1}^{\btheta}\text{-a.e. }\bz.
\end{equation}
By Bayes' rule,
\[
\frac{\dd\nu_{\bo}^{\btheta}}{\dd\nu_{\mathbf 1}^{\btheta}}(\bz)
=
\frac{\bbP_{\btheta}(\bO=\mathbf 1)}{\bbP_{\btheta}(\bO=\bo)}
\cdot
\frac{\bbP_{\btheta}(\bO=\bo\mid \bZ=\bz)}
{\bbP_{\btheta}(\bO=\mathbf 1\mid \bZ=\bz)}.
\]
Using conditional independence of \(\bO_1,\dots,\bO_p\) given \(\bZ\), we obtain
\begin{align}
    \frac{\bbP_{\btheta}(\bO=\bo\mid \bZ=\bz)}
{\bbP_{\btheta}(\bO=\mathbf 1\mid \bZ=\bz)}
=
\prod_{j \in S(\bo)}
\frac{1-\sigma(\pi_j(\bz))}{\sigma(\pi_j(\bz))}
=
\exp\Bigl(-\sum_{j \in S(\bo)} \pi_j(\bz)\Bigr)
\end{align}
where $S(\bo) \coloneq \{j\in [p]: \bo_j = 0\}$ for the missing-coordinate set. 
Since the fact
\begin{align}
    \frac{1-\sigma(u)}{\sigma(u)}=e^{-u}.
\end{align}
Therefore,
\[
\frac{d\nu_{\bo}^{\btheta}}{d\nu_{\mathbf 1}^{\btheta}}(\bz)
=
c_{\bo}\exp\Bigl(\sum_{j \in S(\bo)}- \pi_j(\bz)\Bigr),
\qquad
c_{\bo}\coloneq
\frac{\bbP_{\btheta}(\bO=\mathbf 1)}{\bbP_{\btheta}(\bO=\bo)}.
\]
The same formula holds under \(\btheta^*\), with the same constant \(c_{\bo}\), because
\(Q_{o,1}^{\btheta}=Q_{o,1}^{\btheta^*}\) implies
\[
\bbP_{\btheta}(\bO=\bo)=\bbP_{\btheta^*}(\bO=\bo),
\qquad \forall \bo.
\]
Substituting into \eqref{eq:RN_relation_new}, we obtain
\begin{equation}\label{eq:subset_sum_equal_new}
\sum_{j\in S(\bo)} \pi^{\btheta}_j(\bz)
=
\sum_{j\in S(\bo)} \pi^{\btheta^*}_j(\calT(\bz))
\qquad
\nu_{\mathbf 1}^{\btheta}\text{-a.e. }\bz,
\quad \forall \bo\in\widetilde{\calO}.
\end{equation}
Recall that the singleton-missing pattern collection
\begin{align}
    \widetilde{\calO}
=
\{\bo^{(1)},\dots,\bo^{(p)}\},
\end{align}
where \(\bo^{(j)}\) is the pattern missing only the \(j\)-th coordinate.
 Hence $S(\bo^{(j)}) = \{j\}$ and we arrive at
\begin{equation}\label{eq:h_ae_equal_new}
\pi^{\btheta}_j(\bz)=\pi^{\btheta^*}_j(\calT(\bz))
\qquad
\nu_{\mathbf 1}^{\btheta}\text{-a.e. }\bz,
\quad j=1,\dots,p.
\end{equation}
Since \(\nu_{\mathbf 1}^{\btheta}\ll \mu_{\bZ}\) and \( \mu_{\bZ} \ll \nu_{\mathbf 1}^{\btheta}\), this implies
\begin{equation}\label{eq:h_mu_equal_new}
\pi^{\btheta}_j(\bz)=\pi^{\btheta^*}_j(\calT(\bz))
\qquad
\mu_{\bZ}\text{-a.e. }\bz,
\quad j=1,\dots,p.
\end{equation}
\noindent\emph{Step 4: recover the latent prior.}
It remains to show that
\begin{align}
    \mu_{\bZ}=\calT_{\#}\mu_{\bZ}.
\end{align}
For \(\eta\in\{\btheta,\btheta^*\}\), define the strictly positive quantity
\begin{align}
\psi_{\eta}(\bz)
\coloneq
\bbP_{\eta}(\bO=\mathbf 1\mid \bZ=\bz)
=
\prod_{j=1}^p \sigma\bigl(\pi_j^{\eta}(\bz)\bigr)\in(0,1).
\end{align}
Since the latent prior is common under \(\btheta\) and \(\btheta^*\), for any
Borel set \(A\subseteq\mathbb R^d\),
\[
\mu_{\bZ}(A)
=
\bbP_{\btheta^*}(\bO=\mathbf 1)
\int_A \psi_{\btheta^*}(\bz^*)^{-1}\,
\nu_{\mathbf 1}^{\btheta^*}(\dd\bz^*).
\]
Using $\nu_{\mathbf 1}^{\btheta^*}=\calT_{\#}\nu_{\mathbf 1}^{\btheta}$ in  \eqref{eq:nu_one}, the equality
\(\bbP_{\btheta^*}(\bO=\mathbf 1)=\bbP_{\btheta}(\bO=\mathbf 1)\), and
\eqref{eq:h_mu_equal_new}, we obtain
\begin{align*}
\mu_{\bZ}(A)
&=
\bbP_{\btheta}(\bO=\mathbf 1)
\int_{\calT^{-1}(A)}
\psi_{\btheta^*}(\calT(\bz))^{-1}\,\nu_{\mathbf 1}^{\btheta}(\dd\bz)\\
&=
\bbP_{\btheta}(\bO=\mathbf 1)
\int_{\calT^{-1}(A)}
\psi_{\btheta}(\bz)^{-1}\,\nu_{\mathbf 1}^{\btheta}(\dd\bz)\\
&=
\mu_{\bZ}(\calT^{-1}(A)).
\end{align*}
Hence
\begin{align}\label{eq:mu_z_transformation}
    \mu_{\bZ}=\calT_{\#}\mu_{\bZ}.
\end{align}
\noindent \emph{Step 5: reconstruct the full one-time law. }
Finally, let \(\bZ\sim\mu_{\bZ}\) and set \(\widetilde{\bZ}\coloneq\calT(\bZ)\).
By Step 4, \(\widetilde{\bZ}\sim\mu_{\bZ}\), the common latent prior under
\(\btheta^*\). Moreover, by \eqref{eq:f_relation_new} and
\eqref{eq:h_mu_equal_new},
\[
f_x^{\btheta^*}(\widetilde{\bZ})=f_x^{\btheta}(\bZ)
\qquad\text{a.s.},
\]
and
\[
\pi_j^{\btheta^*}(\widetilde{\bZ})=\pi_j^{\btheta}(\bZ)
\qquad\text{a.s.},\quad j=1,\dots,p.
\]
Therefore,
\[
\law_{\btheta^*}(\bO\mid \widetilde{\bZ})
=
\law_{\btheta}(\bO\mid \bZ)
\qquad\text{a.s.}.
\]
Consequently,
\[
\law_{\btheta^*}\bigl(f_x^{\btheta^*}(\widetilde{\bZ}),\bO\bigr)
=
\law_{\btheta}\bigl(f_x^{\btheta}(\bZ),\bO\bigr).
\]
Adding the same noise yields
\[
Q_1^{\btheta}\ = \ Q_1^{\btheta^*}.
\]
\end{proof}

%%%%%%%%%%%%%%

\subsection{Proof of lemma~\ref{lem:full_law_finiteT_main}}
\label{pf:lem:full_law_finiteT_main}

\begin{proof}
For \(\eta\in\{\btheta,\btheta^*\}\), write
\[
\mu_{1:T}^{\eta}\coloneq \law_{\eta}(\bZ_{1:T}).
\]
We proceed in three steps.

\medskip
\noindent\emph{Step 1: identify the one-time emission and missingness mechanisms.}
Since
\[
\law_{\btheta}(\bX_{1:T}^{\,o},\bO_{1:T})
=
\law_{\btheta^*}(\bX_{1:T}^{\,o},\bO_{1:T}),
\]
their one-time marginals at \(t=1\) also coincide:
\[
\law_{\btheta}(\bX_1^{\,o},\bO_1)
=
\law_{\btheta^*}(\bX_1^{\,o},\bO_1).
\]
Applying Lemma~\ref{lem:full_law_static_main} yields a latent
reparameterization \(\calT:\bbR^d\to\bbR^d\) such that
\begin{align}\label{result:T=1}
\calT_{\#}\mu_{\bZ}=\mu_{\bZ},\qquad
f_x^{\btheta}=f_x^{\btheta^*}\circ\calT,\qquad
\pi_j^{\btheta}=\pi_j^{\btheta^*}\circ\calT,\quad j=1,\dots,p,
\quad \mu_{\bZ}\text{-a.e.}
\end{align}
% {By Lemma~\ref{lemma:positivity -latent-density}, every time-specific latent
% state \(\bZ_t\) has a nondegenerate density on \(\bbR^d\). Therefore the
% one-time identities in \eqref{result:T=1} can be applied at each time point:
% for \(t=1,\ldots,T\),
% \[
% f_x^{\btheta}(\bZ_t)=f_x^{\btheta^*}\{\calT(\bZ_t)\},
% \qquad
% \pi_j^{\btheta}(\bZ_t)=\pi_j^{\btheta^*}\{\calT(\bZ_t)\},\quad j=1,\ldots,p,
% \]
% almost surely under the finite trajectory law.}
Since \(\mu_{\bZ}\) has an almost-everywhere positive density, the functional
identities in \eqref{result:T=1} hold for Lebesgue-almost every \(\bz\).
Condition (vi) of Assumption~\ref{assum:static_base_main} then implies that,
under either model, they hold almost surely at every time point:
for \(t=1,\ldots,T\),
\[
f_x^{\btheta}(\bZ_t)=f_x^{\btheta^*}\{\calT(\bZ_t)\},
\qquad
\pi_j^{\btheta}(\bZ_t)=\pi_j^{\btheta^*}\{\calT(\bZ_t)\},
\quad j=1,\ldots,p.
\]

\medskip
\noindent\emph{Step 2: align the latent trajectory law.}
Let
\[
E_T\coloneq\{\bO_1=\mathbf 1,\dots,\bO_T=\mathbf 1\}
\]
be the fully observed mask event. Since
\[
\law_{\btheta}(\bX_{1:T}^{\,o},\bO_{1:T})
=
\law_{\btheta^*}(\bX_{1:T}^{\,o},\bO_{1:T}),
\]
their restrictions to the event \(E_T\) coincide. On \(E_T\), the full vector
\(\bX_{1:T}\) is observed. Define
\[
F_T^{\eta}:(\bbR^{d})^{T}\to(\bbR^{p})^{T},
\qquad
F_T^{\eta}(\bz_{1:T})
=
\bigl(f_x^{\eta}(\bz_1),\dots,f_x^{\eta}(\bz_T)\bigr),
\qquad \eta\in\{\btheta,\btheta^*\},
\]
and define the coordinatewise lift of \(\calT\) by
\[
\calT_T:(\bbR^{d})^{T}\to(\bbR^{d})^{T},
\qquad
\calT_T(\bz_{1:T})=(\calT(\bz_1),\dots,\calT(\bz_T)).
\]
Then
\[
F_T^{\btheta}
=
F_T^{\btheta^*}\circ\calT_T
\qquad
\mu_{1:T}^{\btheta}\text{-a.e.}
\]
Let
\[
\nu_{E}^{\btheta} \coloneq
\law_{\btheta}(\bZ_{1:T} \mid E_T),
\qquad
\nu_{E}^{\btheta^*} \coloneq
\law_{\btheta^*}(\bZ_{1:T} \mid E_T).
\]
Because the output noises are independent with known law and non-vanishing characteristic function, deconvolution on the event \(E_T\) yields
\begin{align}
    (F_T^{\btheta})_{\#}\nu^{\btheta}_{E}
=
(F_T^{\btheta^*})_{\#}\nu^{\btheta^*}_{E}.
\end{align}
Since \(F_T^{\btheta^*}\) is one-to-one and \(F_T^{\btheta}
=
F_T^{\btheta^*}\circ\calT_T\) on a set of full \(\nu_{E}^{\btheta}\)-measure,
\begin{align}\label{eq:latent_align_T}
    \nu_{E}^{\btheta^*} = (\calT_{T})_{\#}\nu_{E}^{\btheta}.
\end{align}
Next, define
\[
\psi_{\eta}(\bz_{1:T})
\coloneq
\prod_{t=1}^{T}\prod_{j=1}^{p}\sigma\bigl(\pi_j^{\eta}(\bz_t)\bigr),
\qquad \eta\in\{\btheta,\btheta^*\}.
\]
Bayes' rule gives
\begin{align}
    \nu^{\btheta}_{E}(d\bz_{1:T}) \cdot \bbP_{\btheta}(E_T)
=
\psi_{\btheta}(\bz_{1:T})\,\mu_{1:T}^{\btheta}(d\bz_{1:T}),
\qquad
\nu^{\btheta^*}_{E}(d\bz_{1:T})\cdot \bbP_{\btheta^*}(E_T)
=
\psi_{\btheta^*}(\bz_{1:T})\,\mu_{1:T}^{\btheta^*}(d\bz_{1:T}).
\label{eq:latent_align_T2}
\end{align}
By Step 1,
\[
\psi_{\btheta}
=
\psi_{\btheta^*}\circ\calT_T
\qquad
\mu_{1:T}^{\btheta}\text{-a.e.}
\]
Since \(\sigma(u)\in(0,1)\) for every \(u\in\bbR\), both \(\psi_{\btheta}\) and
\(\psi_{\btheta^*}\) are strictly positive. Using \eqref{eq:latent_align_T},
\eqref{eq:latent_align_T2}, and
\(\bbP_{\btheta}(E_T)=\bbP_{\btheta^*}(E_T)\), we obtain, for every Borel set
\(A\subseteq(\bbR^{d})^{T}\),
\begin{align*}
\mu_{1:T}^{\btheta^*}(A)
&= \bbP_{\btheta^*}(E_T)
\int_A \psi_{\btheta^*}(\bz_{1:T})^{-1}\,\nu^{\btheta^*}_{E}(d\bz_{1:T})\\
&= \bbP_{\btheta^*}(E_T)
\int_A \psi_{\btheta^*}(\bz_{1:T})^{-1}\,((\calT_T)_{\#}\nu^{\btheta}_{E})(d\bz_{1:T})\\
&= \bbP_{\btheta^*}(E_T)
\int_{\calT_T^{-1}(A)}
\psi_{\btheta^*}(\calT_T(\bz_{1:T}))^{-1}\,\nu^{\btheta}_{E}(d\bz_{1:T})\\
&=\bbP_{\btheta}(E_T)
\int_{\calT_T^{-1}(A)}
\psi_{\btheta}(\bz_{1:T})^{-1}\,\nu^{\btheta}_{E}(d\bz_{1:T})\\
&=
\mu_{1:T}^{\btheta}(\calT_T^{-1}(A)).
\end{align*}
Thus
\begin{align}\label{eq:mu_z_transformation_T}
\mu_{1:T}^{\btheta^*} \ = \ (\calT_T)_{\#}\mu_{1:T}^{\btheta}.
\end{align}

\medskip
\noindent\emph{Step 3: reconstruct the full trajectory law.}
For \(\bz_{1:T}\in(\bbR^{d})^{T}\), let \(K_{\eta}(\bz_{1:T},\cdot)\) be the
conditional law of \((\bX_{1:T},\bO_{1:T})\) given
\(\bZ_{1:T}=\bz_{1:T}\) under \(\eta\). \syb{By the coordinatewise identities from
Step 1 and the common noise laws,}
\begin{align}
    K_{\btheta}(\bz_{1:T},\cdot)
=
K_{\btheta^*}(\calT_T(\bz_{1:T}),\cdot)
\qquad
\mu_{1:T}^{\btheta}\text{-a.e.}
\end{align}
Let \(\varphi\) be any bounded measurable function on $(\bbR^{p})^{T}\times(\{0,1\}^{p})^{T}$.
Then
\begin{align*}
\int \varphi d Q_{T}^{\btheta}
&=
\int K_{\btheta}(\bz_{1:T},\varphi)\,\mu_{1:T}^{\btheta}(d\bz_{1:T})\\
&=
\int K_{\btheta^*}(\calT_T(\bz_{1:T}),\varphi)\,\mu_{1:T}^{\btheta}(d\bz_{1:T})\\
&=
\int K_{\btheta^*}(\bz_{1:T},\varphi)\,\mu_{1:T}^{\btheta^*}(d\bz_{1:T})\\
&=
\int \varphi d Q_{T}^{\btheta^*}.
\end{align*}
Therefore,
\begin{align}
    Q_T^{\btheta}  = Q_{T}^{\btheta^*}.
\end{align}
\end{proof}
% %%%%%%%%%%%%%%

%%%%%%%%%%%% proof for k-order

%%%%%%%%%%%% Lemmas
\section{Auxiliary lemmas and their proofs}\label{sec:appendix:lemmas}
Recall that $\wbX_{1:T} = (\bX_{1:T}, \bO_{1:T})$, and that $\bX_{1:T}^{\text{obs}}$ is the observed part of $\bX_{1:T}$.

\begin{lemma}\label{lemma:factorization}
   Suppose model \eqref{model:1} holds. Then, the joint distribution can be factored as follows: 
    \bes
    p\blp \bZ_{1:T}, \wbX_{1:T}, \bX_{1:T}^{\text{obs}} \brp  = p(\bZ_1)p(\wbX_1\mid \bZ_1)p(\bX_1^{\text{obs}}\mid \wbX_1) \cdot \prod_{s = 2}^T p(\bZ_s\mid \bZ_{s-1})p(\wbX_s\mid \bZ_{s}) p(\bX_s^{\text{obs}}\mid \wbX_s)
    \ees
\end{lemma}
Lemma \ref{lemma:factorization} follows directly from the factorization property of hidden Markov model \citep{wainwright2008graphical}. This factorization of the joint law can be represented graphically as shown in the Figure \ref{fig:dag}. Analogously, the generalization to $k$-order is as follows.
\begin{lemma}\label{lemma:factorization-k}
   Suppose model \eqref{model:k-order} holds. Then, the joint distribution can be factored as follows: 
    \bes
    && p\blp \bZ_{1:T}, \wbX_{1:T}, \bX_{1:T}^{\text{obs}} \brp  \cr
    &=& p(\bZ_1)p(\wbX_1\mid \bZ_1)p(\bX_1^{\text{obs}}\mid \wbX_1) \cdot \prod_{s = 2}^T p(\bZ_s\mid \bZ_{(s-k)_{+}:s-1})p(\wbX_s\mid \bZ_{s}) p(\bX_s^{\text{obs}}\mid \wbX_s)
    \ees
    where $(s-k)_{+} \coloneq \max\{1, s-k\}$.
\end{lemma}
This follows from the same graphical factorization argument as
Lemma~\ref{lemma:factorization}, with the latent parent of \(\bZ_s\) replaced
by \(\bZ_{(s-k)_+:s-1}\).
\begin{lemma}\label{lemma:noise_outsourcing}(Noise-outsourcing lemma.)
	Let $(\bX,\bY)$ be a random vector. Suppose random variables $\hbX \overset{\text{d}}{=} \bX$ and $\deta \sim \calN(\bzero,\bI_{m})$ for any given $m\geq 1$,  which are independent of $\hbX$. Then there exists a random variable $\hbY$ and a measurable function $g$ such that $\hbY= g(\deta, \hbX)$ and 
    \bes
    (\bX,\bY)\overset{\text{d}}{=}(\hbX,\hbY).
    \ees
    
\end{lemma}
This result is a Gaussian extension of the noise-outsourcing lemma (Theorem~5.10 in \citealt{Kallenberg2021FoundationsOM}), which originally assumes uniform noise. A proof for the Gaussian case can be found in Lemma 2.1 of \citet{zhou2023deep}.
% \begin{lemma}[Ergodicity preserved under measurable observation]\label{lemma:ergodicity}

% \begin{lemma}[Glivenko-Cantelli via Bracketing]\label{lemma:gc}
%     Let $\mathcal{F}$ be a class of measurable functions and let $P$ be a probability measure. If the $L_1(P)$ bracketing number satisfies
%     %
%     \[
%     N_{[ ]}(\varepsilon, \mathcal{F}, L_1(P)) < \infty \quad \text{for every } \varepsilon > 0,
%     \]
%     %
%     then $\mathcal{F}$ is a P-Glivenko-Cantelli class. That is,
%     %
%     \[
%     \| P_n - P \|_{\mathcal{F}} = \sup_{f \in \mathcal{F}} \left| \frac{1}{n} \sum_{i=1}^n f(X_i) - \int f \, dP \right| \xrightarrow{\text{a.s.}} 0.
%     \]
% \end{lemma}
% Lemma \ref{lemma:gc} is .
%%%%%%%%%%%%%%%%%%%%%%%%%%%%% consistency of M-estimator
% \begin{lemma}[Consistency of M-estimators]\label{lemma:m_estimator_consistency}
%     Let \( M_n \) be random functions and let \( M \) be a fixed function of \( \btheta \) such that for every \( \varepsilon > 0 \),
%     %
%     \begin{align*}
%         \sup_{\btheta \in \Theta} \left| M_n(\btheta) - M(\btheta) \right| &\xrightarrow{p} 0, \\
%         \sup_{\btheta: d(\btheta, \btheta_0) \geq \varepsilon} M(\btheta) &< M(\btheta_0).
%     \end{align*}
%     %
%     Then any sequence of estimators \( \hat{\btheta}_n \) satisfying
%     %
%     \[
%         M_n(\hat{\btheta}_n) \geq M_n(\btheta_0) - o_p(1)
%     \]
%     %
%     converges in probability to \( \btheta_0 \).
% \end{lemma}
%%%%%%%%%%%%%%----
\begin{lemma}[Network growth and moments]
\label{lemma:network_moments}
Under Assumption~\ref{assum:model_encoder_regularity}, the model-generated
trajectories, the recursively encoder-assisted trajectories, and the encoder
outputs on model-generated histories are uniformly square integrable over
\(\Theta\times\Phi\). Moreover, for some \(C_T<\infty\), every fixed history
\(\wbx_{1:t-1}^{\mathrm{obs}}\) satisfies
\begin{align}
\sup_{(\btheta,\bphi)\in\Theta\times\Phi}
\bbE_{Q_{t\mid1:t-1}^{\btheta,\bphi}
(\cdot\mid\wbx_{1:t-1}^{\mathrm{obs}})}\|\bX_t\|_2^2
\le
C_T\left(1+\sum_{s<t}
\|\bx_s^{\mathrm{obs}}\odot\bo_s\|_2^2\right),
\label{eq:conditional-network-moment}
\end{align}
with \(\bphi\) omitted when \(t=1\).
\end{lemma}
%%%%%%%%%%%---
\begin{lemma}[Continuity of the model-implied laws]
\label{lemma:theta_law_continuity}
Under Assumption~\ref{assum:model_encoder_regularity}, the maps
\[
\btheta\longmapsto Q_T^{\btheta},
\qquad
\btheta\longmapsto Q_{o,T}^{\btheta}
\]
are continuous on \(\Theta\) under \(W_1\) and \(W_{1,o}\), respectively, for
each fixed \(T<\infty\).
\end{lemma}
%%%%%%%%%%%---
\begin{lemma}[Continuity of the encoder calibration loss]
\label{lemma:encoder_calibration_loss_continuity}
Under Assumption~\ref{assum:model_encoder_regularity},
 for \(\btheta\in\Theta\), write
\(\calL_{\mathrm{syn},\enc}^{\btheta}(\bphi)\) for the calibration loss
generated under \(\btheta\). The map
\[
(\btheta,\bphi)\mapsto
\calL_{\mathrm{syn},\enc}^{\btheta}(\bphi)
\]
is jointly continuous on \(\Theta\times\Phi\).
\end{lemma}
%%%%%%%%%%%---
\begin{lemma}[Continuity of the encoder-assisted laws]
\label{lemma:encoder_assisted_continuity}
Under Assumption \ref{assum:model_encoder_regularity}, for each fixed \(T<\infty\),
let \(W_{1,o}\) be induced by the observed-space metric
\(\operatorname{dist}(\cdot,\cdot)\) defined as Eq.~\eqref{dist-proj-main}.
Then:
\begin{enumerate}[label=(\roman*)]
\item for every \(t\in[T]\) and every fixed history \(\wbx_{1:t-1}^{\text{obs}}\), the map
\[
(\btheta,\bphi)
\mapsto
Q_{o,t\mid1:t-1}^{\btheta,\bphi}(\cdot\mid\wbx_{1:t-1}^{\text{obs}})
\]
is continuous under \(W_{1,o}\), with \(\bphi\)
omitted when \(t=1\);

\item the map
\[
(\btheta,\bphi)\mapsto Q_{o,T}^{\btheta,\bphi}
\]
is continuous under \(W_{1,o}\);

\item for each fixed trajectory \(\wbx_{1:T}^{\text{obs}}\), the map
\[
(\btheta,\bphi)
\mapsto
\ell_{\btheta,\bphi}(\wbx_{1:T}^{\text{obs}})
\]
is continuous. Consequently, \((\btheta,\bphi)\mapsto\calL(\btheta,\bphi)\) is continuous on any compact subset of \(\Theta\times\Phi\).
\end{enumerate}
\end{lemma}

%%%%%%%%%%%
% \begin{lemma}[Positivity density of the latent trajectory law]\label{lemma:positivity -latent-density}
% Under Assumption~\ref{assum:static_base_main}, for each fixed \(T<\infty\) and each
% \(\eta\in\{\btheta,\btheta'\}\), the latent trajectory law
% \[
% \mu_{1:T}^{\eta}\coloneq\law_{\eta}(\bZ_{1:T}^{\eta})
% \]
% is absolutely continuous with respect to Lebesgue measure on \((\bbR^d)^T\).
% Moreover, its density is positive for Lebesgue-a.e. on 
% \((\bbR^d)^T\).
% \end{lemma}
%%%%%%%%%%%%%%%%%%
%%%%%%%%%%%%%%
\subsection{Proof of lemma \ref{lemma:network_moments}}
\begin{proof}
Compactness of the parameter spaces uniformly bounds all weights and biases.
Since the architectures are fixed, the ReLU networks and Transformers have
uniform linear-growth bounds. Thus, for constants independent of
\((\btheta,\bphi)\),
\[
\|\trans^{\btheta}(\bz,\bv)\|_2+\|f_x^{\btheta}(\bz)\|_2
\le C(1+\|\bz\|_2+\|\bv\|_2)
\]
and
\[
\|\enc^{\bphi}(\wbx_{1:t}^{\mathrm{obs}},\cdeps)\|_2
\le C_T\left(1+\sum_{s\le t}
\|\bx_s^{\mathrm{obs}}\odot\bo_s\|_2
+\|\cdeps\|_2\right),
\]
where the bounded mask inputs are absorbed into the constants. Using common
source noises for all \((\btheta,\bphi)\), define
\[
M_T=C_T\left\{1+\|\bZ_1\|_2
+\sum_{s=1}^{T-1}\|\bv_s\|_2
+\sum_{s=1}^{T}\bigl(\|\br_s\|_2+\|\cdeps_s\|_2\bigr)\right\}.
\]
Induction over the fixed horizon shows that the norm of every model-generated
or encoder-assisted state, observation, and encoder output is bounded by
\(M_T\), uniformly over \(\Theta\times\Phi\). Since the initial state and
reference noises have finite second moments, \(\bbE M_T^2<\infty\). This common
\(L^2\) envelope gives the claimed uniform square integrability. Applying the
same bounds conditionally on a fixed history gives
\eqref{eq:conditional-network-moment}.
\end{proof}
%%%%%%%%%%%%%%%
\subsection{Proof of lemma \ref{lemma:theta_law_continuity}}
\begin{proof}
Let \(\btheta_n\to\btheta\), and write
\[
\trans_n=\trans^{\btheta_n},\qquad
f_{x,n}=f_x^{\btheta_n},\qquad
\pi_{n,j}=\pi_j^{\btheta_n},
\qquad
\trans=\trans^{\btheta},\qquad
f_x=f_x^{\btheta},\qquad
\pi_j=\pi_j^{\btheta}.
\]
	Let the common source noises be
	\[
	\bA=(\bZ_1,\bv_1,\ldots,\bv_{T-1},\br_1,\ldots,\br_T),
	\]
	and let \(\{\bU_{t,j}:t\in[T],\,j\in[p]\}\) be independent uniform random
	variables, independent of \(\bA\).
For each \(n\), define recursively
\[
\bZ^{(n)}_1=\bZ_1,
\qquad
\bZ^{(n)}_{t+1}=\trans_n(\bZ^{(n)}_t,\bv_t),
\qquad
\bX^{(n)}_t=f_{x,n}(\bZ^{(n)}_t)+\br_t.
\]
Define \((\bZ_t,\bX_t)\) analogously under \(\btheta\). Continuity of the
fixed networks and induction over \(t\) give
	\begin{align}\label{eq:coupling-ZX}
	(\bZ^{(n)}_{1:T},\bX^{(n)}_{1:T})
	\longrightarrow
	(\bZ_{1:T},\bX_{1:T})
	\qquad\text{for common source noises }\bU.
	\end{align}
	Couple the masks by
	\[
	\bO_{t,j}^{(n)}
	=
	\mathds{1}\{\bU_{t,j}\le\sigma(\pi_{n,j}(\bZ_t^{(n)}))\},
	\qquad
	\bO_{t,j}
	=
	\mathds{1}\{\bU_{t,j}\le\sigma(\pi_j(\bZ_t))\}.
	\]
	The continuity of the logits gives \(\bO_{1:T}^{(n)}\to\bO_{1:T}\) almost
	surely. Hence the coupled full trajectories
\begin{align}\label{coupling:X}
    \{(\bX_1^{(n)},\bO_1^{(n)}),\ldots,(\bX_T^{(n)},\bO_T^{(n)})\} \longrightarrow \{(\bX_1,\bO_1),\ldots,(\bX_T,\bO_T)\} \qquad\text{a.s.}
\end{align}
	Lemma~\ref{lemma:network_moments} implies uniform integrability of the full
	trajectory distance. Therefore,
	\[
	W_1(Q_T^{\btheta_n},Q_T^{\btheta})\le \ \bbE \rho\left((\bX_{1:T}^{(n)},\bO_{1:T}^{(n)}),(\bX_{1:T},\bO_{1:T})\right) \ \to0.
	\]
	The corresponding observed trajectories
	\[
	\{(\bX_1^{(n)}\odot\bO_1^{(n)},\bO_1^{(n)}),\ldots,
	(\bX_T^{(n)}\odot\bO_T^{(n)},\bO_T^{(n)})\}
	\]
	also converge almost surely, and the projected trajectory distances are
	uniformly integrable by the same lemma. Thus
\begin{align}\label{W_1o:theta}
    	W_{1,o}(Q_{o,T}^{\btheta_n},Q_{o,T}^{\btheta})\to0.
\end{align}
	\end{proof}
%%%%%%%%%%%%%%%
\subsection{Proof of lemma \ref{lemma:encoder_calibration_loss_continuity}}
\begin{proof}
Let \((\btheta_n,\bphi_n)\to(\btheta,\bphi)\). Using the common source
noises and uniforms from the proof of
Lemma~\ref{lemma:theta_law_continuity}, couple synthetic trajectories generated
under \(\btheta_n\) and \(\btheta\) so that
\[
\left(\bZ_{1:T}^{\mathrm{syn},(n)},
\wbX_{1:T}^{\mathrm{obs},\mathrm{syn},(n)}\right)
\longrightarrow
\left(\bZ_{1:T}^{\mathrm{syn}},
\wbX_{1:T}^{\mathrm{obs},\mathrm{syn}}\right)
\qquad\text{a.s.}
\]
Let \(\cdeps_{t,1}\) and \(\cdeps_{t,2}\) be independent encoder noises,
independent of the synthetic trajectories and shared across \(n\). For
\(j=1,2\), define
\[
\hbZ_{t,j}^{(n)}
=
\enc^{\bphi_n}
\left(\wbX_{1:t}^{\mathrm{obs},\mathrm{syn},(n)},\cdeps_{t,j}\right),
\qquad
\hbZ_{t,j}
=
\enc^{\bphi}
\left(\wbX_{1:t}^{\mathrm{obs},\mathrm{syn}},\cdeps_{t,j}\right).
\]
Joint continuity of the encoder gives
\(\hbZ_{t,j}^{(n)}\to\hbZ_{t,j}\) almost surely.
For each \(t\), set
\[
H_{n,t}
\coloneq
\left\|\bZ_t^{\mathrm{syn},(n)}-\hbZ_{t,1}^{(n)}\right\|_2
-\frac12\left\|\hbZ_{t,1}^{(n)}-\hbZ_{t,2}^{(n)}\right\|_2,
\]
and define \(H_t\) analogously under \((\btheta,\bphi)\). Then
\(H_{n,t}\to H_t\) almost surely and
\[
\calL_{\mathrm{syn},\enc}^{\btheta_n}(\bphi_n)
=
\frac1T\sum_{t=1}^T\bbE H_{n,t}.
\]
Lemma~\ref{lemma:network_moments} gives uniform second-moment bounds for the
synthetic latent states and encoder outputs, so \(\{H_{n,t}\}_n\) is uniformly
integrable. The Vitali convergence theorem
\citep[Theorem~4.6.3]{durrett2019probability} therefore yields
\(\bbE H_{n,t}\to\bbE H_t\). Since \(T\) is fixed,
\[
\calL_{\mathrm{syn},\enc}^{\btheta_n}(\bphi_n)
\longrightarrow
\calL_{\mathrm{syn},\enc}^{\btheta}(\bphi),
\]
which proves joint continuity.
\end{proof}
%%%%%%%%%%%%%%%
\subsection{Proof of lemma \ref{lemma:encoder_assisted_continuity}}
\begin{proof}
Let \((\btheta_n,\bphi_n)\to(\btheta,\bphi)\).

For (i), fix \(t\ge2\) and \(\wbx_{1:t-1}^{\mathrm{obs}}\). Use the
common-noise constructions from the proofs of
Lemmas~\ref{lemma:theta_law_continuity} and
\ref{lemma:encoder_calibration_loss_continuity}. Joint continuity of the
encoder and model networks gives almost-sure
convergence of the coupled one-step observed draws. Their distances are
uniformly integrable by \eqref{eq:conditional-network-moment}. The same
coupling argument as in \eqref{W_1o:theta} therefore gives
\[
W_{1,o}\!\left(
Q_{o,t\mid1:t-1}^{\btheta_n,\bphi_n}
(\cdot\mid\wbx_{1:t-1}^{\mathrm{obs}}),
Q_{o,t\mid1:t-1}^{\btheta,\bphi}
(\cdot\mid\wbx_{1:t-1}^{\mathrm{obs}})
\right)\to0.
\]
For \(t=1\), this follows directly from
Lemma~\ref{lemma:theta_law_continuity}.

For (ii), apply the same construction recursively, feeding each generated
observed history into the encoder. The induction used in the proof of
Lemma~\ref{lemma:theta_law_continuity}, now also using joint continuity of the
encoder, gives almost-sure convergence of the coupled observed trajectories.
Lemma~\ref{lemma:network_moments} makes their distances uniformly integrable;
hence, again as in \eqref{W_1o:theta},
\[
W_{1,o}\left(Q_{o,T}^{\btheta_n,\bphi_n},
Q_{o,T}^{\btheta,\bphi}\right)\to0.
\]

For (iii), fix \(\wbx_{1:T}^{\mathrm{obs}}\). Part (i) and the masked-score
continuity argument in Step~1 of the proof of
Proposition~\ref{prop:partial_law_main} show that every summand defining
\(\ell_{\btheta,\bphi}(\wbx_{1:T}^{\mathrm{obs}})\) is continuous. Since
\(T\) is fixed, so is \(\ell_{\btheta,\bphi}(\wbx_{1:T}^{\mathrm{obs}})\).
Moreover, \eqref{supp:S-mask_bounded} and
\eqref{eq:conditional-network-moment} give the integrable envelope
\[
|\ell_{\btheta,\bphi}(\wbx_{1:T}^{\mathrm{obs}})|
\le C_T\left(1+\sum_{t=1}^T
\|\bx_t^{\mathrm{obs}}\odot\bo_t\|_2\right).
\]
Its integrability under the true law follows from
Lemma~\ref{lemma:network_moments}, so dominated convergence proves the stated
continuity of \(\calL\).
\end{proof}

\section{Additional experimental details}\label{supp:sec:experiment}
\subsection{Tuning parameters selection}\label{supp:subsec:implementation}
\begin{figure}
    \centering
    \includegraphics[width=1\linewidth]{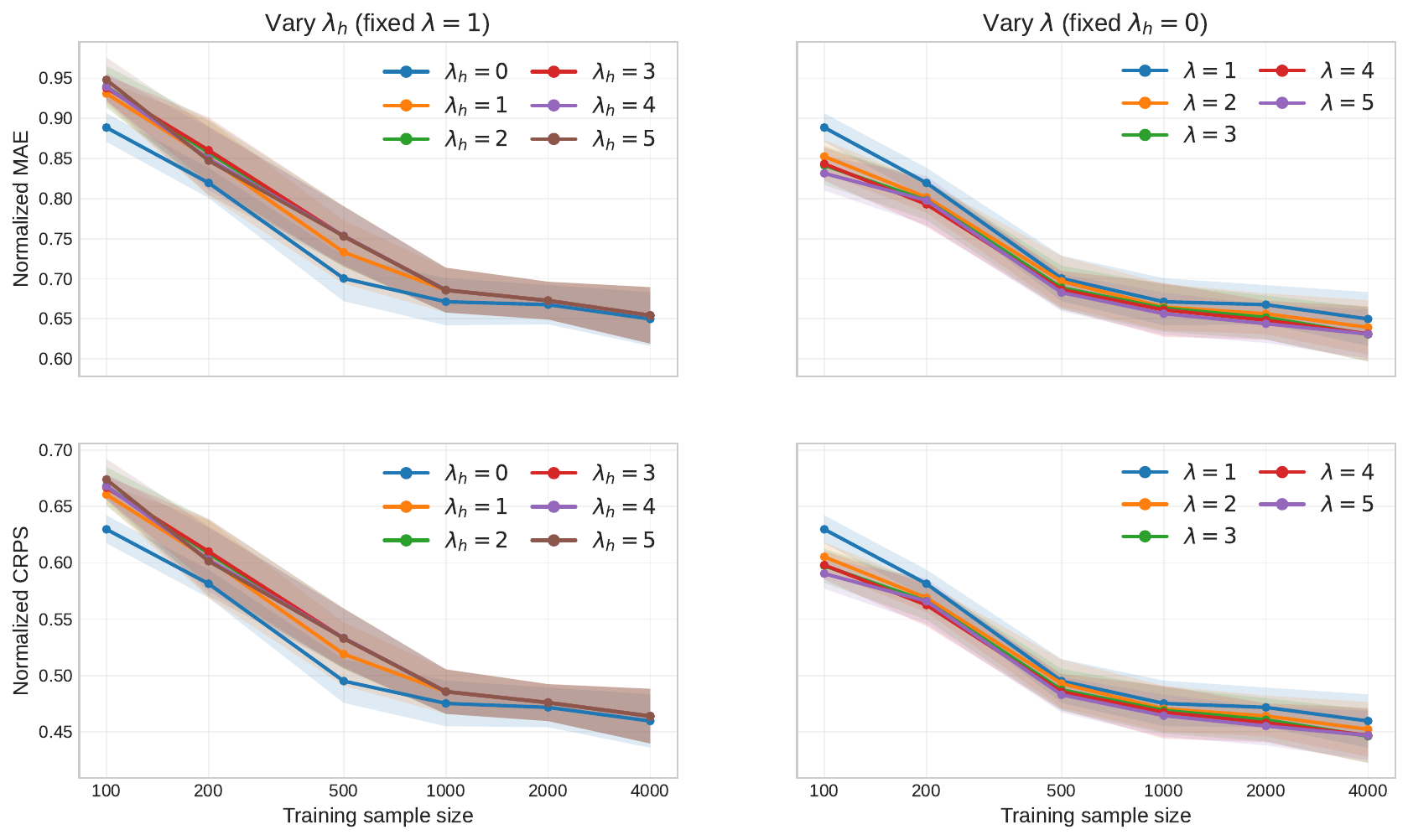}
    \caption{Hyperparameter sensitivity for Case 3. The left panels fix $\lambda=1$ and vary $\lambda_{\enc}\in\{0,1,2,3,4,5\}$, while the right panels fix $\lambda_{\enc}=0$ and vary $\lambda\in\{1,2,3,4,5\}$. The top and bottom rows report normalized MAE and normalized CRPS, respectively. The experimental setting uses $p=20$, $d=5$, $T=10$, and an LDM mechanism with $50\%$ missingness.}
    \label{fig:lambda_vs_lambda_enc0}
\end{figure}
\begin{figure}
    \includegraphics[width=1\linewidth]{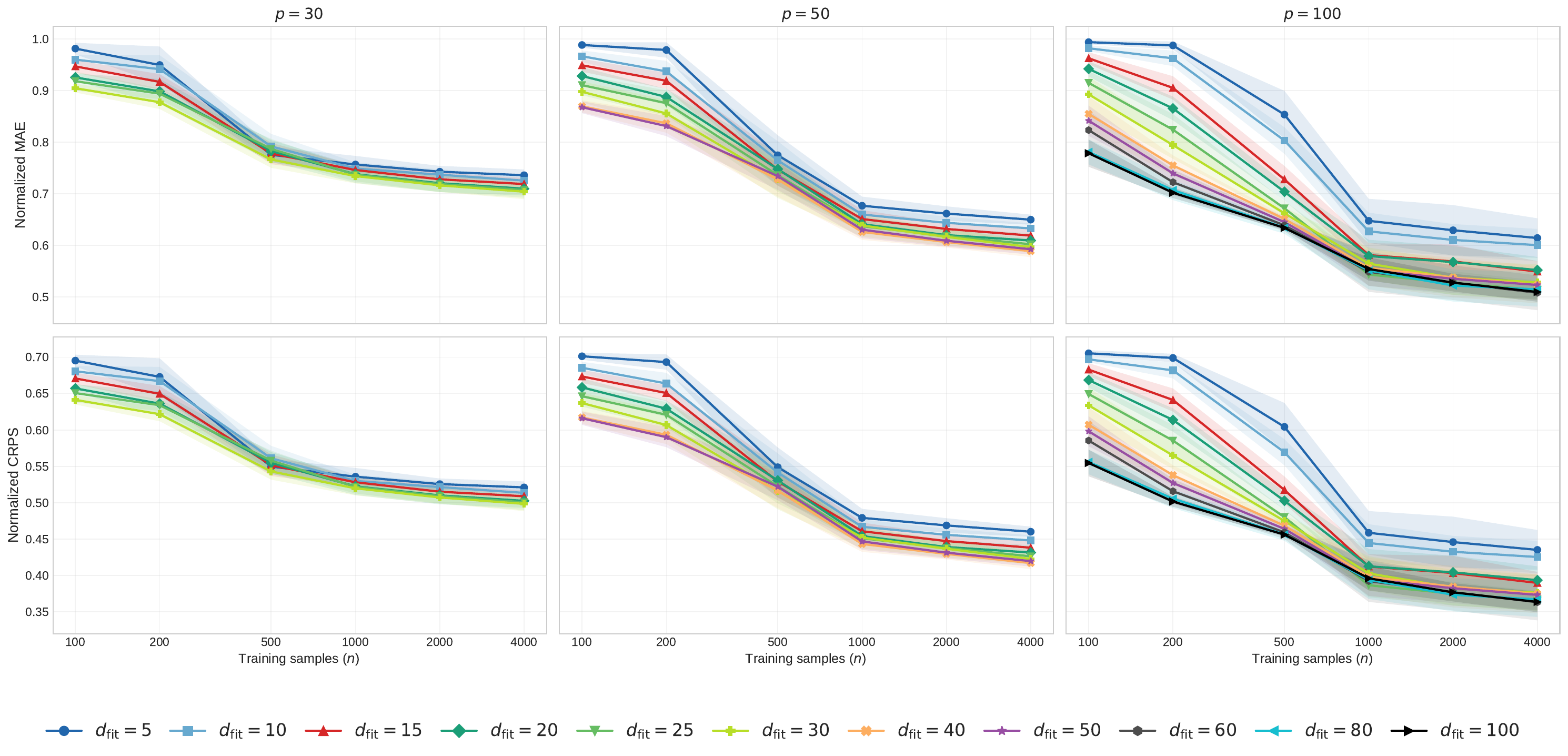}
    \caption{Sensitivity to the fitted latent dimension $d_{\operatorname{fit}}$ for Case 3 under an LDM mechanism with true latent dimension $d=15$ and $50\%$ missingness. Columns correspond to $p\in\{30,50,100\}$, while the top and bottom rows report normalized MAE and normalized CRPS, respectively. Each curve represents a different value of $d_{\operatorname{fit}}$.}
    \label{fig:d-fit-sensitivity}
\end{figure}
\noindent \textbf{Tuning of $\lambda$ and $\lambda_{\enc}$}. We conducted a preliminary grid experiment
in Case~3 of Section~\ref{sec:simulation}, with
\(\lambda\in\{1,2,3,4,5\}\) and
\(\lambda_{\enc}\in\{0,1,2,3,4,5\}\). As shown in Figure~\ref{fig:lambda_vs_lambda_enc0}, fixed \(\lambda_{\enc} = 0\), the empirical performance was relatively
insensitive to \(\lambda\). Moreover, adding the synthetic encoder-calibration
loss did not yield a systematic improvement in imputation
performance, while incurring the additional computational cost of generating
synthetic trajectories. We therefore set \(\lambda=5\) and
\(\lambda_{\enc}=0\) in all remaining simulation and real-data experiments.
\\~\\
% \midskip
\noindent \textbf{Tuning of latent dimension $d_{\operatorname{fit}}$}.
For each observed dimension $p\in\{30,50,100\}$, we generated data from Case 3 with true latent dimension $d=15$, sequence length $T=10$, and an LDM mechanism with $50\%$ missingness. We varied the fitted latent dimension over $d_{\operatorname{fit}}\in\{5,10,15,20,25,30\}$ for $p=30$, added $\{40,50\}$ for $p=50$, and further considered $\{40,50,60,80,100\}$ for $p=100$. The training sample size was varied over $n\in\{100,200,500,1000,2000,4000\}$, while the test size was fixed at 200. Each configuration was repeated over 10 random seeds, with all other training and inference settings held fixed; the curves and shaded regions report the seed-wise mean and one standard deviation, respectively. As shown in Figure~\ref{fig:d-fit-sensitivity}, both normalized MAE and normalized CRPS decrease substantially as $n$ increases. Under-specified models, particularly those with $d_{\operatorname{fit}}\in\{5,10\}$, generally perform worse, whereas moderate over-specification is consistently benign and often beneficial. At $n=4000$, the best results are obtained around $d_{\operatorname{fit}}=30$, $40$, and $60$ for $p=30$, $50$, and $100$, respectively, while larger values remain competitive. These findings indicate that accurate imputation does not require recovering the
true latent dimension exactly. In particular, moderate over-specification appears to be robust, suggesting that $d_{\operatorname{fit}}$ can be fixed at a moderately large value without additional tuning.

%%%%%%%%%%%%%%%%%%%%%%%%%%%%%%%%%%%%%%%%%%%%%%%%%%%%%%%%%%%%%%%%%%%%%%%%%%%%%%%%%%%%%%%%%%%%%%%%%%%%%%%%%%
%%%%%%%%%%%%%%%%%%%%%%%%%%%%%%%%%%%%%%%%%%%%%%%%%%%%%%%%%%%%%%%%%%%%%%%%%%%%%%%%%%%%%%%%%%%%%%%%%%%%%%%%%%%%%%%%%%%%%%%%%%%%%%%%%%%%%%%%%%%%%%%%%%%%%%%%%%%%%%%%%%%%%%%%%%%%%%%%%%%%%%%%%%%%%%%%%%%%%%%%%%%%%%%%%%%%%%%%%%%%%%%%%%%%%%%%%%%%%%%%%%%%%%%%%%%%%%%%%%%%%%%%%%%%%%%%%%%%%%%%%%%%%%%%%%%%%%%%%%%%%%%%%%%%%%%%%%%%%%%%%%%%%%%%%%%%%%%%%%%%%%%%%%%%%%%%%%%%%%%%%%%%%%%%%%%%%%%%%%%%%%%%%%%%%%%%%%%%%%%%%%%%%%%%%%%%%%%%%%%%%%%%
\subsection{Simulation studies}\label{supp:subsec:simulation}
\begin{figure}[t]
    \centering
    \includegraphics[width=1\linewidth]{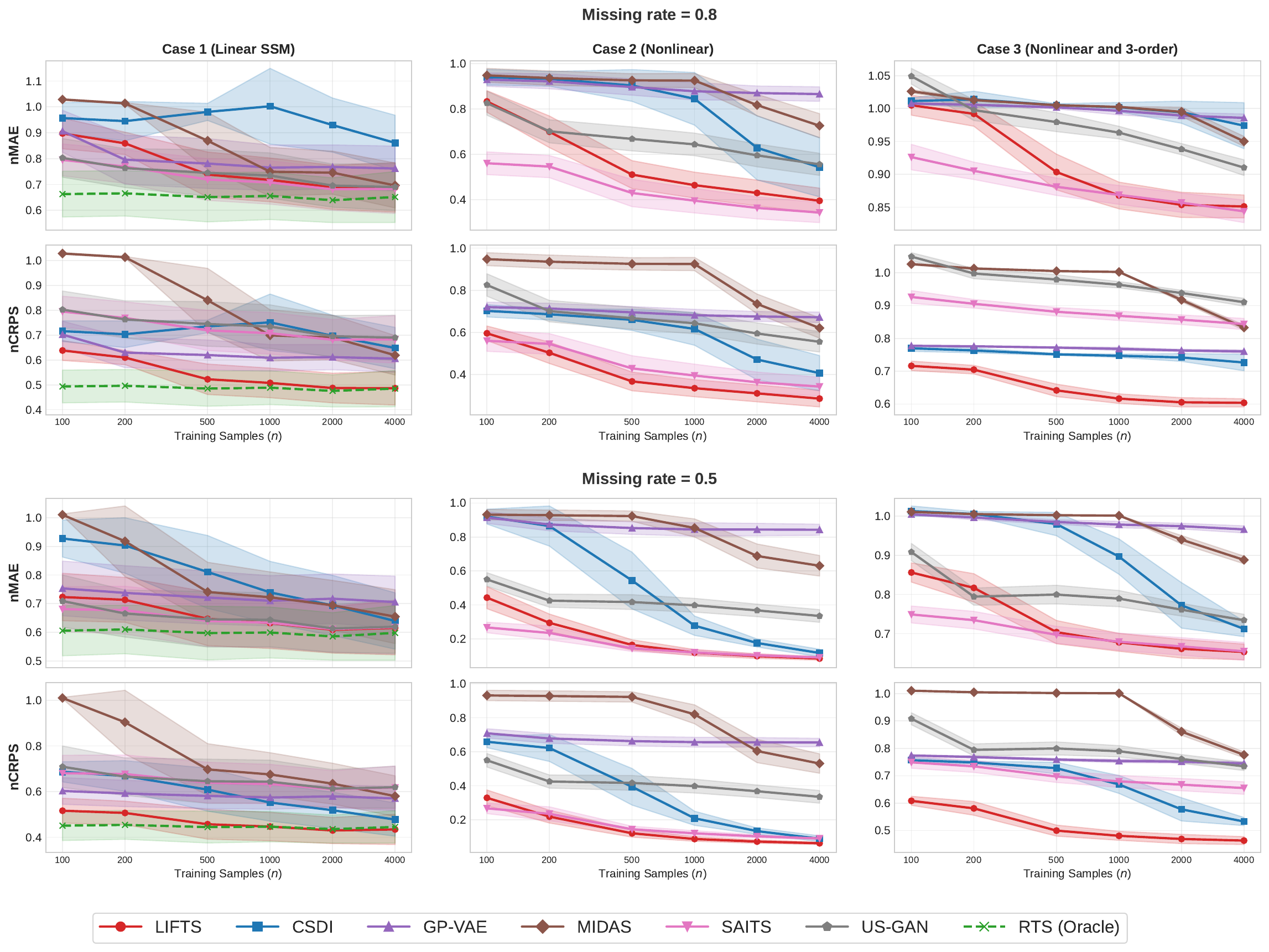}
    \caption{Performance across Cases 1--3 with increasing training sample size under MCAR. The top two rows correspond to an $80\%$ missingness rate and the bottom two rows to a $50\%$ missingness rate; within each pair, the rows report nMAE and nCRPS. Lines and shaded bands denote the mean and one standard deviation over 10 random seeds.}
    \label{fig:sim-mcar-obs-combined}
\end{figure}
\begin{figure}
    \includegraphics[width=1\linewidth]{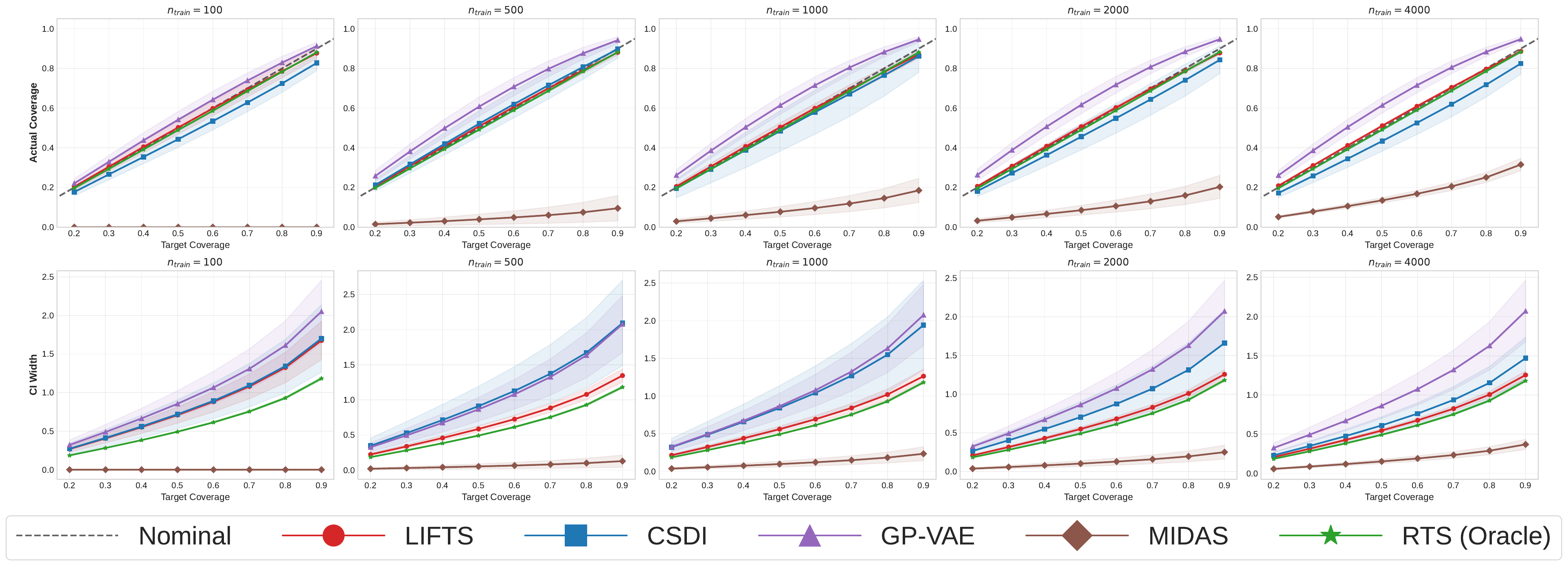}
    \caption{Empirical coverage and corresponding average prediction interval width under Case 1 using the MCAR mechanism with an $80\%$ missing rate.}
    \label{fig:sim-coverage-case1-mcar-obs0.2}
\end{figure}
\begin{figure}[t]
    \includegraphics[width=1\linewidth]{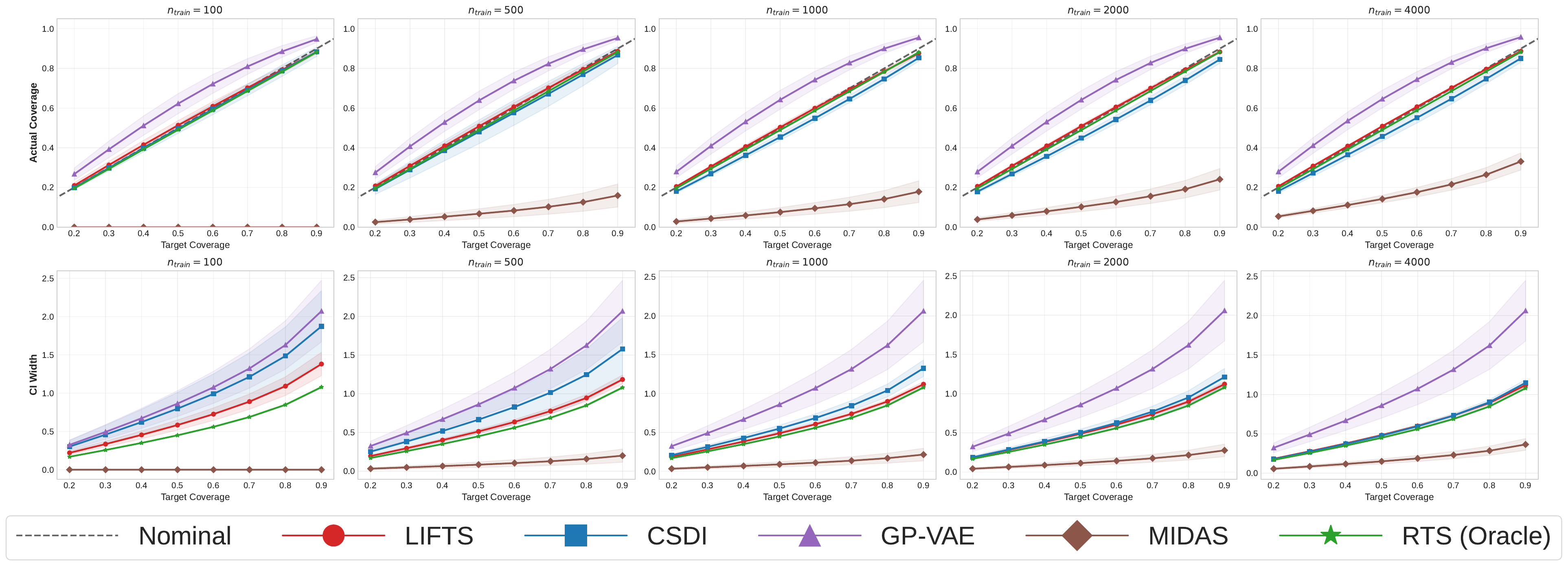}
    \caption{Empirical coverage and corresponding average prediction interval width under Case 1 using the MCAR mechanism with a $50\%$ missing rate.}
    \label{fig:sim-coverage-case1-mcar-obs0.5}
\end{figure}
In this subsection, we present additional results for the simulation studies.  Crucially, for the linear Gaussian model (Case 1), the optimal posterior remains analytically tractable under MCAR. Therefore, we employ the exact Rauch--Tung--Striebel smoothing (RTS; \citealp{rauch1965maximum}) as an oracle benchmark for imputation.

Figure \ref{fig:sim-mcar-obs-combined} shows the performance of the three cases under an MCAR mechanism with  $80\%$ and $50\%$ missing rate. The results are consistent with the main claims established in Section \ref{sec:simulation} under LDM. In particular, in case 1, we observe that the proposed method achieves comparable performance to the RTS smoother, which serves as an oracle benchmark. This demonstrates the effectiveness of our approach in handling missing data under MCAR.

Furthermore, Figures \ref{fig:sim-coverage-case1-mcar-obs0.2} and \ref{fig:sim-coverage-case1-mcar-obs0.5} report the empirical coverage results under MCAR with $80\%$ and $50\%$ missing rates, respectively. Across both settings, the coverage achieved by our proposed method is consistently closer to the nominal level than that of the baseline methods and is comparable to that of the oracle RTS smoother. These results demonstrate that our method provides reliable uncertainty quantification even under high rates of missingness.

%%%%%%%%%%%%%%%%%%%%%%%%%%%%%%%%%%%%%%%%%%%%%%%%%%%%%%%%%%%%%%%%%%%%%%%%%%%%%%%%%%%%%%%%%%%%%%%%%%%%%%%%%%%%%%%%%%%%%%%%%%%%%%%%%%%%%%%%%%%%%%%%%%%%%%%%%%%%%%%%%%%%%%%%%%%%%%%%%%%%%%%%%%%%%%%%%%%%%%%%%%%%%%%%%%%%%%%%%%%%%%%%%%%%%%%%%%%%%%%%%%%%%%%%%%%%%%%%%%%%%%%%%%%%%%%%%%%%%%%%%%%%%%%%%%%%%%%%%%%%%%%%%%%%%%%%%%%%%%%%%%%%%%%%%%%%%%%%%%%%%%%%%%%%%%%%%%%%%%%%%%%%%%%%%%%%%%%%%%%%%%%%%%%%%%%%%%%%%%%%%%%%%%%%%%%%%%%%%%%%%%%%%%%%%%%%%%%%%%%%%%%%%%%%%%%%%%%%%%%%%%%%%%%%%%%%%%%%%%%%%%%%%%%%%%%%%%%%%%%%%%%%%%%%%%%%%%%%%%%%%%%%%%%%%%
\subsection{Real data analysis}\label{supp:subsec:real_data_analysis}
In this subsection, we present additional results for the real data analysis
in Section \ref{sec:real-data}. 

Figure \ref{fig:coverage-forecasting} reports the empirical coverage and average
prediction-interval width for test-set forecasts over \(t=37,\ldots,48\) across
different training sample sizes. LIFTS achieves empirical coverage closest to
the nominal levels while maintaining relatively narrow prediction intervals,
although it exhibits mild overcoverage at the lower nominal levels. In contrast,
CSDI and GP-VAE tend to undercover, with the undercoverage being particularly
pronounced for MIDAS.

All experiments were conducted under the same hardware configuration, using a single NVIDIA A6000 GPU.
Figure \ref{fig:real-time} presents the training and inference time comparisons across all methods. LIFTS requires training time comparable to CSDI but is approximately 100 times faster at inference.

\begin{figure}
    \centering
    \includegraphics[width=1\linewidth]{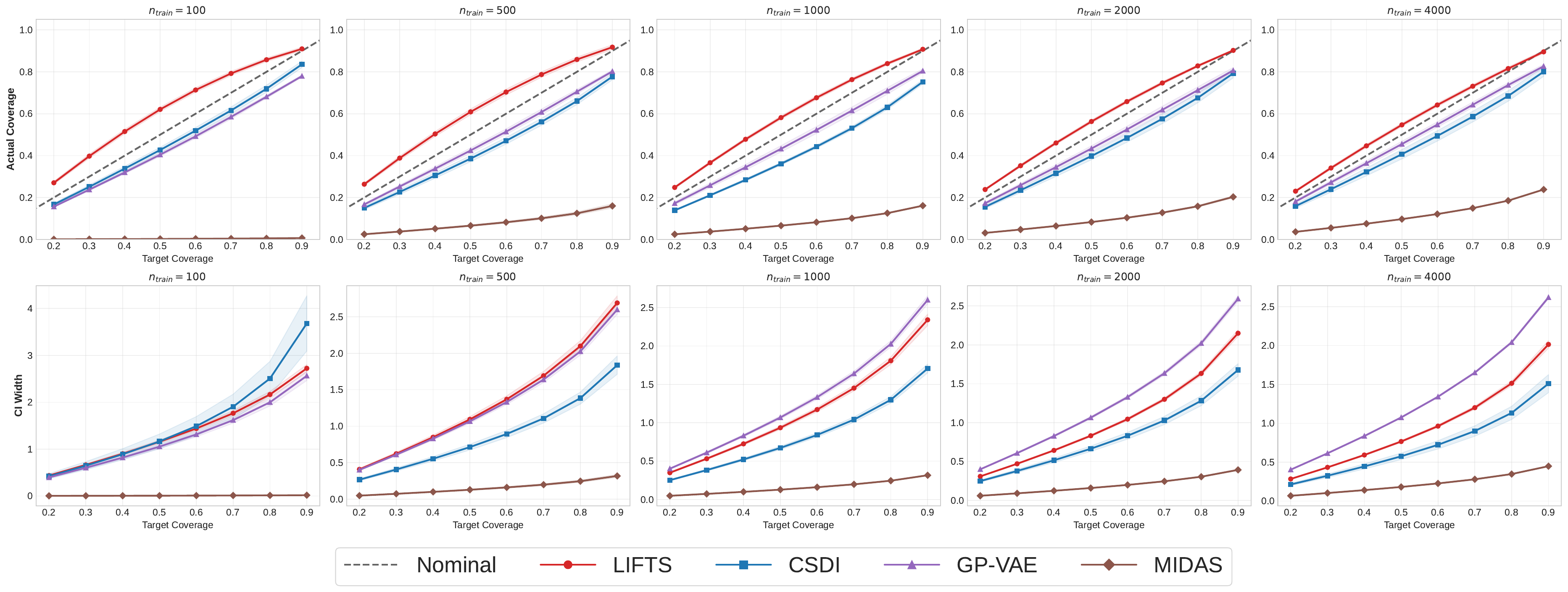}
    \caption{Empirical coverage and average prediction-interval width for test-set forecasting over \(t=37,\ldots,48\), across different training sample sizes.}
    \label{fig:coverage-forecasting}
\end{figure}

\begin{figure}
    \centering
    \includegraphics[width=1\linewidth]{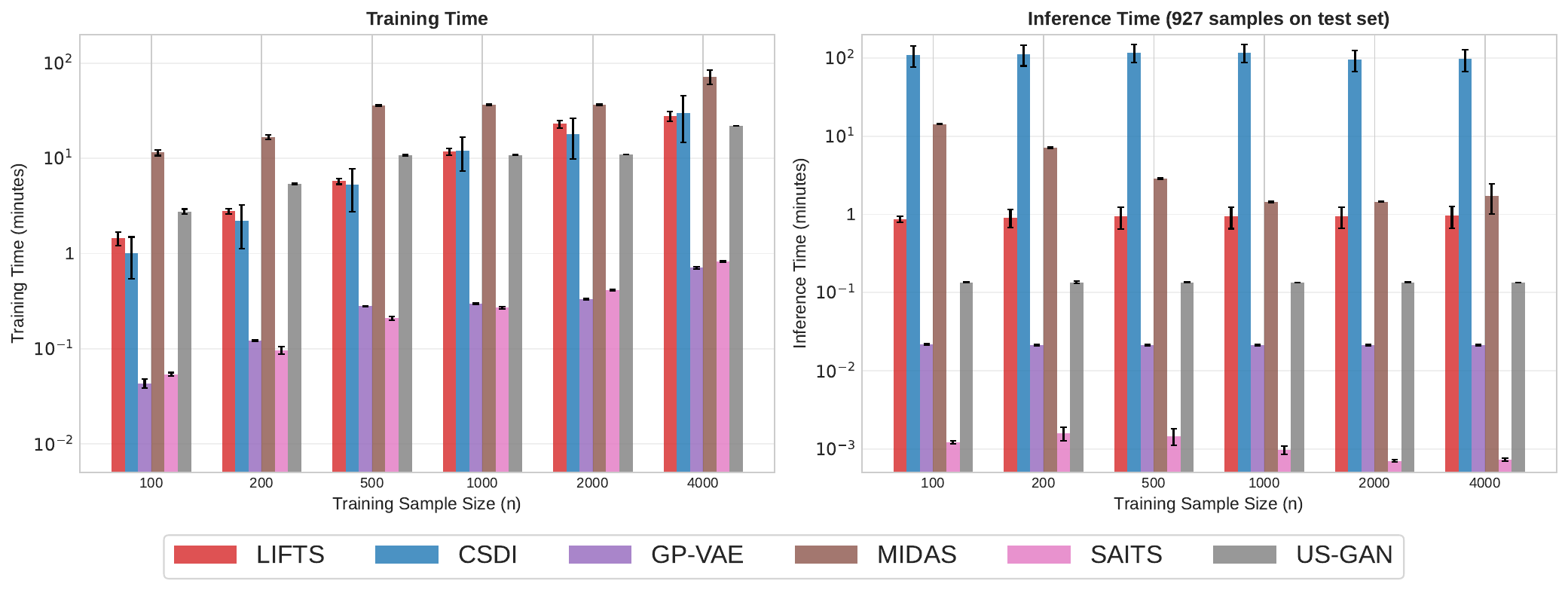}
    \caption{The training and inference time comparisons across all methods.}
    \label{fig:real-time}
\end{figure}

% %%%%%%%%%%%%%%%%%%%%%%%%%%%%%%%%%%%%%%%%%%%%%%%%%%%
%%%%%%%%%%%%%%%%%%%%%%%%%%%%%%%%%%%%%%%%%%%%%%%%%%%
%%%%%%%%%%%%%%%%%%%%%%%%%%%%%%%%%%%%%%%%%%%%%%%%%%%
%%%%%%%%%%%%%%%%%%%%%%%%%%%%%%%%%%%%%%%%%%%%%%%%%%%
%%%%%%%%%%%%%%%%%%%%%%%%%%%%%%%%%%%%%%%%%%%%%%%%%%%
%%%%%%%%%%%%%%%%%%%%%%%%%%%%%%%%%%%%%%%%%%%%%%%%%%%
%%%%%%%%%%%%%%%%%%%%%%%%%%%%%%%%%%%%%%%%%%%%%%%%%%%
% \newpage
\end{document}